\documentclass[manuscript,nonacm]{acmart} 

\AtBeginDocument{%
  \providecommand\BibTeX{{%
    \normalfont B\kern-0.5em{\scshape i\kern-0.25em b}\kern-0.8em\TeX}}}

\copyrightyear{2021}
\acmYear{2021}
\setcopyright{acmlicensed}\acmConference[PODC '21]{Proceedings of the 2021 ACM Symposium on Principles of Distributed Computing}{July 26--30, 2021}{Virtual Event, Italy}
\acmBooktitle{Proceedings of the 2021 ACM Symposium on Principles of Distributed Computing (PODC '21), July 26--30, 2021, Virtual Event, Italy}
\acmPrice{15.00}
\acmDOI{10.1145/3465084.3467938}
\acmISBN{978-1-4503-8548-0/21/07}

\newcommand{\rmv}[1]{}

\newcommand{\rmvOF}[1]{}

\newcommand{\MP}[1]{}

\newcommand{\toOmit}[1]{#1}

\usepackage[leqno]{amsmath}
\usepackage{enumitem}
\usepackage{pwmath}

\usepackage{aliascnt}
\usepackage{hyperref}

\newtheorem{theorem}{Theorem}

\newaliascnt{lem}{theorem}  
\newtheorem{lem}[lem]{Lemma}  
\aliascntresetthe{lem}

\newaliascnt{corol}{theorem}  
  
\aliascntresetthe{corol}

\newaliascnt{observation}{theorem}  
  
\aliascntresetthe{observation}

\usepackage{balance}

\usepackage{algorithm}
\usepackage[noend]{algpseudocode}
\usepackage{caption}

\algnewcommand\algorithmicforeach{\textbf{for each}}
\algdef{S}[FOR]{ForEach}[1]{\algorithmicforeach\ #1\ \algorithmicdo}

\algnewcommand\algorithmicto{\textbf{to}}
\algrenewtext{For}[3]{\algorithmicfor\ {#1}~$\gets$ {#2} \algorithmicto\ {#3} \algorithmicdo}

\algblockdefx[]{Part}{EndPart}%
  [1][]{\textit{#1}:}

\makeatletter
\def\BState{\State\hskip-\ALG@thistlm}
\makeatother



\begin{document}

\title{Tight Lower Bounds for the RMR Complexity
	of Recoverable Mutual Exclusion}


\author{David Yu Cheng Chan}
\email{david.chan1@ucalgary.ca}
\affiliation{%
	\institution{University of Calgary}
	\city{Calgary}
	\state{Alberta}
	\country{Canada}%
}

\author{Philipp Woelfel}
\email{woelfel@ucalgary.ca}
\affiliation{%
	\institution{University of Calgary}
	\city{Calgary}
	\state{Alberta}
	\country{Canada}%
}

\renewcommand{\shortauthors}{Chan and Woelfel}

\begin{abstract}
	We present a tight RMR complexity lower bound for 
		the recoverable mutual exclusion (RME) problem, 
		defined by Golab and Ramaraju \cite{GR2019a}.
	In particular, we show that any $n$-process RME algorithm using 
		only atomic read, write, fetch-and-store, fetch-and-increment, 
		and compare-and-swap operations, has an RMR complexity 
		of $\Omega(\log n/\log\log n)$ on the CC and DSM model.
	This lower bound covers all realistic synchronization primitives 
		that have been used in RME algorithms
		and matches the best upper bounds of algorithms
		employing swap objects (e.g., \cite{GH2017a,JJJ2019a,DM2020a}).

	Algorithms with better RMR complexity than that 
		have only been obtained by either 
		(i) assuming that all failures are system-wide \cite{GH2018a},
		(ii) employing fetch-and-add objects of size 
		$(\log n)^{\omega(1)}$ \cite{KM2020a}, 
		or (iii) using artificially defined synchronization primitives 
		that are not available in actual systems \cite{GH2017a,JJJ2018a}.
\end{abstract}




\maketitle


\section{Introduction}
\label{sec:Intro}
Recent research on the mutual exclusion problem has focused on \emph{recoverable} algorithms \cite{Ram2015a,JJ2017a,DM2020a,JJJ2018a,GH2018a,GH2017a,GR2019a,JJJ2019a,CW2020a,KM2020a}.
Here, a process may crash at any point during the execution, upon which its entire local state is being reset, including all its local program variables.
Shared memory, however, is not affected by process crashes.
When a process crashes, it starts a recovery routine that allows it to resume participation in the mutual exclusion protocol.
This variant of mutual exclusion has been motivated by recent advances in non-volatile memory architectures \cite{GR2019a}.

The standard complexity measure for mutual exclusion algorithms is RMR complexity.
The RMR complexity of conventional, non-recoverable $n$-process mutual exclusion is well understood: If only read-write registers and compare-and-swap objects are available, then a worst-case RMR complexity of $\Theta(\log n)$ is optimal  \cite{YA1995a,AHW2008a}.
Using other standard synchronization primitives, such as fetch-and-store (swap) or fetch-and-increment, the RMR complexity can be reduced to $O(1)$ \cite{MS1991a,Cra1993a,MLH1994a}.

Recoverable mutual exclusion is seemingly harder: Many algorithms use fetch-and-store objects, and the best ones achieve an RMR complexity of $\Theta(\log n/\log\log n)$ \cite{GH2017a,JJJ2019a}.
To study what it takes to achieve better RMR complexity than that, artificially defined synchronization primitives have been used that do not exist in actual hardware \cite{GH2017a,JJJ2018a}.
Katzan and Morrison \cite{KM2020a} observed that one can obtain an RMR complexity of $O(\log_w n)$ using $w$-bit fetch-and-add objects.
In particular, if $w=n^\varepsilon$ for some $\varepsilon>0$, then constant RMR complexity is possible.
But it is a standard (and reasonable) assumption that $w=O(\log n)$.
Even for poly-logarithmic values of $w$ their algorithm does not beat the best known upper bounds of $O(\log n/\log\log n)$.

It is therefore not surprising that it has been stated as an open problem (see e.g., \cite{KM2020a,CW2020a}), whether there are algorithms with better RMR complexity.
In this paper we provide an answer, for almost all standard synchronization primitives that have been used to solve RME, with the exception of fetch-and-add.
(For fetch-and-add objects that can store $w=(\log n)^{\omega(1)}$ bits, the lower bound does not apply due to \cite{KM2020a}, but for the  more realistic assumption $w=(\log n)^{O(1)}$ it remains open if there exist algorithms with $o(\log n/\log\log n)$ RMR complexity.)

\begin{theorem}
	\label{thm:MainTheorem}
  Any deadlock-free $n$-process RME algorithm, where all objects support only read, write,
  fetch-and-store, fetch-and-increment, and compare-and-swap operations,
  has RMR complexity $\Omega(\log n / \log \log n)$ in the CC and the DSM model.
\end{theorem}

This lower bound is tight: 
	It shows that the algorithms by Golab and Hendler \cite{GH2017a} (for the CC model) 
	and Jayanti, Jayanti, and Joshi \cite{JJJ2019a} 
	(for the CC and DSM model) are optimal.
Both algorithms use registers and fetch-and-store objects, 
	and \cite{GH2017a} also uses compare-and-swap.
The lower bound demonstrates for the first time, 
	that RME is strictly harder than non-recoverable mutual exclusion,
	because, as mentioned above, the latter has constant RMR complexity 
	if fetch-and-store or fetch-and-increment objects 
	and registers are available \cite{MS1991a,Cra1993a,MLH1994a}.
Chan and Woelfel \cite{CW2020a} gave an RME algorithm 
	with constant \emph{amortized} RMR complexity 
	using fetch-and-increment and compare-and-swap objects in addition to registers.
Thus, our result separates worst-case from amortized RMR complexity for RME.
Interestingly, for non-recoverable mutual exclusion, 
	worst-case and amortized RMR complexity is the same for any subset of primitives 
	(that includes at least read-write registers) to which our lower bound applies.
It is also worth pointing out that our lower bound 
	applies to all deadlock-free RME algorithms, 
	and does not rely on the critical-section re-entry property, 
	which is usually required for RME.

In addition, Golab and Hendler \cite{GH2018a} showed that 
	with the stricter system-wide failure model (all processes crash simultaneously), 
	there exist RME algorithms with $O(1)$ RMR complexity.
Thus our lower bound also demonstrates that RME is strictly harder 
	in the more general failure model that allows processes to crash independently  
	than in the system-wide failure model.

Recently, Dhoked and Mittal \cite{DM2020a} gave an algorithm that adapts to the number of process crashes (in the ``recent'' past), $F$.
In particular, they achieve an RMR complexity of $O(\min\{\sqrt{F},\log n/\log\log n\})$.
In our proof we construct an execution, in which one process incurs $\Omega(\log n/\log\log n)$ RMRs, even though each process attempts to enter the critical section at most once and crashes at most once.
This shows that the RMR complexity of mutual exclusion algorithms, such as the one by Dhoked and Mittal, can only adapt to the \emph{total number} of crashes---bounding the number of crashes per process does not not suffice to improve RMR complexity.

\section{Preliminaries}
We consider the standard asynchronous shared memory model, 
	where $n$ processes with unique IDs communicate 
	by executing atomic operations (called steps) 
	on shared base objects.

A mutual exclusion algorithm is a shared (implemented) object that supports methods \textsc{Enter()} and \textsc{Exit()}, and calls to these methods must alternate, starting with \textsc{Enter()}.
A process is in the \emph{critical section} when it has finished an \textsc{Enter()} call but not yet called \textsc{Exit()}.
It is in the \emph{remainder section}, if it is not in the critical section and has no pending \textsc{Enter()} or \textsc{Exit()} call.
Such algorithms satisfy at least two conditions: \emph{Mutual exclusion} requires that no two processes are in the critical section at the same time, and \emph{deadlock-freedom} requires that some process with a pending \textsc{Enter()} or \textsc{Exit()} call will eventually finish this call, provided that all processes that are not in the remainder section keep taking steps.

A \emph{recoverable} mutual exclusion (RME) algorithm provides an additional method, \textsc{Recover()}.
It is assumed that a process may crash at any point. 
(Formally, a process performs a crash step, 
	which does not alter any shared memory objects.)
If a process that is not in the remainder section crashes, all its local variables are reset to their initial values, and the process immediately begins executing method \textsc{Recover()}.
\emph{Deadlock-freedom} is only required if the number of crashes is finite.
(Otherwise a process might repeatedly crash in its critical section, preventing other processes from making progress.)
In addition to mutual exclusion and deadlock-freedom, RME algorithms must satisfy \emph{critical section re-entry}, which means that if a process crashes in the critical section, then it will reenter the critical section before any other process.
Our lower bound is independent of that property (deadlock-freedom and mutual exclusion are sufficient).

There are two common models used for the complexity analysis of mutual exclusion algorithms.
In the cache-coherent (CC) model, processes are equipped with caches, and whenever a process performs a read operation it obtains a cache-copy of the corresponding memory location.
Any non-read operation of that memory location invalidates all cache copies.
A process's operation incurs a \emph{remote memory reference} (RMR), if it is a read operation and the process has no valid cache copy, or if it is not a read operation.
In the distributed shared memory (DSM) model, the shared memory is partitioned into segments, one for each process.
An operation on a shared memory location incurs an RMR if and only if that memory location is not in the calling process's memory segment.
Almost all work on the efficiency of (recoverable and non-recoverable) mutual exclusion algorithms has considered RMR complexity in one of those two models.

A \emph{passage} of a process begins when it calls \textsc{Enter()} and ends when the process crashes or when it finishes its following \textsc{Exit()} call.
A \emph{super-passage} of a process begins when the process calls \textsc{Enter()} and when it completes its following \textsc{Exit()} call.
(Note that in the absence of process crashes, all passages and super-passages coincide.)
The RMR complexity of a mutual exclusion algorithm is the maximum number of RMRs a process may incur in any passage.

For our lower bound proof, base objects can store values from arbitrary (even uncountable) domains.
We assume that processes can perform atomically any of the following operations on a shared object with value $x$:
\begin{itemize}
  \item \texttt{read()}: returns $x$;
  \item \texttt{FAS($x'$)}: writes $x'$ and returns $x$;
  \item \texttt{CAS($y,x'$}): writes $x'$ and returns true, provided that $x=y$; otherwise it leaves the value unchanged and returns false.
  \item \texttt{FAI()}: writes $x+1$ and returns $x$, provided that $x$ is an integer. Otherwise it does not change the value of the object and returns $x$.
\end{itemize}
Note that a FAS operation is strictly stronger than a write operation (which does not return anything), so our lower bound proof does not consider write operations separately.




\newcommand{\sigmaOld}{\sigma_\textit{old}}
\newcommand{\sigmaNew}{\sigma_\textit{new}}
\newcommand{\sigmaFinal}{\sigma_\textit{round}}
\newcommand{\SMax}{S_\textit{max}}
\newcommand{\SOMax}{S^\textit{old}_\textit{max}}
\newcommand{\SSMax}{S^\textit{setup}_\textit{max}}
\newcommand{\SLMax}{S^\textit{low}_\textit{max}}

\newcommand{\sigmaSetupA}{\sigma_\textit{setupA}}
\newcommand{\sigmaSetupB}{\sigma_\textit{setupB}}

\newcommand{\sigmaLowA}{\sigma_\textit{lowA}}
\newcommand{\sigmaLowB}{\sigma_\textit{lowB}}
\newcommand{\sigmaLowC}{\sigma_\textit{lowB}}

\newcommand{\POP}{\textit{opt}_H}
\newcommand{\sigmaHighA}{\sigma_\textit{highA}}
\newcommand{\sigmaHighB}{\sigma_\textit{highB}}
\newcommand{\sigmaHighC}{\sigma_\textit{highC}}

\section{The RME Lower Bound Proof}
\label{sec:proof}
We consider an arbitrary algorithm that solves the RME problem
with $o(\log n)$ RMR complexity.

\paragraph{Assumptions:}
We make the following assumptions w.l.o.g.:
\begin{enumerate}[label=(A\arabic*)]
	\item In the critical section, each process performs
		operation(s) that incur at least one RMR.
		\label{RME-assumption:RMR-in-CS}

	\item Since the algorithm incurs
		$o(\log n)$ RMRs in every passage of every execution,
		we assume $n$ is sufficiently large such that
		every passage of every execution incurs
		no more than $\log n$ RMRs.
		\label{RME-assumption:logNRMRs}

	\item Each process begins at most one super-passage, i.e.,
		it leaves the remainder section at most once.
		(Note that this assumption makes our proof stronger,
		since a solution for the RME problem that
		allows multiple super-passages per process
		clearly also solves the RME problem
		in the scenario where each process
		can begin at most one super-passage.)
		\label{RME-assumption:OneSP}
\end{enumerate}

\paragraph{Definitions:}
We define the following:
\begin{itemize}
	\item Let $\mathcal{P}=\{1,\dots,n\}$ be the set of processes
		and $\mathcal{R}$ be the set of objects.

	\item Given any array $A[0..2^n-1]$ and any set $S \subseteq \mathcal{P}$,
		we use $A[S]$ to denote $A[\sum_{p \in S} 2^{p-1}]$.

	\item A schedule is a sequence over $\{p,\hat{p}:p\in\mathcal{P}\}$,
		where $p$ denotes a non-crash step by process $p$,
		and $\hat{p}$ denotes a crash-step by $p$.

	\item Let $C_0$ denote the initial configuration.

	\item For each schedule $\sigma$, a configuration $C$,
		and a register $R$, we define the following:
		\begin{itemize}
			\item $P(\sigma)$: the set of all processes that have steps in $\sigma$.
			

			\item $E(C,\sigma)$: the execution determined by $\sigma$
				starting in configuration $C$.
			\textbf{Note:} An execution is a sequence of \emph{events}, where each event corresponds to a step by some process and contains the following information: The process that is executing the step, the shared memory operation that process is executing, the object (register) on which the shared memory operation is executed, and whether the shared memory operation incurs an RMR.

			\item $\textit{val}_R(C,\sigma)$:
				the value of $R$ at the end of $E(C,\sigma)$.

			\item $\textit{state}_p(C,\sigma)$:
				the state of $p$ at the end of $E(C,\sigma)$.

			\item $\textit{last}_R(C,\sigma)$:
				the process that last performed an operation on $R$
				at the end of $E(C,\sigma)$; or $\bot$
				if no process has ever performed an operation on $R$.

			\item $F(C,\sigma)$: the set of processes
				that have finished their super-passage
				at the end of $E(C,\sigma)$.
		\end{itemize}
	We also define $E(\sigma)=E(C_0,\sigma)$,
		$\textit{val}_R(\sigma)=\textit{val}_R(C_0,\sigma)$,
		$\textit{state}_p(\sigma)=\textit{state}_p(C_0,\sigma)$,
		$\textit{last}_R(\sigma)=\textit{last}_R(C_0,\sigma)$,
		and $F(\sigma)=F(C_0,\sigma)$.
\end{itemize}


\subsection{Overview of the Proof}

\newcommand{\sigmaEnd}{\sigma_\textit{goal}}

Our goal is to show that the algorithm has $\Omega(\log n / \log \log n)$ RMR complexity even when each process begins at most one super-passage and crashes at most once.
Towards that end, we will construct a schedule $\sigmaEnd$
	such that during $E(\sigmaEnd)$:
\begin{itemize}
	\item Each process begins at most one super-passage, and crashes at most once.
	\item Some process never crashes and never enters the critical section,
		yet incurs $\Omega(\log n / \log \log n)$ RMRs.
\end{itemize}
On a very high level the construction follows the outline of Anderson and Kim's lower bound for non-recoverable mutual exclusion algorithms \cite{AK2002a}.
Their proof applies only to read-write registers.
In order to deal with stronger primitives, we have to crash processes at opportune points in time, so that they ``forget'' information they may have observed (e.g., as a result of FAI or FAS operations).

We begin with a simple observation:
	if multiple processes are 'actively'
	attempting to enter the critical section,
	then they cannot safely enter the critical section
	before discovering one another,
	lest they violate mutual exclusion. 
Thus throughout the proof,
	we will construct several closely related schedules in which
	we attempt to maximize both the number of these \emph{active} processes
	and the number of RMRs they incur without discovering one another.

More formally, let $\sigmaFinal[0..\infty][0..2^n-1]$
	be an initially empty table of schedules
	with an unbounded number of rows and $2^n$ columns.
Roughly speaking, for every non-negative integer $i$,
	the $i$-th row of the table
	will contain only schedules in which
	the active processes have incurred
	at least $i$ RMRs.
For every integer $s \in \{0,1,\ldots,2^n-1\}$,
	we associate the $s$-th column
	with the unique set $S \subseteq \mathcal{P}$
	of processes such that $s = \sum_{p \in S} 2^{p-1}$.
Then the $s$-th column will contain only schedules in which
	only the processes in $S$ can begin super-passages.

Filling the first row of the table is simple:
	in the empty schedule,
	every active process has incurred $0$ RMRs,
	and the set of processes that have begun super-passages
	is $\varnothing$, a subset of every possible set of processes.
Thus we set every cell of $\sigmaFinal[0][0..2^n-1]$
	to contain the empty schedule.

The proof then proceeds in rounds,
	where in each round $i \geq 1$,
	we fill in some cells of the $i$-th row with schedules
	derived by appending more steps to the schedules in the $(i-1)$-th row.
Since we only want schedules in which
	the active processes do not discover one another,
	many of the cells in each row will be left with the value $\bot$,
	indicating that we did not find a schedule
	matching the required criteria.
Thus as we go down through the rows of the table,
	the number of cells in each row
	that we fill with schedules decreases.

As such, the goal of each round is to limit
	this decrease, such that $\Omega(\log n / \log \log n)$ rounds
	complete before the number of schedules becomes too few to continue.
After which, every schedule in the final round
	would have active processes that incur $\Omega(\log n / \log \log n)$ RMRs
	without entering the critical section (or crashing). 

To facilitate this, we maintain a number of invariants
	on every row of schedules that we construct.
Roughly speaking, these invariants are:
\begin{enumerate}
	\item There is a \emph{maximal} schedule which has
		the maximal number of active processes,
		and all other schedules are 'sub'-schedules
		that correspond to every possible subset of
		the active processes in the maximal schedule.
	This invariant ensures that if the maximal schedule cannot be extended
		without allowing some active processes to discover one another,
		then a sub-schedule can be extended and made into
		the new maximal schedule for the next round.

	\item The state of every process is the same
		in every schedule it is part of.
	This invariant ensures that the active processes
		have not discovered one another,
		since they have the same state in a schedule
		where there are no other active processes.

	\item For each register, its value in each schedule
		depends only on whether the schedule contains
		the process that last accessed it in the maximal schedule.
	This invariant ensures that register values are sufficiently similar
		across different schedules that it becomes difficult
		for the active processes to later distinguish between different schedules.

	\item In every schedule, each process crashes at most once,
		and every process that is within a super-passage
		has not yet entered the critical section.
	The invariant makes the proof significantly simpler,
		since it prevents interactions between the active processes
		and the inactive processes that have already entered the critical section
		but not yet completed their super-passage.

	\item In the DSM model, the registers that are owned by active processes
		have not been accessed by any other active process.
	This invariant also simplifies the proof,
		since it prevents non-RMR-incurring steps from allowing
		an active process to discover another active process, and thus
		allows the proof to focus on the RMR-incurring steps.
		
	\item In the CC model, for each process $p$, the set of registers
		that $p$ has valid cache copies of is identical over
		all schedules that contain $p$.
	This invariant ensures that in the CC model,
		for each process $p$, the number of RMRs incurred by $p$
		is the same in every schedule it is part of.

	\item In the $i$-th row, every active process in every schedule
		has incurred at least $i$ RMRs.
\end{enumerate}

It is easy to see that these invariants all hold for row $0$.
Furthermore, for every non-negative integer $i$,
	let $n_i$ be the number of active processes 
	in the maximal schedule of row $i$.
Then the first invariant asserts that row $i$ has $2^{n_i}$ schedules.
Moreover, to show that $\Omega(\log n / \log \log n)$ rounds can be completed,
	it suffices to show that for every integer $i \geq 1$,
	$n_i > n_{i-1} / O(\log^{O(1)} n)$.

Each round of the proof is divided into two phases:
	a setup phase in which non-RMR-incurring steps are appended to the schedules
	until every active process in every schedule is poised to incur an RMR,
	and a contention phase,
	in which RMR-incurring steps are appended in specific orders
	that limit the fraction of active processes discovered.

In the setup phase, multiple non-RMR-incurring step(s) are appended
	for each active process
	until they are poised to incur an RMR.
By the above invariants,
	the non-RMR-incurring steps appended
	for each process are the same
	in every schedule that contains the process.
This is because each process begins with the same state
	in every schedule that contains the process,
	and then:
\begin{itemize}
	\item In the DSM model, its non-RMR-incurring steps only access its own registers,
		which have never been accessed by any other active process,
		and thus these steps intuitively provide no new information that would cause
		the process to change its next steps.

	\item In the CC model, its non-RMR-incurring steps would be reads
		on registers that it has valid cache copies of in
		every schedule that contains it.
	Then, since the process already has valid cache copies of these registers,
		they intuitively provide no new information that would cause
		the process to change its next steps.
\end{itemize}

In the contention phase, our construction method differs depending
	on the relative number of registers that the active processes are poised to access
	(in the maximal schedule).

In a low contention scenario, the active processes are poised to access
	a relatively large number of registers,
	and so on average, each register has relatively few processes poised to access it.
In this case, we construct a graph with nodes representing the active processes,
	and edges that intuitively indicate processes that could discover one another:
	either because they are poised to access the same register,
	or they are poised to access a register that is owned
	or previously accessed by another active process.
Since the contention is relatively low,
	the resulting graph is relatively sparse,
	and thus contains a relatively large independent set.
We now discard any schedule that contains 
	any process outside of this independent set,
	so that the remaining schedules only contain active processes that
	would not discover one another with their next step.
The remaining schedules then have a single step appended for each active process,
	and then are used to fill the next row of $\sigmaFinal[0..\infty][0..2^n-1]$.
It is straightforward to show that the above invariants 
	still hold for this new row of schedules. 
Furthermore, due to the relative largeness of the independent set, 
	it is also straightforward to show that 
	$n_i > n_{i-1} / O(\log^{O(1)} n)$
	for every row $i \geq 1$ constructed in a low contention scenario.

In a high contention scenario, the active processes are poised to access
	a relatively small number of registers,
	and so on average, each register has relatively 
	many processes poised to access it.
In this scenario, it is often inevitable that 
	some active processes are discovered by the others,
	and these active processes must then be inactivated 
	by allowing them to enter the critical section,
	then complete their super-passage.
To further complicate matters,
	each such process could discover $o(\log n)$ other active processes
	before completing its super-passage, and 
	these discovered processes must then be removed
	(schedules that contain such processes are discarded).
Nevertheless, we can limit the number of discovered processes as follows.

First, we determine the plurality type of operation that
	the plurality of active processes are poised to perform.
Every active process that is not poised to perform
	this plurality type of operation is then removed
	(schedules that contain such processes are discarded).
Note that since there are only a constant number of operation types,
	a constant fraction of the active processes must remain.

We then divide these remaining active processes 
	into groups of $O(\log^{O(1)} n)$ processes,
	such that within each group, all processes are poised to access the same register
	(we remove any active processes that cannot be placed into such groups, 
	the number of which is at most a constant fraction 
	of the remaining active processes).
Then within each group, we select two active processes (preferentially
	those applying operations that would change the value of the register)
	that we call the alpha processes.
Every schedule that does not contain all of the alpha processes is then discarded.

Intuitively, these alpha processes are the processes that will be discovered:
	they will be crashed, and then
	allowed to run until they complete their super-passages.
Any other active processes that they discover along the way will be removed
	(schedules that contain such processes are discarded). 

Now recall that each process incurs at most $o(\log n)$ RMRs 
	during its super-passage,
	whereas each group contains $O(\log^{O(1)} n)$ processes.
Thus we can ensure that a constant fraction of the groups
	still contain active processes that have not been discovered.
Then one undiscovered active process in each such group,
	called the beta process, is allowed to take an RMR-incurring step
	that is intuitively hidden by the steps
	of the alpha processes in its group as follows:
\begin{itemize}
	\item If the plurality type of operation is read, then since reads
		do not change the value of a register, 
		the beta process can safely perform its read
		between the reads of the alpha processes 
		without affecting the value of the register.

	\item If the plurality type of operation is fetch-and-store,
		then since fetch-and-stores completely overwrite the value of a register,
		the beta process can safely perform its fetch-and-store
		between the fetch-and-stores of the alpha processes without
		affecting the final value of the register.

	\item If the plurality type of operation is fetch-and-increment,
		we can replace the fetch-and-increment of the first alpha process
		with a fetch-and-increment by the beta process,
		and the final value of the register will remain the same.

	\item If the plurality type of operation is compare-and-swap,
		then there must be an ordering of the alpha and beta processes
		such that the beta process fails its compare-and-swap operation,
		and so has no effect on the final value of the register.
\end{itemize}
Roughly speaking, this allows the beta processes to take
	their RMR-incurring steps without changing the value of any register,
	and although the alpha processes can immediately discover the beta processes,
	they will immediately crash and forget the beta processes,
	and will never discover the beta processes again. 
Thus we can construct schedules that allow
	the beta processes to remain active
	without being discovered;
	all other remaining active processes are removed. 
Since a constant fraction of the groups of $O(\log^{O(1)} n)$ processes
	yield an undiscovered beta process for the new maximal schedule, 
	we can also prove that $n_i > n_{i-1} / O(\log^{O(1)} n)$
	for every row $i \geq 1$ constructed in a high contention scenario.

Thus, regardless of whether each round $i \geq 1$ has
	a low contention phase or a high contention phase,
	$n_i > n_{i-1} / O(\log^{O(1)} n)$.
By the first invariant,  the number of schedules 
	in each row $i$ is $2^{n_i}$.
So $\Omega(\log n / \log \log n)$ rounds
	complete before the number of schedules becomes too few to continue.
After which, the schedules in the final round
	would have active processes that incur 
	$\Omega(\log n / \log \log n)$ RMRs
	without entering the critical section (or crashing).
Consequently, the algorithm has $\Omega(\log n / \log \log n)$ RMR complexity.

\subsection{Proof Details}


\paragraph{Invariants:}
To prove the main theorem, we will iteratively
	construct arrays of schedules.

Let $i$ be a non-negative integer,
	and $A[0..2^n-1]$ be an array such that
	each array entry contains either a schedule or $\bot$.
Then we say that $A[0..2^n-1]$ is
	\emph{$i$-compliant} 
	if it satisfies the following invariants:
\begin{enumerate}[label=(I\arabic*)]
	\item For every set $S \subseteq \mathcal{P}$,
		if $A[S] \neq \bot$, then $P(A[S]) \subseteq S$.
		\label{RME-invar:PsubsetS}
		(Note that this implies $F(A[S]) \subseteq S$.)

	\item There is a unique set $\SMax \subseteq \mathcal{P}$ such that
		for every set $S \subseteq \mathcal{P}$, $A[S] \neq \bot$
		if and only if $F(A[\SMax]) \subseteq S \subseteq \SMax$.
		\label{RME-invar:uniqueSMax}

	\item For every process $p \in \SMax$
		and every set $S \subseteq \mathcal{P}$ that contains $p$,
		if $A[S] \neq \bot$, then
		$\textit{state}_p(A[S]) = \textit{state}_p(A[\SMax])$.
		\label{RME-invar:subsetStatesSame2}

	\item $F(A[S]) = F(A[\SMax])$ for every
		set $S \subseteq \mathcal{P}$ with $A[S] \neq \bot$.
		\label{RME-invar:FSame}
		(Note that this invariant immediately follows
			from Invariants~\ref{RME-invar:PsubsetS},
			\ref{RME-invar:uniqueSMax}, and \ref{RME-invar:subsetStatesSame2}.)

	\item For every register $R \in \mathcal{R}$,
		there is a value $y_R$ such that
		for every set $S \subseteq \mathcal{P}$,
		if $A[S] \neq \bot$, then:
		\begin{displaymath}
			\textit{val}_R\bparen{A[S]}=
			\begin{cases}
					\textit{val}_R(A[\SMax]) &
						\text{if $\textit{last}_R(A[\SMax]) \in S$} \\
					y_R & \text{otherwise}
			\end{cases}
		\end{displaymath}
		Note that it is possible that
			$y_R=\textit{val}_R(A[\SMax])$.
		Furthermore, if $\textit{last}_R(A[\SMax]) = \bot \not\in S$,
			then $\textit{val}_R(A[S])= y_R$
			for every set $S \subseteq \mathcal{P}$
			with $A[S] \neq \bot$.
		\label{RME-invar:subsetRegsSimilar}

	\item For every set $S \subseteq \mathcal{P}$ with $A[S] \neq \bot$,
		during $E(A[S])$, each process crashes at most once,
		and each process that is not in $F(A[S])$ never crashes.
		\label{RME-invar:crashes}

	\item For every set $S \subseteq \mathcal{P}$ with $A[S] \neq \bot$,
		each process that is not in $F(A[S])$
		does not enter the critical section during $E(A[S])$.
		\label{RME-invar:CS}

	\item In the DSM model, for every process
		$p \in \SMax \setminus F(A[\SMax])$,
		every register $R \in \mathcal{R}$ owned by $p$,
		and every set $S \subseteq \mathcal{P}$ with $A[S] \neq \bot$,
		$R$ can only be accessed by $p$ during $E(A[S])$.
		\label{RME-invar:DSMSelfAccessOnly}
	(Or equivalently, In the DSM model,
		for every set $S \subseteq \mathcal{P}$ such that $A[S] \neq \bot$,
		during $E(A[S])$, each register $R \in \mathcal{R}$
		can only be accessed by its owner
		if the owner of $R$ is in $\SMax \setminus F(A[\SMax])$.)

	\item In the CC model, for every process
		$p \in \SMax \setminus F(A[\SMax])$,
		there is a set $\mathcal{R}_p$ of registers such that
		for every set $S \subseteq \mathcal{P}$ that contains $p$,
		if $A[S] \neq \bot$, then
		the set of registers that $p$ has valid cache copies of
		at the end of $E(A[S])$ is exactly $\mathcal{R}_p$.
		\label{RME-invar:ValidCCSame}
	(Or equivalently, for every set $S \subseteq \mathcal{P}$
		such that $A[S] \neq \bot$,
		and every process $p \in S \cap (\SMax \setminus F(A[\SMax]))$,
		the set of registers that $p$ has valid cache copies of
		at the end of $E(A[S])$ is exactly the same as at the end of $E(A[\SMax])$.)

	\item For every set $S \subseteq \mathcal{P}$
		and every process $p \in S \setminus F(A[S])$,
		if $A[S] \neq \bot$, then
		$p$ incurs at least $i$ RMRs
		during $E(A[S])$.
		\label{RME-invar:iRMRs}
\end{enumerate}
Let $i$ be a non-negative integer,
	and $A[0..2^n-1]$ be an array that is $i$-compliant.
Then we denote by $\SMax(A[0..2^n-1])$
	the unique set of Invariant~\ref{RME-invar:uniqueSMax}.

Let $\sigmaFinal[0..\infty,0..2^n-1]$ be a table
	with all entries initially containing $\bot$.
Our goal is to fill in the table such that
	for every non-negative integer $i$,
	either $\sigmaFinal[i,0..2^n-1]$ is $i$-compliant,
	or $i \in \Omega(\log n / \log \log n)$.

Let $d$ be a sufficiently large constant
	and $k = \log^d n$.

\paragraph{Base Case:}

For every set $S \subseteq \mathcal{P}$,
	let $\sigmaFinal[0,S]$ be set to the empty schedule
	(so every entry of $\sigmaFinal[0,0..2^n-1]$ is the empty schedule).
Clearly, the array $\sigmaFinal[0,0..2^n-1]$ is $0$-compliant
	with $\SMax(\sigmaFinal[0,0..2^n-1]) = \mathcal{P}$
	and has $2^n$ non-$\bot$ entries.

We now iterate through $i = 1,2,\ldots$ as follows:

\paragraph{$i$-th Iteration (Termination Phase):}

If $\sigmaFinal[i-1,0..2^n-1]$ is not $(i-1)$-compliant
	or has less than $2^{(k^3)}$ non-$\bot$ entries,
	terminate.


\paragraph{$i$-th Iteration (Setup Phase):} 

For every set $S \subseteq \mathcal{P}$, 
	let $\sigmaOld[S] = \sigmaFinal[i-1][S]$.
So the array $\sigmaOld[0..2^n-1]$ is $(i-1)$-compliant.

Thus by Invariant~\ref{RME-invar:uniqueSMax},
	there is a unique set $\SOMax \subseteq \mathcal{P}$ such that
	$\sigmaOld[\SOMax] \neq \bot$ and
	for every set $S \subseteq \mathcal{P}$, $\sigmaOld[S] \neq \bot$ 
	if and only if $F(\sigmaOld[\SOMax]) \subseteq S \subseteq \SOMax$.
So by definition, $\SMax(\sigmaOld[0..2^n-1]) = \SOMax$.
Then for every process $p \in \SOMax \setminus F(\sigmaOld[\SOMax])$,
	let $S_p = \{p\} \cup F(\sigmaOld[\SOMax])$.
Note that $F(\sigmaOld[\SOMax]) \subseteq S_p \subseteq \SOMax$,
	so $\sigmaOld[S_p] \neq \bot$.

Now for every process $p \in \SOMax \setminus F(\sigmaOld[\SOMax])$,
	let $\sigma_p$ be a schedule consisting only 
	of the maximum non-negative number 
	of non-crash steps of $p$ such that any RMRs incurred by 
	$p$ in $E(\sigmaOld[S_p] \circ \sigma_p)$
	were also incurred in $E(\sigmaOld[S_p])$.
Then let $C_p$ be the configuration at the end of $E(\sigmaOld[S_p])$.
So by definition, $p$ does not incur any RMRs in $E(C_p,\sigma_p)$.
Furthermore, if $\sigma_p$ is finite, 
	then an RMR would be incurred by $p$ 
	at the end of $E(\sigmaOld[S_p] \circ \sigma_p \circ p)$.

\begin{lem}
	\label{thm:SetupNotInfinite}
	For every process $p \in \SOMax \setminus F(\sigmaOld[\SOMax])$,
		$\sigma_p$ is finite.
\end{lem}

\begin{proof}
	Let $p$ be any process in $\SOMax \setminus F(\sigmaOld[\SOMax])$.
	Suppose, for contradiction, that $\sigma_p$ is infinite.
	
	Recall that $S_p = \{p\} \cup F(\sigmaOld[\SOMax])$, 
		so $F(\sigmaOld[\SOMax]) \subseteq S_p \subseteq \SOMax$.
	Since $\sigmaOld[0..2^n-1]$ is $i-1$-compliant
		with $\SMax(\sigmaOld[0..2^n-1]) = \SOMax$,
		by Invariants \ref{RME-invar:uniqueSMax} 
		and \ref{RME-invar:FSame}, 
		$\sigmaOld[S_p] \neq \bot$ and
		$F(\sigmaOld[S_p]) = F(\sigmaOld[\SOMax])$.
	So $p \in S_p \setminus F(\sigmaOld[S_p])$.
	Thus by Invariant \ref{RME-invar:CS}, 
		$p$ does not enter the critical section during $E(\sigmaOld[S_p])$. 

	By definition, $C_p$ is the configuration at the end of $E(\sigmaOld[S_p])$
		and $p$ does not incur any RMRs in $E(C_p,\sigma_p)$. 
	So by Assumption \ref{RME-assumption:RMR-in-CS},
		$p$ does not enter the critical section during the 
		infinite execution $E(\sigmaOld[S_p] \circ \sigma_p)$.
	
	Now recall that $S_p = \{p\} \cup F(\sigmaOld[\SOMax])$, 
		so every process $q \neq p$
		is either in $F(\sigmaOld[S_p])$ or not in $S_p$.
	By Invariant~\ref{RME-invar:PsubsetS},
		every process $q \not\in S_p$ 
		takes no steps in $E(\sigmaOld[S_p])$.
	Since $\sigma_p$ contains only steps of $p$
		and $p \in S_p$, every process $q \not\in S_p$
		also takes no steps in 
		$E(\sigmaOld[S_p] \circ \sigma_p)$.
	So in $E(\sigmaOld[S_p] \circ \sigma_p)$,
		$p$ takes infinitely many steps without entering the critical section
		while every process $q \neq p$ is in the remainder section
		--- contradicting that $E(\sigmaOld[S_p] \circ \sigma_p)$
		is an execution of an algorithm that solves the RME problem.
\end{proof}

Since $\sigmaOld[0..2^n-1]$ is $(i-1)$-compliant
	with $\SMax(\sigmaOld[0..2^n-1]) = \SOMax$,
	by Invariant~\ref{RME-invar:uniqueSMax},
	for every set $S \subseteq \mathcal{P}$, $\sigmaOld[S] \neq \bot$ 
	if and only if $F(\sigmaOld[\SOMax]) \subseteq S \subseteq \SOMax$. 
So for each set $S \subseteq \mathcal{P}$ such that $\sigmaOld[S] \neq \bot$,
	let $p_{1,S}$ be the process with the smallest ID in 
	$S \setminus F(\sigmaOld[\SOMax])$,
	$p_{2,S}$ be the process with the second smallest ID,
	and so on.
Then let $C_S$ be the configuration at the end of $E(\sigmaOld[S])$,
	and let $\sigma_S = \sigma_{p_{1,S}} \circ \sigma_{p_{2,S}} \circ \ldots$.
Note that by \autoref{thm:SetupNotInfinite} and 
	the fact that the system has only $n$ processes, 
	$\sigma_S$ is a finite schedule. 

\begin{lem}
	\label{thm:SetupPreInvars}
	For every set $S \subseteq \mathcal{P}$ such that 
		$\sigmaOld[S] \neq \bot$:
	\begin{enumerate}[label=(S\arabic*)]
		\item No RMRs are incurred during $E(C_S,\sigma_S)$.
			\label{SPI:noRMRs}
		
		
		\item For each process $p \in S$,
			$\textit{state}_p(C_p,\sigma_p) =
			\textit{state}_p(C_S,\sigma_S)$.
			\label{SPI:statesSame} 
	
		\item For each process $p \in S \setminus F(\sigmaOld[\SOMax])$,
			$p$ incurs an RMR at the end of $E(C_S, \sigma_S \circ p)$. 
			\label{SPI:aboutToRMR}
		
		\item For each process $p \in S \setminus F(\sigmaOld[\SOMax])$,
			$p$ has not left the critical section at the end of $E(C_S,\sigma_S)$.
			\label{SPI:notLeftCS} 
			
		\item $F(\sigmaOld[S]) = F(\sigmaOld[S] \circ \sigma_S)$.
			\label{SPI:FUnchanged}
			
		\item In the DSM model, 
			each register $R \in \mathcal{R}$ can only be 
			accessed by its owner during $E(C_S,\sigma_S)$.
			\label{SPI:DSMSelfAccessOnly}
			
		\item In the DSM model,
			for each process $p \in S \setminus F(\sigmaOld[\SOMax])$
			and each register $R \in \mathcal{R}$ owned by $p$,
			$\textit{val}_R(C_p,\sigma_p) = \textit{val}_R(C_S,\sigma_S)$.
			\label{SPI:DSMRegsSame}
		
		\item In the CC model, 
			each register $R \in \mathcal{R}$ can only be read during $E(C_S,\sigma_S)$.
			\label{SPI:CCReadOnly}
			
		\item In the CC model, during $E(C_S,\sigma_S)$,
			each process $p$ can only read registers that 
			it already has valid cache copies of. 
			\label{SPI:CCNoNewCaches}
	\end{enumerate}
\end{lem}

\begin{proof}
	Let $S \subseteq \mathcal{P}$ be any set of processes
		such that $\sigmaOld[S] \neq \bot$.
	Since $\sigmaOld[0..2^n-1]$ is $(i-1)$-compliant
		with $\SMax(\sigmaOld[0..2^n-1]) = \SOMax$,
		by Invariant~\ref{RME-invar:uniqueSMax},
		for every set $S' \subseteq \mathcal{P}$, $\sigmaOld[S'] \neq \bot$ 
		if and only if $F(\sigmaOld[\SOMax]) \subseteq S' \subseteq \SOMax$. 
	Thus $F(\sigmaOld[\SOMax]) \subseteq S \subseteq \SOMax$.
	Then since $S \subseteq \SOMax$, by definition we have that
		for every process $p \in S \setminus F(\sigmaOld[\SOMax])$,
		$S_p = \{p\} \cup F(\sigmaOld[\SOMax])$.
	
	Since $\sigmaOld[0..2^n-1]$ is $(i-1)$-compliant,
		by Invariant~\ref{RME-invar:subsetStatesSame2},
		for every process $p \in S \setminus F(\sigmaOld[\SOMax])$
		and every process $q \in S_p$,
		$\textit{state}_q(\sigmaOld[S]) =
		\textit{state}_q(\sigmaOld[S_p])$.
	Furthermore, by Invariant~\ref{RME-invar:FSame},
		$F(\sigmaOld[\SOMax]) = F(\sigmaOld[S])$.
	So by Invariant~\ref{RME-invar:CS},
		every process $p \in S \setminus F(\sigmaOld[\SOMax])$
		has not entered the critical section during $E(\sigmaOld[S])$. 

	The proof now differs depending on the model:
	\begin{description}
		\item[\textbf{CC Model:}]
	
			In the CC model, any step that does not incur an RMR
				must be a read operation on a register that the invoking process
				already has a valid cache copy of.
			Thus, by definition,
				for every process $p \in S \setminus F(\sigmaOld[\SOMax])$,
				$p$ only performs read operations on registers that 
				it already has valid cache copies of
				during $E(C_p,\sigma_p)$.
			
			Since $\sigmaOld[0..2^n-1]$ is $(i-1)$-compliant,
				by Invariant~\ref{RME-invar:ValidCCSame},
				for every process $p \in S \setminus F(\sigmaOld[\SOMax])$,
				the set of registers that $p$ has valid cache copies of
				is in same in $C_p$ as in $C_S$.
			Furthermore, read operations clearly cannot invalidate
				any valid cache copies.
			Thus, by the definition of $\sigma_S$, observe that
				for every process $p \in S \setminus F(\sigmaOld[\SOMax])$,
				the operations performed by $p$ during $E(C_S,\sigma_S)$
				are the same as in during $E(C_p,\sigma_p)$, i.e.,
				$p$ only performs read operations on registers that 
				it already has valid cache copies of
				during $E(C_S,\sigma_S)$ (\ref{SPI:CCNoNewCaches}).
				
			This implies the following:
			\begin{itemize}
				\item Since RMRs are not incurred by any 
					read operation on a register that the invoking
					process already has a valid cache copy of,
					no RMRs are incurred during $E(C_S,\sigma_S)$
					(\ref{SPI:noRMRs}).
			
				\item Since only read operations are performed during
					$E(C_S,\sigma_S)$, each register can only be read
					during $E(C_S,\sigma_S)$ (\ref{SPI:CCReadOnly}).
				
				\item By definition, $\sigma_S$ contains only steps of
					processes in $S \setminus F(\sigmaOld[\SOMax])$.
				Thus for every process $p \in S$,
					observe that $\textit{state}_p(C_S,\sigma_S)
					= \textit{state}_p(C_p,\sigma_p)$ (\ref{SPI:statesSame}).
					
				Now recall that for every process $p \in S \setminus F(\sigmaOld[\SOMax])$,
					the set of registers that $p$ has valid cache copies of
					is in same in $C_p$ as in $C_S$.
				Then, since (i) $\textit{state}_p(C_S,\sigma_S)
					= \textit{state}_p(C_p,\sigma_p)$,
					(ii) the valid cache copies of $p$
					are the same in $C_p$ as in $C_S$,
					(iii) new cache copies cannot be created by reading registers 
					that valid cache copies already exist for,
					and (iv) $p$ incurs an RMR in $E(C_p \circ \sigma_p, p)$ 
					by the definition of $\sigma_p$,
					observe that $p$ also incurs an RMR
					at the end of $E(C_S, \sigma_S \circ p)$ (\ref{SPI:aboutToRMR}). 
				
				\item Since every process $p \in S \setminus F(\sigmaOld[\SOMax])$
					has not entered the critical section during $E(\sigmaOld[S])$
					and no RMRs are incurred during $E(C_S,\sigma_S)$,
					by Assumption~\ref{RME-assumption:RMR-in-CS},
					every process $p \in S \setminus F(\sigmaOld[\SOMax])$
					has not left the critical section at the end of $E(C_S,\sigma_S)$
					(\ref{SPI:notLeftCS}). 
					
				Then, since $\sigma_S$ contains only non-crash steps,
					no process completes during $E(C_S,\sigma_S)$.
				Thus $F(\sigmaOld[S]) = F(\sigmaOld[S] \circ \sigma_S)$ (\ref{SPI:FUnchanged}).
			\end{itemize} 
				
		\item[\textbf{DSM Model:}]
		
			In the DSM model, any step that does not incur an RMR
				must be an operation on a register owned by the invoking process.
			Thus, by definition,
				for every process $p \in S \setminus F(\sigmaOld[\SOMax])$,
				$p$ only performs operations on its own registers
				during $E(C_p,\sigma_p)$.
				
			Since $\sigmaOld[0..2^n-1]$ is $(i-1)$-compliant
				with $\SMax(\sigmaOld[0..2^n-1]) = \SOMax$,
				by Invariant~\ref{RME-invar:DSMSelfAccessOnly}, 
				for every process $p \in \SOMax \setminus F(\sigmaOld[\SOMax])$,
				every register $R \in \mathcal{R}$ owned by $p$,
				and every set $S' \subseteq \mathcal{P}$ with $\sigmaOld[S'] \neq \bot$, 
				$R$ can only be accessed by $p$ during $E(\sigmaOld[S'])$,
				so $\textit{last}_R(\sigmaOld[S'])$ is either $p$ or $\bot$.
			Thus by Invariant~\ref{RME-invar:subsetRegsSimilar},
				if $\textit{init}_R$ is the initial value of $R$, then:
			\begin{displaymath}
				\textit{val}_R\bparen{\sigmaOld[S']}=
				\begin{cases}
						\textit{val}_R(\sigmaOld[\SOMax]) & 
							\text{if $p \in S'$} \\
						\textit{init}_R & \text{otherwise}
				\end{cases}
			\end{displaymath}
			
			So for every process $p \in S \setminus F(\sigmaOld[\SOMax])$
				and every register $R \in \mathcal{R}$ owned by $p$,
				$\textit{val}_R(\sigmaOld[S]) =
				\textit{val}_R(\sigmaOld[S_p])$.
			Furthermore, operations on registers not owned by $p$
				clearly cannot change the value of registers owned by $p$.
			Consequently, by the definition of $\sigma_S$, observe that
				for every process $p \in S \setminus F(\sigmaOld[\SOMax])$,
				the operations performed by $p$ during $E(C_S,\sigma_S)$
				are the same as in during $E(C_p,\sigma_p)$, i.e.,
				$p$ only performs operations on its own registers 
				during $E(C_S,\sigma_S)$.	
			
			This implies the following:
			\begin{itemize}
				\item Since RMRs are not incurred by any 
					operation on a register owned by the invoking process,
					no RMRs are incurred during $E(C_S,\sigma_S)$
					(\ref{SPI:noRMRs}).
			
				\item Since each process only accesses 
					its own registers during $E(C_S,\sigma_S)$,
					for each register $R \in \mathcal{R}$,
					$R$ can only be accessed by its owner 
					during $E(C_S,\sigma_S)$ (\ref{SPI:DSMSelfAccessOnly}).
				
				Furthermore, since the operations performed by each process
					$p \in S \setminus F(\sigmaOld[\SOMax])$ during $E(C_S,\sigma_S)$
					are the same as in during $E(C_p,\sigma_p)$,
					$\textit{val}_R(C_p,\sigma_p) = \textit{val}_R(C_S,\sigma_S)$
					(\ref{SPI:DSMRegsSame}).
				
				\item By definition, $\sigma_S$ contains only steps of
					processes in $S \setminus F(\sigmaOld[\SOMax])$.
				Thus for every process $p \in S$,
					observe that $\textit{state}_p(C_S,\sigma_S)
					= \textit{state}_p(C_p,\sigma_p)$ (\ref{SPI:statesSame}).
					
				Furthermore, by the definition of $\sigma_p$,
					$p$ incurs an RMR at the end of $E(C_p, \sigma_p \circ p)$,
					i.e., $p$ is poised to access a register that it does not own
					at the end of $E(C_p, \sigma_p)$.
				Thus, since $\textit{state}_p(C_S,\sigma_S)
					= \textit{state}_p(C_p,\sigma_p)$, 
					$p$ is also poised to access a register that it does not own
					at the end of $E(C_S, \sigma_S)$, and so
					$p$ also incurs an RMR
					at the end of $E(C_S, \sigma_S \circ p)$ (\ref{SPI:aboutToRMR}). 
				
				\item Since every process $p \in S \setminus F(\sigmaOld[\SOMax])$
					has not entered the critical section during $E(\sigmaOld[S])$
					and no RMRs are incurred during $E(C_S,\sigma_S)$,
					by Assumption~\ref{RME-assumption:RMR-in-CS},
					every process $p \in S \setminus F(\sigmaOld[\SOMax])$
					has not left the critical section at the end of $E(C_S,\sigma_S)$
					(\ref{SPI:notLeftCS}). 
					
				Then, since every process $p \in S \setminus F(\sigmaOld[\SOMax])$
					has not left the critical section at the end of $E(C_S,\sigma_S)$,
					every process $p \in S \setminus F(\sigmaOld[\SOMax])$
					has not completed its super-passage at the end of $E(C_S,\sigma_S)$.
				Thus $F(\sigmaOld[S]) = F(\sigmaOld[S] \circ \sigma_S)$ (\ref{SPI:FUnchanged}).
			\end{itemize}   
	\end{description} 
\end{proof}

We now construct a new array $\sigmaSetupA[0..2^n-1]$ such that
	for every set $S \subseteq \mathcal{P}$,
	$\sigmaSetupA[S] = \bot$ if $\sigmaOld[S] = \bot$;
	otherwise $\sigmaSetupA[S] = \sigmaOld[S] \circ \sigma_S$.

\begin{lem}	
	\label{thm:SetupAInvars}
	Except for Invariant~\ref{RME-invar:CS},
		$\sigmaSetupA[0..2^n-1]$ is $(i-1)$-compliant
		with $\SMax(\sigmaSetupA[0..2^n-1]) = \SOMax$. 
\end{lem}

\begin{proof}
	For every set $S \subseteq \mathcal{P}$,
		if $\sigmaSetupA[S] \neq \bot$,
		then by construction, $\sigmaSetupA[S] = \sigmaOld[S] \circ \sigma_S$.
	Since $\sigmaOld[0..2^n-1]$ is $(i-1)$-compliant, 
		by Invariant~\ref{RME-invar:PsubsetS},
		$P(\sigmaOld[S]) \subseteq S$.
	By the definition of $\sigma_S$, $\sigma_S$ contains only steps of processes in $S$.
	Thus $P(\sigmaSetupA[S]) \subseteq S$ (Invariant~\ref{RME-invar:PsubsetS}).
	
	Since $\sigmaOld[0..2^n-1]$ is $(i-1)$-compliant
		with $\SMax(\sigmaOld[0..2^n-1]) = \SOMax$,
		by Invariant~\ref{RME-invar:uniqueSMax},
		for every set $S \subseteq \mathcal{P}$,
		$\sigmaOld[S] \neq \bot$ if and only if
		$F(\sigmaOld[\SOMax]) \subseteq S \subseteq \SOMax$.
	By construction, for every set $S \subseteq \mathcal{P}$,
		$\sigmaSetupA[S] = \bot$ if and only if
		$\sigmaOld[S] = \bot$. 
	Furthermore, by \autoref{thm:SetupPreInvars} (\ref{SPI:FUnchanged}),
		$F(\sigmaOld[\SOMax]) = F(\sigmaSetupA[\SOMax])$.
	Thus for every set $S \subseteq \mathcal{P}$,
		$\sigmaSetupA[S] \neq \bot$ if and only if
		$F(\sigmaSetupA[\SOMax]) \subseteq S \subseteq \SOMax$
		(Invariant~\ref{RME-invar:uniqueSMax}).
	
	By \autoref{thm:SetupPreInvars} (\ref{SPI:statesSame}),
		for every set $S \subseteq \mathcal{P}$ such that 
		$\sigmaOld[S] \neq \bot$,
		and every process $p \in S$,
		$\textit{state}_p(\sigmaSetupA[S_p]) = \textit{state}_p(C_p,\sigma_p)
		\textit{state}_p(C_S,\sigma_S) = \textit{state}_p(\sigmaSetupA[S])$.
	By construction, for every set $S \subseteq \mathcal{P}$,
		$\sigmaSetupA[S] = \bot$ if and only if $\sigmaOld[S] = \bot$.
	Furthermore, we have already proven that for every set $S \subseteq \mathcal{P}$,
		$\sigmaSetupA[S] \neq \bot$ if and only if
		$F(\sigmaSetupA[\SOMax]) \subseteq S \subseteq \SOMax$.
	Thus observe that for every process $p \in \SOMax$
		and every set $S \subseteq \mathcal{P}$ that contains $p$,
		if $\sigmaSetupA[S] \neq \bot$, then 
		$\textit{state}_p(\sigmaSetupA[S]) = \textit{state}_p(\sigmaSetupA[\SOMax])$
		(Invariant~\ref{RME-invar:subsetStatesSame2}).
	
	Since we have already proven that Invariants~\ref{RME-invar:PsubsetS},
		\ref{RME-invar:uniqueSMax}, and \ref{RME-invar:subsetStatesSame2} hold
		for $\sigmaSetupA[0..2^n-1]$,
		it immediately follows that Invariant~\ref{RME-invar:FSame} also holds.
		
	Since $\sigmaOld[0..2^n-1]$ is $(i-1)$-compliant,
		Invariant~\ref{RME-invar:crashes} holds for $\sigmaOld[0..2^n-1]$.
	For every set $S \subseteq \mathcal{P}$ such that $\sigmaOld[S] \neq \bot$,
		$\sigma_S$ contains no crash steps.
	Thus Invariant~\ref{RME-invar:crashes} also holds for $\sigmaSetupA[0..2^n-1]$.
	
	Since $\sigmaOld[0..2^n-1]$ is $(i-1)$-compliant,
		Invariant~\ref{RME-invar:DSMSelfAccessOnly} holds for $\sigmaOld[0..2^n-1]$.
	By \autoref{thm:SetupPreInvars} (\ref{SPI:DSMSelfAccessOnly}),
		for every set $S \subseteq \mathcal{P}$ such that 
		$\sigmaOld[S] \neq \bot$,
		each register $R \in \mathcal{R}$ can only be 
		accessed by its owner during $E(C_S,\sigma_S)$.
	Thus Invariant~\ref{RME-invar:DSMSelfAccessOnly} also holds
		for $\sigmaSetupA[0..2^n-1]$.
	
	Since $\sigmaOld[0..2^n-1]$ is $(i-1)$-compliant,
		Invariant~\ref{RME-invar:ValidCCSame} holds for $\sigmaOld[0..2^n-1]$.
	In the CC model, by \autoref{thm:SetupPreInvars} (\ref{SPI:CCReadOnly}),
		for every set $S \subseteq \mathcal{P}$ such that 
		$\sigmaOld[S] \neq \bot$,
		each register $R \in \mathcal{R}$ can only be 
		read during $E(C_S,\sigma_S)$.
	Thus no valid cache copy can be invalidated during $E(C_S,\sigma_S)$.
	Furthermore, by \autoref{thm:SetupPreInvars} (\ref{SPI:CCReadOnly}),
		during $E(C_S,\sigma_S)$, each process $p$
		can only read registers that it already has valid cache copies of.
	Thus no new cache copies can be created during $E(C_S,\sigma_S)$.
	Consequently, Invariant~\ref{RME-invar:ValidCCSame} 
		also holds for $\sigmaSetupA[0..2^n-1]$.
	
	Since $\sigmaOld[0..2^n-1]$ is $(i-1)$-compliant,
		Invariant~\ref{RME-invar:iRMRs} holds for $\sigmaOld[0..2^n-1]$.
	Then, since no steps are removed in the construction of the schedules
		for $\sigmaSetupA[0..2^n-1]$, clearly
		Invariant~\ref{RME-invar:iRMRs} also holds for $\sigmaSetupA[0..2^n-1]$.
		
	We will now prove that Invariant~\ref{RME-invar:subsetRegsSimilar}
		holds for $\sigmaSetupA[0..2^n-1]$ as follows.
	Let $R \in \mathcal{R}$ be any register.
	Our goal is to show that there exists a value $y_R$ such that
		for every set $S \subseteq \mathcal{P}$,
		if $\sigmaSetupA[S] \neq \bot$, then:
	\begin{displaymath}
		\textit{val}_R\bparen{\sigmaSetupA[S]}=
		\begin{cases}
				\textit{val}_R(\sigmaSetupA[\SOMax]) & 
					\text{if $\textit{last}_R(\sigmaSetupA[\SOMax]) \in S$} \\
				y_R & \text{otherwise}
		\end{cases}
	\end{displaymath}
	
	Note that since $\sigmaOld[0..2^n-1]$ is $(i-1)$-compliant
		with $\SMax(\sigmaOld[0..2^n-1]) = \SOMax$,
		by Invariant~\ref{RME-invar:subsetRegsSimilar},
		there is a value $y_R$ such that
		for every set $S \subseteq \mathcal{P}$,
		if $\sigmaOld[S] \neq \bot$, then:
	\begin{displaymath}
		\textit{val}_R\bparen{\sigmaOld[S]}=
		\begin{cases}
				\textit{val}_R(\sigmaOld[\SOMax]) & 
					\text{if $\textit{last}_R(\sigmaOld[\SOMax]) \in S$} \\
				y_R & \text{otherwise}
		\end{cases}
	\end{displaymath}
	
	First, suppose that $R$ is not accessed during $E(C_S,\sigma_S)$ 
		for every set $S \subseteq \mathcal{P}$ such that $\sigmaOld[S] \neq \bot$. 
	Then since $\sigmaOld[S] = \bot$ if and only if $\sigmaSetupA[S] = \bot$,
		and $R$ is not accessed during $E(C_{\SOMax},\sigma_{\SOMax})$,
		$\textit{last}_R(\sigmaOld[\SOMax]) = \textit{last}_R(\sigmaSetupA[\SOMax])$.
	Thus as we wanted, for every set $S \subseteq \mathcal{P}$,
		if $\sigmaSetupA[S] \neq \bot$, then:
	\begin{displaymath}
		\textit{val}_R\bparen{\sigmaSetupA[S]}=
		\begin{cases}
				\textit{val}_R(\sigmaSetupA[\SOMax]) & 
					\text{if $\textit{last}_R(\sigmaSetupA[\SOMax]) \in S$} \\
				y_R & \text{otherwise}
		\end{cases}
	\end{displaymath}
			
	So suppose instead that there exists a set $S' \subseteq \mathcal{P}$ 
		such that $\sigmaOld[S'] \neq \bot$,
		and a process $p \in S' \setminus F(\sigmaOld[\SOMax])$ 
		that accesses $R$ during $E(C_{S'},\sigma_{S'})$.
	The proof now differs depending on the model.
	
	In the DSM model, by \autoref{thm:SetupPreInvars} (\ref{SPI:DSMSelfAccessOnly}),
		$p$ must be the owner of $R$.
	Then since $\sigmaOld[0..2^n-1]$ is $(i-1)$-compliant,
		by Invariant~\ref{RME-invar:DSMSelfAccessOnly},
		either $\textit{last}_R(\sigmaOld[\SOMax]) = p$
		or $\textit{val}_R(\sigmaOld[\SOMax]) = y_R$.
	Furthermore, by \autoref{thm:SetupPreInvars} 
		(\ref{SPI:DSMSelfAccessOnly} and \ref{SPI:DSMRegsSame}),
		for every set $S \subseteq \mathcal{P}$ such that
		$\sigmaOld[S] \neq \bot$,
		if $p \in S$, then
		$\textit{val}_R(C_S,\sigma_S) = \textit{val}_R(C_p,\sigma_p)$;
		otherwise $\textit{val}_R(C_S,\sigma_S) = y_R$.
	Thus as we wanted, for every set $S \subseteq \mathcal{P}$,
		if $\sigmaSetupA[S] \neq \bot$, then:
	\begin{displaymath}
		\textit{val}_R\bparen{\sigmaSetupA[S]} =
		\begin{cases}
				\textit{val}_R(\sigmaSetupA[\SOMax]) & 
					\text{if $\textit{last}_R(\sigmaSetupA[\SOMax]) \in S$} \\
				y_R & \text{otherwise}
		\end{cases}
	\end{displaymath}
	
	Finally, in the CC model, by \autoref{thm:SetupPreInvars} (\ref{SPI:CCNoNewCaches}),
		$p$ already has a valid cache copy of $R$ in $C_{S'}$.
	So since $\sigmaOld[0..2^n-1]$ is $(i-1)$-compliant,
		by Invariant~\ref{RME-invar:ValidCCSame},
		for every set $S \subseteq \mathcal{P}$
		such that $\sigmaOld[S] \neq \bot$ and $p \in S$,
		$\textit{val}_R(\sigmaOld[S]) = \textit{val}_R(\sigmaOld[S'])$.
	Thus either $\textit{last}_R(\sigmaOld[\SOMax]) = p$
		or $\textit{val}_R(\sigmaOld[\SOMax]) = y_R$.
	Note that if any process other than $p$ also accesses $R$ during $E(C_S,\sigma_S)$,
		then $\textit{val}_R(\sigmaOld[\SOMax]) = y_R$.
	Furthermore, by \autoref{thm:SetupPreInvars} (\ref{SPI:CCReadOnly}),
		for every set $S \subseteq \mathcal{P}$
		such that $\sigmaOld[S] \neq \bot$,
		$\textit{val}_R(\sigmaOld[S]) = \textit{val}_R(\sigmaSetupA[S])$.
	Thus observe that as we wanted, for every set $S \subseteq \mathcal{P}$,
		if $\sigmaSetupA[S] \neq \bot$, then:
	\begin{displaymath}
		\textit{val}_R\bparen{\sigmaSetupA[S]} =
		\begin{cases}
				\textit{val}_R(\sigmaSetupA[\SOMax]) & 
					\text{if $\textit{last}_R(\sigmaSetupA[\SOMax]) \in S$} \\
				y_R & \text{otherwise}
		\end{cases}
	\end{displaymath}
	Consequently we have proven that Invariant~\ref{RME-invar:subsetRegsSimilar}
		holds for $\sigmaSetupA[0..2^n-1]$.
\end{proof}

We now construct another array $\sigmaSetupB[0..2^n-1]$ 
	with the goal of satisfying Invariant~\ref{RME-invar:CS} as follows.
If no process is within the critical section at the end of $E(\sigmaSetupA[\SOMax])$,
	we simply construct $\sigmaSetupB[0..2^n-1]$ such that
	$\sigmaSetupB[0..2^n-1] = \sigmaSetupA[0..2^n-1]$.
Furthermore, we define $\SSMax = \SOMax$.
  
Otherwise, to avoid violating mutual exclusion,
	there must be exactly one process 
	$p \in \SOMax \setminus F(\sigmaSetupA[\SOMax])$ such that 
	at the end of $E(\sigmaSetupA[\SOMax])$,
	$p$ is within the critical section.
Then note that by \autoref{thm:SetupAInvars} and Invariant~\ref{RME-invar:subsetStatesSame2},
	for every set $S \subseteq \mathcal{P}$ such that $\sigmaSetupA[S] \neq \bot$,
	if $p \in S$ then $p$ is also within the critical section at the end of $E(\sigmaSetupA[S])$;
	otherwise no process is within the critical section at the end of $E(\sigmaSetupA[S])$.
Thus we construct $\sigmaSetupB[0..2^n-1]$ such that
	for every set $S \subseteq \mathcal{P}$,
	if $p \in S$, then $\sigmaSetupB[S] = \bot$;
	otherwise $\sigmaSetupB[S] = \sigmaSetupA[S]$.
Furthermore, we define $\SSMax = \SOMax \setminus \{p\}$.
	
\begin{lem}
	\label{thm:SetupBInvars}
	This new array $\sigmaSetupB[0..2^n-1]$ is $(i-1)$-compliant
		with $\SMax(\sigmaSetupB[0..2^n-1]) = \SSMax$.
\end{lem}
 
\toOmit{

	\begin{proof}
		By \autoref{thm:SetupAInvars},
			except for Invariant~\ref{RME-invar:CS},
			$\sigmaSetupA[0..2^n-1]$ is $(i-1)$-compliant
			with $\SMax(\sigmaSetupB[0..2^n-1]) = \SOMax$. 
			
		If no process is within the critical section at the end of $E(\sigmaSetupA[\SOMax])$,
			then $\sigmaSetupB[0..2^n-1] = \sigmaSetupA[0..2^n-1]$.
		Thus by \autoref{thm:SetupAInvars}
			and \autoref{thm:SetupPreInvars} (\ref{SPI:notLeftCS}),
			Invariant~\ref{RME-invar:CS} also holds for $\sigmaSetupA[0..2^n-1]$,
			and so it follows that $\sigmaSetupB[0..2^n-1] = \sigmaSetupA[0..2^n-1]$ 
			is $(i-1)$-compliant with $\SMax(\sigmaSetupB[0..2^n-1]) = \SSMax = \SOMax$.
			
		Otherwise, there is exactly one process 
			$p \in \SOMax \setminus F(\sigmaSetupA[\SOMax])$ such that  
			$p$ is within the critical section at the end of $E(\sigmaSetupA[\SOMax])$,
			and for every set $S \subseteq \mathcal{P}$,
			if $p \in S$, then $\sigmaSetupB[S] = \bot$;
			otherwise $\sigmaSetupB[S] = \sigmaSetupA[S]$.
		By \autoref{thm:SetupAInvars},
			except for Invariant~\ref{RME-invar:CS},
			$\sigmaSetupA[0..2^n-1]$ is $(i-1)$-compliant
			with $\SMax(\sigmaSetupA[0..2^n-1]) = \SOMax$. 
		Thus observe that by the construction of $\sigmaSetupB[0..2^n-1]$,
			Invariants~\ref{RME-invar:PsubsetS}, \ref{RME-invar:uniqueSMax},
			\ref{RME-invar:subsetStatesSame2}, \ref{RME-invar:FSame},
			\ref{RME-invar:subsetRegsSimilar}, \ref{RME-invar:crashes},
			\ref{RME-invar:DSMSelfAccessOnly}, \ref{RME-invar:ValidCCSame},
			\ref{RME-invar:iRMRs} must all also hold for $\sigmaSetupB[0..2^n-1]$
			with $\SMax(\sigmaSetupB[0..2^n-1]) = \SSMax
			= \SOMax \setminus \{p\}$.
		
		By the construction of $\sigmaSetupB[0..2^n-1]$,
			for every set $S \subseteq \mathcal{P}$,
			if $p \in S$, then $\sigmaSetupB[S] = \bot$.
		Thus for every set $S \subseteq \mathcal{P}$,
			if $\sigmaSetupB[S] \neq \bot$,
			then $p \not\in S$.
		Consequently, for every set $S \subseteq \mathcal{P}$ 
			such that $\sigmaSetupB[S] \neq \bot$,
			no process is within the critical section at the end of 
			$E(\sigmaSetupA[S]) = E(\sigmaSetupB[S])$.
		Therefore by \autoref{thm:SetupPreInvars} (\ref{SPI:notLeftCS}), 
			Invariant~\ref{RME-invar:CS} holds for $\sigmaSetupB[0..2^n-1]$,
			and thus $\sigmaSetupB[0..2^n-1]$ is $(i-1)$-compliant
			with $\SMax(\sigmaSetupB[0..2^n-1]) = \SSMax$.
	\end{proof}
}

\paragraph{$i$-th Iteration (Decision Phase):}

Then for each register $R \in \mathcal{R}$,
	let $B_R$ be the set of processes poised to access $R$
	at the end of $E(\sigmaSetupB[\SSMax])$.
(So $B_R \cap F(\sigmaSetupB[\SSMax]) = \emptyset$.)
Further, let
\begin{displaymath}
	H=\bigcup_{R\in\mathcal{R}\atop |B_{R}|\geq k} B_R.
\end{displaymath}
and
\begin{displaymath}
	L=\SSMax \setminus \bparen{H \cup F(\sigmaSetupB[\SSMax])}.
\end{displaymath}


\paragraph{$i$-th Iteration (Low Contention Phase if $|L| \geq |H|$):} 

We begin by constructing an undirected graph 
	where the processes in $L$ are the nodes,
	and for every pair of nodes $p$ and $q$,
	we connect an edge between $p$ and $q$
	if and only if at least one of the following is true 
	at the end of $E(\sigmaSetupB[\SSMax])$:
\begin{itemize}
	\item $p$ and $q$ are poised to access the same register.
	
	\item $p$ is poised to access a register owned by $q$, 
		or vice versa. 
		
	\item $p$ is poised to access a register
		that $q$ has previously performed an operation on, 
		or vice versa.
\end{itemize} 

Thus in this graph:
\begin{itemize}
	\item Since $|B_R| < k$ for every register $R \in \mathcal{R}$
		that processes in $L$ are poised to access,
		there are at most $k|L|$ edges representing 
		processes that are poised to access the same register.
		
	\item Since every register $R \in \mathcal{R}$ 
		is owned by at most one process,
		there are at most $|L|$ edges representing
		processes that are poised to access a register 
		owned by some process in $L$
		(one edge for each process, 
		connecting it to the owner of the register 
		it is poised to access).	
		
	\item By Assumption~\ref{RME-assumption:logNRMRs},
		every passage incurs at most $\log n$ RMRs.
	By \autoref{thm:SetupBInvars},
		$\sigmaSetupB[0..2^n-1]$ is $(i-1)$-compliant.
	Recall that $L \cap F(\sigmaSetupB[\SSMax]) = \varnothing$.
	So by Invariant~\ref{RME-invar:crashes},
		each process in $L$ has never crashed,
		and so has started at most one passage.
	Thus each process in $L$ has performed operations 
		on at most $\log n$ registers 
		that are not owned by itself.
	(Note that registers owned by itself are excluded because 
		their associated edges would have already been added in the previous step.)
	Then since $|B_R| < k$ for every register $R \in \mathcal{R}$
		that processes in $L$ are poised to access,
		there are at most $k |L| \log n$ edges representing 
		processes that are poised to access a register 
		that some process in $L$ 
		has previously performed an operation on
		($k$ edges from each of the $\log n$ registers
		previously accessed by each process in $L$).
\end{itemize} 
Hence, the total number of edges is at most 
	$k|L| + |L| + k |L| \log n
	< 3 k |L| \log n$.
	
Let $I$ be the maximum independent set of the graph.	
	
\begin{lem}
	\label{thm:Isize}
	$|I| \geq |L|/(7 k \log n)$.
\end{lem}

\begin{proof}
	Since there are at most $3 k |L| \log n$ edges,
		the average degree of the graph is at most 
		$6 k \log n$.
	The lemma immediately follows by Turan's Theorem.
\end{proof}	

Let $S_I = F(\sigmaSetupB[\SSMax]) \cup I$.
Note that since $I \subseteq L$,
	$F(\sigmaSetupB[\SSMax]) \subseteq S_I \subseteq \SSMax$.
We now construct a new array $\sigmaLowA[0..2^n-1]$ such that
	for every set $S \subseteq \mathcal{P}$,
	if $S \not\subseteq S_I$, then $\sigmaLowA[S] = \bot$;
	otherwise $\sigmaLowA[S] = \sigmaSetupB[S]$.

\begin{lem}
	\label{thm:LowAInvars}
	This new array $\sigmaLowA[0..2^n-1]$ is $(i-1)$-compliant
		with $\SMax(\sigmaLowA[0..2^n-1]) = S_I$.
\end{lem} 

\toOmit{

	\begin{proof}
		By \autoref{thm:SetupBInvars},
			$\sigmaSetupB[0..2^n-1]$ is $(i-1)$-compliant
			with $\SMax(\sigmaSetupB[0..2^n-1]) = \SSMax$.
		By construction, $\sigmaLowA[0..2^n-1]$ is simply 
			a modification of $\sigmaSetupB[0..2^n-1]$
			where every set $S \subseteq \mathcal{P}$ 
			that contains any process in $\SSMax \setminus S_I$
			has had $\sigmaLowA[S]$ set to $\bot$,
			where $F(\sigmaSetupB[\SSMax]) \subseteq S_I \subseteq \SSMax$.
		It suffices to observe that every invariant still holds
			with $\SMax(\sigmaLowA[0..2^n-1]) = S_I$,
			and thus $\sigmaLowA[0..2^n-1]$ is $(i-1)$-compliant
			with $\SMax(\sigmaLowA[0..2^n-1]) = S_I$.
	\end{proof}
}

Now for each set $S \subseteq \mathcal{P}$ such that
	$\sigmaLowA[S] \neq \bot$,
	let $C'_S$ be the configuration at the end of $E(\sigmaLowA[S])$, and
	let $\sigma'_S$ be the schedule consisting of exactly one non-crash step
	by each process in $S \setminus F(\sigmaLowA[S_I])$
	in order from the process with the smallest ID
	to the process with the largest ID.
Note that $\sigma'_S$ is a finite schedule 
	since there are only $n$ processes in the system.
Also note that by \autoref{thm:LowAInvars},
	$\sigmaLowA[0..2^n-1]$ is $(i-1)$-compliant
	with $\SMax(\sigmaLowA[0..2^n-1]) = S_I$,
	so $\sigmaLowA[S] \neq \bot$ if and only if
	$F(\sigmaLowA[S_I]) \subseteq S \subseteq S_I$.
Thus $\sigma'_S$ contains exactly one non-crash step of each process in $I \cap S$
	and no other steps.

\begin{lem}
	\label{thm:LowPreInvars}
	For every set $S \subseteq \mathcal{P}$ such that 
		$\sigmaLowA[S] \neq \bot$:
	\begin{enumerate}[label=(L\arabic*)]
		\item For each register $R \in \mathcal{R}$,
			$R$ is accessed by at most one process $p \in I \cap S$ 
			during $E(C'_S,\sigma'_S)$ and no other processes.
			\label{LPI:IOneAccess}
			
		\item For each process $p \in I$,
			if $p$ accesses a register $R$ during $E(C'_S,\sigma'_S)$,
			then the owner of $R$ is not in $I \setminus \{p\}$.
			\label{LPI:DSMOwnerIsGone}
		
		\item For each process $p \in I$,
			if $p$ accesses a register $R$ during $E(C'_S,\sigma'_S)$,
			then $R$ has never been accessed by any process in $I \setminus \{p\}$ 
			during $E(\sigmaLowA[S])$.
			\label{LPI:CCPrevNotInIminusP}  
		
		\item For each process $p \in I$,
			during $E(C'_S,\sigma'_S)$,
			$p$ cannot invalidate any cache copy of any process in $I \setminus \{p\}$.
			\label{LPI:CCNoInvalidateIminusP}  
			
		\item For every process $p \in I$,
			if $p$ accesses a register $R$ during $E(C'_{S_I},\sigma'_{S_I})$,
			then there is a value $y_R$ such that for every set $S' \subseteq \mathcal{P}$,
			if $\sigmaLowA[S'] \neq \bot$, then:
		\begin{displaymath}
			\textit{val}_R\bparen{\sigmaLowA[S']}=
			\begin{cases}
					\textit{val}_R(\sigmaLowA[S_I]) & 
						\text{if $p \in S'$} \\
					y_R & \text{otherwise}
			\end{cases}
		\end{displaymath}
		\label{LPI:IRegsWereSimilar}
		Note that this implies that for each register $R \in \mathcal{R}$,
			if $R$ is accessed during $E(C'_S,\sigma'_S)$ then
			$\textit{val}_R(\sigmaLowA[S_I]) = 
			\textit{val}_R(\sigmaLowA[S])$. 
			
		\item For each register $R \in \mathcal{R}$,
			if $R$ is accessed during $E(C'_S,\sigma'_S)$ then
			$\textit{val}_R(C'_{S_I},\sigma'_{S_I}) = 
			\textit{val}_R(C'_S,\sigma'_S)$.
			\label{LPI:IRegsSimilar}
			
		\item For each process $p \in S$,
			$\textit{state}_p(C'_{S_I},\sigma'_{S_I}) =
			\textit{state}_p(C'_S,\sigma'_S)$.
			\label{LPI:statesSame}  
		
		\item Each process in $I \cap S$
			incurs exactly one RMR
			during $E(C'_S,\sigma'_S)$.
			\label{LPI:OneRMREach} 
		
		\item For each process $p \in S \setminus F(\sigmaLowA[S_I])$,
			$p$ has not left the critical section during $E(\sigmaLowA[S] \circ \sigma'_S)$.
			\label{LPI:notLeftCS} 
			
		\item $F(\sigmaLowA[S]) = F(\sigmaLowA[S] \circ \sigma'_S)$.
			\label{LPI:FUnchanged}   
	\end{enumerate}
\end{lem}

\begin{proof}
	Let $S \subseteq \mathcal{P}$ be any set of processes
		such that $\sigmaLowA[S] \neq \bot$.
	By \autoref{thm:LowAInvars},
		$\sigmaLowA[0..2^n-1]$ is $(i-1)$-compliant
		with $\SMax(\sigmaLowA[0..2^n-1]) = S_I$.
	So by Invariant~\ref{RME-invar:uniqueSMax},
		for every set $S' \subseteq \mathcal{P}$, $\sigmaLowA[S'] \neq \bot$ 
		if and only if $F(\sigmaLowA[S_I]) \subseteq S' \subseteq S_I$. 
	Thus $F(\sigmaLowA[S_I]) \subseteq S \subseteq S_I$.
	
	Furthermore, by Invariants~\ref{RME-invar:subsetStatesSame2}
		and \ref{RME-invar:FSame},
		for every process $p \in S$,
		$\textit{state}_p(\sigmaLowA[S]) =
		\textit{state}_p(\sigmaLowA[S_I])$
		and $F(\sigmaLowA[S]) = F(\sigmaLowA[S_I])$.
	So by Invariant~\ref{RME-invar:CS},
		every process $p \in S \setminus F(\sigmaLowA[S_I])$
		has not entered the critical section during $E(\sigmaLowA[S])$.
		
	Now recall that $\sigma'_S$ contains exactly one non-crash step of 
		each process in $I \cap S$ and no other steps.
	Also recall that by construction,
		for every process $p \in S$,
		$\textit{state}_p(\sigmaLowA[S]) =
		\textit{state}_p(\sigmaLowA[S_I]) =
		\textit{state}_p(\sigmaSetupB[S_I]) = 
		\textit{state}_p(\sigmaSetupB[\SSMax])$.
	Then since $I$ is an independent set of the graph we constructed earlier:
	\begin{itemize}
		\item No pair of processes in $I$
			are poised to access the same register
			at the end of $E(\sigmaSetupB[\SSMax])$. 
		Since $\textit{state}_p(\sigmaLowA[S]) =
			\textit{state}_p(\sigmaSetupB[\SSMax])$ for every process $p \in S$,
			no pair of processes in $I$
			are poised to access the same register
			at the end of $E(\sigmaLowA[S])$. 
		Thus every process in $I \cap S$
			accesses a different register during $E(C'_S,\sigma'_S)$
			(\ref{LPI:IOneAccess}).
			
		\item No process $p \in I$ is poised
			to access a register owned by a different process $q \in I$
			at the end of $E(\sigmaSetupB[\SSMax])$.
		Since $\textit{state}_p(\sigmaLowA[S]) =
			\textit{state}_p(\sigmaSetupB[\SSMax])$ for every process $p \in S$,
			no process $p \in I$ is poised
			to access a register owned by a different process $q \in I$
			at the end of $E(\sigmaLowA[S])$.
		Thus for each process $p \in I$,
			if $p$ accesses a register $R$ during $E(C'_S,\sigma'_S)$,
			then the owner of $R$ is not in $I \setminus \{p\}$.
			(\ref{LPI:DSMOwnerIsGone}).
		
		\item No process $p \in I$ is poised
			to access a register that has previously 
			been accessed by a different process $q \in I$
			at the end of $E(\sigmaSetupB[\SSMax])$. 
		Since $\textit{state}_p(\sigmaLowA[S]) =
			\textit{state}_p(\sigmaSetupB[\SSMax])$ for every process $p \in S$,
			no process $p \in I$ is poised
			to access a register that has previously 
			been accessed by a different process $q \in I$
			at the end of $E(\sigmaLowA[S])$. 
		Thus for each process $p \in I$,
			if $p$ accesses a register $R$ during $E(C'_S,\sigma'_S)$,
			then $R$ has never been accessed by any process in $I \setminus \{p\}$ 
			during $E(\sigmaLowA[S])$.
			(\ref{LPI:CCPrevNotInIminusP}). 
		Therefore, for each process $p \in I$,
			if $p$ accesses a register $R$ during $E(C'_S,\sigma'_S)$,
			no process in $I \setminus \{p\}$
			makes a cache copy of $R$ during $E(\sigmaLowA[S])$.
		Consequently, for each process $p \in I$,
			during $E(C'_S,\sigma'_S)$,
			$p$ cannot invalidate any cache copy of any process in $I \setminus \{p\}$
			(\ref{LPI:CCNoInvalidateIminusP}).
			
		Furthermore, since $\textit{state}_p(\sigmaLowA[S_I]) =
			\textit{state}_p(\sigmaSetupB[\SSMax])$ for every process $p \in S$,
			no process $p \in I$ is poised
			to access a register that has previously 
			been accessed by a different process $q \in I$
			at the end of $E(\sigmaLowA[S_I])$.
		Thus for every process $p \in I$,
			if $p$ accesses a register $R$ during $E(C'_{S_I},\sigma'_{S_I})$,
			then $\textit{last}_R(\sigmaLowA[S_I]) \not\in I \setminus \{p\}$,
			so either $\textit{last}_R(\sigmaLowA[S_I]) = p$
			or $\textit{last}_R(\sigmaLowA[S_I]) \in F(\sigmaLowA[S_I]) 
			\cup (\mathcal{P} \setminus S_I)$. 
			
		By \autoref{thm:LowAInvars}, $\sigmaLowA[0..2^n-1]$ is $(i-1)$-compliant 
			with $\SMax(\sigmaLowA[0..2^n-1]) = S_I$.
		Thus by Invariant~\ref{RME-invar:subsetRegsSimilar},
			for every register $R \in \mathcal{R}$,
			there is a value $y_R$ such that for every set $S' \subseteq \mathcal{P}$,
			if $\sigmaLowA[S'] \neq \bot$, then:
		\begin{displaymath}
			\textit{val}_R\bparen{\sigmaLowA[S']}=
			\begin{cases}
					\textit{val}_R(\sigmaLowA[S_I]) & 
						\text{if $\textit{last}_R(\sigmaLowA[S_I]) \in S'$} \\
					y_R & \text{otherwise}
			\end{cases}
		\end{displaymath}
		
		We just showed that for every process $p \in I$,
			if $p$ accesses a register $R$ during $E(C'_{S_I},\sigma'_{S_I})$,
			then either $\textit{last}_R(\sigmaLowA[S_I]) = p$
			or $\textit{last}_R(\sigmaLowA[S_I]) \in F(\sigmaLowA[S_I]) 
			\cup (\mathcal{P} \setminus S_I)$. 
		Thus for every process $p \in I$,
			if $p$ accesses a register $R$ during $E(C'_{S_I},\sigma'_{S_I})$,
			then there is a value $y_R$ such that for every set $S' \subseteq \mathcal{P}$,
			if $\sigmaLowA[S'] \neq \bot$, then:
		\begin{displaymath}
			\textit{val}_R\bparen{\sigmaLowA[S']}=
			\begin{cases}
					\textit{val}_R(\sigmaLowA[S_I]) & 
						\text{if $p \in S'$} \\
					y_R & \text{otherwise}
			\end{cases}
		\end{displaymath}
		
		Consequently, observe that for each register $R \in \mathcal{R}$,
			if $R$ is accessed during $E(C'_S,\sigma'_S)$ then
			$\textit{val}_R(\sigmaLowA[S_I]) = 
			\textit{val}_R(\sigmaLowA[S])$
			(\ref{LPI:IRegsWereSimilar}).
	\end{itemize} 
	
	By \autoref{thm:LowAInvars},
		$\sigmaLowA[0..2^n-1]$ is $(i-1)$-compliant
		with $\SMax(\sigmaLowA[0..2^n-1]) = S_I$.
	So by Invariant~\ref{RME-invar:subsetStatesSame2},
		for each process $p \in S$,
		$\textit{state}_p(\sigmaLowA[S]) = \textit{state}_p(\sigmaLowA[S_I])$.
	
	Then, since we have proven (\ref{LPI:IOneAccess}) and (\ref{LPI:IRegsWereSimilar}),
		during both $E(C'_S,\sigma'_S)$ and $E(C'_{S_I},\sigma'_{S_I})$,
		for each process $p \in I \cap S$,
		$p$ performs the same operation on the same register,
		causing the same resulting state and response.
	This implies that:
	\begin{itemize}
		\item For each register $R \in \mathcal{R}$,
			if $R$ is accessed during $E(C'_S,\sigma'_S)$ then
			$\textit{val}_R(C'_{S_I},\sigma'_{S_I}) = 
			\textit{val}_R(C'_S,\sigma'_S)$
			(\ref{LPI:IRegsSimilar}).
			
		\item For each process $p \in S$,
			$\textit{state}_p(C'_{S_I},\sigma'_{S_I}) =
			\textit{state}_p(C'_S,\sigma'_S)$
			(\ref{LPI:statesSame}).
		(Recall that $\sigma'_S$ consists of exactly one non-crash step
			of each process in $I \cap S$ and no other steps,
			so the states of other processes do not change.)
	\end{itemize} 
	
	Since $\sigmaLowA[S] \neq \bot$, by construction,
		$\sigmaOld[S] \neq \bot$.
	By \autoref{thm:SetupPreInvars} (\ref{SPI:aboutToRMR}),
		for each process $p \in S \setminus F(\sigmaOld[\SOMax])$,
		$p$ incurs an RMR at the end of $E(\sigmaSetupA[S] \circ p)$.
	By construction and \autoref{thm:SetupAInvars} 
		(Invariant~\ref{RME-invar:FSame}),
		$\sigmaLowA[S] = \sigmaSetupB[S] = \sigmaSetupA[S]$,
		and $F(\sigmaLowA[S_I]) = F(\sigmaSetupA[S_I]) = F(\sigmaSetupA[\SOMax])
		= F(\sigmaOld[\SOMax])$.
	Thus for each process $p \in S \setminus F(\sigmaLowA[S_I])$,
		$p$ incurs an RMR at the end of $E(\sigmaLowA[S] \circ p)$.
	By definition, $S_I = I \cup F(\sigmaLowA[S_I])$,
		so since $S \subseteq S_I$, $S \setminus F(\sigmaLowA[S_I]) = I \cap S$.
	Finally, recall that $\sigma'_S$ contains exactly one 
		non-crash step of each process in $I \cap S$ and no other steps.
	Consequently, since we have already proven that 
		each process $I \cap S$ accesses a different register during $E(C'_S,\sigma'_S)$,
		observe that each process in $I \cap S$
		must incur exactly one RMR during $E(C'_S,\sigma'_S)$ (\ref{LPI:OneRMREach}).
	
	By \autoref{thm:LowAInvars},
		$\sigmaLowA[0..2^n-1]$ is $(i-1)$-compliant,
		so by Invariant~\ref{RME-invar:CS},
		each process that is not in $F(\sigmaLowA[S])$
		does not enter the critical section during $E(\sigmaLowA[S])$.
	Recall that $\sigma'_S$ consists of exactly one non-crash step
		of each process in $I \cap S = 
		S \setminus F(\sigmaLowA[S_I])$ and no other steps.
	Thus, with only one step, although a process could enter the critical section
		during $E(C'_S,\sigma'_S)$,
		it cannot have taken any steps within the critical section.
	Thus by Assumption~\ref{RME-assumption:RMR-in-CS}, 
		no process can leave the critical section during $E(C'_S,\sigma'_S)$.
	So for each process $p \in S \setminus F(\sigmaLowA[S_I])$,
		$p$ has not left the critical section during $E(\sigmaLowA[S] \circ \sigma'_S)$
		(\ref{LPI:notLeftCS}).
	Therefore, since no process in $S \setminus F(\sigmaLowA[S_I])$
		has left the critical section during $E(\sigmaLowA[S] \circ \sigma'_S)$,
		no process in $S \setminus F(\sigmaLowA[S_I])$
		has completed its super-passage during $E(\sigmaLowA[S] \circ \sigma'_S)$. 
	Thus $F(\sigmaLowA[S]) = F(\sigmaLowA[S] \circ \sigma'_S)$ (\ref{LPI:FUnchanged}). 
\end{proof}

We now construct a new array $\sigmaLowB[0..2^n-1]$ such that
	for every set $S \subseteq \mathcal{P}$,
	$\sigmaLowB[S] = \bot$ if $\sigmaLowA[S] = \bot$;
	otherwise $\sigmaLowB[S] = \sigmaLowA[S] \circ \sigma'_S$. 
	
\begin{lem}	
	\label{thm:LowBInvars}
	Except for Invariant~\ref{RME-invar:CS},
		$\sigmaLowB[0..2^n-1]$ is $i$-compliant
		with $\SMax(\sigmaLowB[0..2^n-1]) = S_I$. 
\end{lem}

\begin{proof}
	For every set $S \subseteq \mathcal{P}$,
		if $\sigmaLowB[S] \neq \bot$,
		then by construction, $\sigmaLowB[S] = \sigmaLowA[S] \circ \sigma'_S$.
	By \autoref{thm:LowAInvars},
		$\sigmaLowA[0..2^n-1]$ is $(i-1)$-compliant, 
		so by Invariant~\ref{RME-invar:PsubsetS},
		$P(\sigmaLowA[S]) \subseteq S$.
	By the definition of $\sigma'_S$, $\sigma'_S$ contains only steps of processes in $S \cap I$.
	Thus $P(\sigmaLowB[S]) \subseteq S$ (Invariant~\ref{RME-invar:PsubsetS}).
	
	By \autoref{thm:LowAInvars},
		$\sigmaLowA[0..2^n-1]$ is $(i-1)$-compliant
		with $\SMax(\sigmaLowA[0..2^n-1]) = S_I$.
	So by Invariant~\ref{RME-invar:uniqueSMax},
		for every set $S \subseteq \mathcal{P}$,
		$\sigmaLowA[S] \neq \bot$ if and only if
		$F(\sigmaLowA[S_I]) \subseteq S \subseteq S_I$.
	By construction, for every set $S \subseteq \mathcal{P}$,
		$\sigmaLowB[S] = \bot$ if and only if
		$\sigmaLowA[S] = \bot$. 
	Furthermore, by \autoref{thm:LowPreInvars} (\ref{LPI:FUnchanged}),
		$F(\sigmaLowB[S_I]) = F(\sigmaLowA[S_I])$.
	Thus for every set $S \subseteq \mathcal{P}$,
		$\sigmaLowB[S] \neq \bot$ if and only if
		$F(\sigmaLowB[S_I]) \subseteq S \subseteq S_I$
		(Invariant~\ref{RME-invar:uniqueSMax}).
	
	By \autoref{thm:LowPreInvars} (\ref{LPI:statesSame}),
		for every set $S \subseteq \mathcal{P}$ such that 
		$\sigmaLowA[S] \neq \bot$,
		and every process $p \in S$,
		$\textit{state}_p(\sigmaLowB[S_I]) = \textit{state}_p(C'_{S_I},\sigma_{S_I})
		\textit{state}_p(C'_S,\sigma'_S) = \textit{state}_p(\sigmaLowB[S])$.
	By construction, for every set $S \subseteq \mathcal{P}$,
		$\sigmaLowB[S] = \bot$ if and only if $\sigmaLowA[S] = \bot$.
	Furthermore, we have already proven that for every set $S \subseteq \mathcal{P}$,
		$\sigmaLowB[S] \neq \bot$ if and only if
		$F(\sigmaLowB[S_I]) \subseteq S \subseteq S_I$.
	Thus for every process $p \in S_I$
		and every set $S \subseteq \mathcal{P}$ that contains $p$,
		if $\sigmaLowB[S] \neq \bot$, then 
		$\textit{state}_p(\sigmaLowB[S]) = \textit{state}_p(\sigmaLowB[S_I])$
		(Invariant~\ref{RME-invar:subsetStatesSame2}).
	
	Since we have already proven that Invariants~\ref{RME-invar:PsubsetS},
		\ref{RME-invar:uniqueSMax}, and \ref{RME-invar:subsetStatesSame2} hold
		for $\sigmaLowB[0..2^n-1]$,
		it immediately follows that Invariant~\ref{RME-invar:FSame} also holds.
		
	By \autoref{thm:LowPreInvars} (\ref{LPI:IRegsWereSimilar}),  
		for every process $p \in I$,
		if $p$ accesses a register $R$ during $E(C'_{S_I},\sigma'_{S_I})$,
		then there is a value $y_R$ such that for every set $S' \subseteq \mathcal{P}$,
		if $\sigmaLowA[S'] \neq \bot$, then:
	\begin{displaymath}
		\textit{val}_R\bparen{\sigmaLowA[S']}=
		\begin{cases}
				\textit{val}_R(\sigmaLowA[S_I]) & 
					\text{if $p \in S'$} \\
				y_R & \text{otherwise}
		\end{cases}
	\end{displaymath}
	Note that since $p$ accesses $R$ during $E(C'_{S_I},\sigma'_{S_I})$,
		by \autoref{thm:LowPreInvars} (\ref{LPI:IOneAccess}),
		$\textit{last}_R(\sigmaLowB[S_I]) = p$.
	By \autoref{thm:LowAInvars}, $\sigmaLowA[0..2^n-1]$ is $(i-1)$-compliant,
		so if $p \in S'$, then $p$ also accesses $R$ during $E(C'_{S'},\sigma'_{S'})$.
	Thus by \autoref{thm:LowPreInvars} (\ref{LPI:IRegsSimilar}),
		if $p \in S'$, then $\textit{val}_R(\sigmaLowB[S']) = \textit{val}_R(\sigmaLowB[S_I])$.
	Furthermore, by \autoref{thm:LowPreInvars} (\ref{LPI:IOneAccess}),
		for every set $S \subseteq \mathcal{P}$ such that $\sigmaLowA[S] \neq \bot$,
		$R$ cannot be accessed by any process other than $p$
		during $E(C'_S,\sigma'_S)$.
	Therefore if $p \not\in S'$, then $R$ cannot be accessed
		during $E(C'_{S'},\sigma'_{S'})$, and so
		$\textit{val}_R(\sigmaLowB[S']) = y_R$.
	Thus we have that for every register $R \in \mathcal{R}$
		such that some process $p$ accesses $R$ during $E(C'_{S_I},\sigma'_{S_I})$, 
		for every set $S' \subseteq \mathcal{P}$,
		if $\sigmaLowB[S'] \neq \bot$, then:
	\begin{displaymath}
		\textit{val}_R\bparen{\sigmaLowB[S']}=
		\begin{cases}
				\textit{val}_R(\sigmaLowB[S_I]) & 
					\text{if $p = \textit{last}_R(\sigmaLowB[S_I]) \in S'$} \\
				y_R & \text{otherwise}
		\end{cases}
	\end{displaymath}
	Then, since any register that is not accessed by any process
		during $E(C'_{S_I},\sigma'_{S_I})$ clearly does not change its state,
		observe that Invariant~\ref{RME-invar:subsetRegsSimilar} 
		holds for $\sigmaLowB[0..2^n-1]$.
		
	By \autoref{thm:LowAInvars}, $\sigmaLowA[0..2^n-1]$ is $(i-1)$-compliant,
		so Invariant~\ref{RME-invar:crashes} holds for $\sigmaLowA[0..2^n-1]$.
	For every set $S \subseteq \mathcal{P}$ such that $\sigmaLowA[S] \neq \bot$,
		$\sigma'_S$ contains no crash steps.
	Thus Invariant~\ref{RME-invar:crashes} also holds for $\sigmaLowB[0..2^n-1]$.
	
	By \autoref{thm:LowAInvars}, $\sigmaLowA[0..2^n-1]$ is $(i-1)$-compliant
		with $\SMax(\sigmaLowA[0..2^n-1]) = S_I$.
	So by Invariant~\ref{RME-invar:DSMSelfAccessOnly},
		for every process $p \in S_I \setminus F(\sigmaLowA[S_I])$,
		every register $R \in \mathcal{R}$ owned by $p$,
		and every set $S \subseteq \mathcal{P}$ with $\sigmaLowA[S] \neq \bot$, 
		$R$ can only be accessed by $p$ during $E(\sigmaLowA[S])$. 
	By \autoref{thm:LowPreInvars} (\ref{LPI:DSMOwnerIsGone}),
		for every set $S \subseteq \mathcal{P}$ such that $\sigmaLowA[S] \neq \bot$,
		for each process $p \in I$,
		if $p$ accesses a register $R$ during $E(C'_S,\sigma'_S)$,
		then the owner of $R$ is not in $I \setminus \{p\}
		= (S_I \setminus F(\sigmaLowA[S_I]) ) \setminus \{p\}$.
	In other words, for every set $S \subseteq \mathcal{P}$ 
		such that $\sigmaLowA[S] \neq \bot$,
		for each process $p \in I$,
		if $p$ owns a register $R$,
		then no other process in $I = S_I \setminus F(\sigmaLowA[S_I])$
		accesses $R$ during $E(C'_S,\sigma'_S)$.
	By \autoref{thm:LowPreInvars} (\ref{LPI:FUnchanged}),
		$F(\sigmaLowA[S_I]) = F(\sigmaLowB[S_I])$.
	Consequently, for every process $p \in S_I \setminus F(\sigmaLowB[S_I])$,
		every register $R \in \mathcal{R}$ owned by $p$,
		and every set $S \subseteq \mathcal{P}$ with $\sigmaLowB[S] \neq \bot$, 
		$R$ can only be accessed by $p$ during $E(\sigmaLowB[S])$
		(Invariant~\ref{RME-invar:DSMSelfAccessOnly}). 
	
	By \autoref{thm:LowAInvars}, $\sigmaLowA[0..2^n-1]$ is $(i-1)$-compliant
		with $\SMax(\sigmaLowA[0..2^n-1]) = S_I$.
	So by Invariant~\ref{RME-invar:subsetStatesSame2},
		for every process $p \in S_I$ 
		and every set $S \subseteq \mathcal{P}$ that contains $p$,
		if $\sigmaLowA[S] \neq \bot$, then 
		$\textit{state}_p(\sigmaLowA[S]) = \textit{state}_p(\sigmaLowA[S_I])$.
	Furthermore, by Invariant~\ref{RME-invar:ValidCCSame}, 
		in the CC model, for every process
		$p \in S_I \setminus F(\sigmaLowA[S_I])$,
		there is a set $\mathcal{R}_p$ of registers such that
		for every set $S \subseteq \mathcal{P}$ that contains $p$,
		if $\sigmaLowA[S] \neq \bot$, then 
		the set of registers that $p$ has valid cache copies of
		at the end of $E(\sigmaLowA[S])$ is exactly $\mathcal{R}_p$.	
	Recall that for every set $S \subseteq \mathcal{P}$ 
		such that $\sigmaLowA[S] \neq \bot$, 
		$\sigma'_S$ contains exactly one non-crash step 
		of each process in $S_I$ and no other steps.
	By \autoref{thm:LowPreInvars} (\ref{LPI:IOneAccess}),
		for every set $S \subseteq \mathcal{P}$ such that $\sigmaLowA[S] \neq \bot$, 
		each register is accessed by at most one process $p \in I \cap S$
		during $E(C'_S,\sigma'_S)$ and no other processes.
	Furthermore, by \autoref{thm:LowPreInvars} (\ref{LPI:IRegsWereSimilar}),
		for every set $S \subseteq \mathcal{P}$ such that $\sigmaLowA[S] \neq \bot$, 
		for every register $R \in \mathcal{R}$,
		if $R$ is accessed during $E(C'_S,\sigma'_S)$,
		then $\textit{val}_R(\sigmaLowA[S_I]) =
		\textit{val}_R(\sigmaLowA[S])$.
	Thus for every process $p \in S_I$ and 
		every set $S \subseteq \mathcal{P}$ such that $\sigmaLowA[S] \neq \bot$, 
		during both $E(C'_S,\sigma'_S)$ and $E(C'_{S_I},\sigma'_{S_I})$,
		$p$ performs the same operation on the same register 
		which begins with the same value,
		causing the same resulting state and response.
	There are two cases: either this operation that $p$ performs on $R$ 
		is a read operation, or it is not.
		
	If it is a read operation, then no cache copies are invalidated by the read,
		and $p$ creates a new valid cache copy of $R$
		during both $E(C'_S,\sigma'_S)$ and $E(C'_{S_I},\sigma'_{S_I})$,
		thus observe that for every set $S' \subseteq \mathcal{P}$ that contains $p$,
		if $\sigmaLowB[S'] \neq \bot$, then 
		the set of registers that $p$ has valid cache copies of
		is exactly $\mathcal{R}_p \cup \{R\}$.
	If it is not a read operation, then cache copies can be invalidated,
		but by \autoref{thm:LowPreInvars} (\ref{LPI:CCNoInvalidateIminusP}),
		the invalidated cache copies cannot belong to 
		any process in $I \setminus \{p\} = 
		(S_I \setminus F(\sigmaLowA[S_I]) ) \setminus \{p\}$.
	Thus the cache copies of every process in 
		$(S_I \setminus F(\sigmaLowA[S_I]) ) \setminus \{p\}$
		are unaffected, whereas observe that 
		for every set $S' \subseteq \mathcal{P}$ that contains $p$,
		if $\sigmaLowB[S'] \neq \bot$, then 
		the set of registers that $p$ has valid cache copies of
		is exactly $\mathcal{R}_p \setminus \{R\}$.
	By \autoref{thm:LowPreInvars} (\ref{LPI:FUnchanged}),
		$F(\sigmaLowA[S_I]) = F(\sigmaLowB[S_I])$.
	Consequently, in both cases, for every process
		$p \in S_I \setminus F(\sigmaLowB[S_I])$,
		there is a set $\mathcal{R}'_p$ of registers 
		(namely either $\mathcal{R}_p \cup R$ or $\mathcal{R}_p \setminus \{R\}$
		where $R$ is the one register that $p$ is poised to access at the end of $E(\sigmaLowA[S_I])$)
		such that for every set $S \subseteq \mathcal{P}$ that contains $p$,
		if $\sigmaLowB[S] \neq \bot$, then 
		the set of registers that $p$ has valid cache copies of
		is exactly $\mathcal{R}'_p$
		(Invariant~\ref{RME-invar:ValidCCSame}). 
	
	By \autoref{thm:LowAInvars}, $\sigmaLowA[0..2^n-1]$ is $(i-1)$-compliant,
		so by Invariant~\ref{RME-invar:iRMRs},
		for every set $S \subseteq \mathcal{P}$
		and every process $p \in S \setminus F(\sigmaLowA[S])$,
		if $\sigmaLowA[S] \neq \bot$, then 
		$p$ incurs at least $i-1$ RMRs
		during $E(\sigmaLowA[S])$.
	By \autoref{thm:LowPreInvars} (\ref{LPI:OneRMREach}),
		for every set $S \subseteq \mathcal{P}$ such that $\sigmaLowA[S] \neq \bot$, 
		each process in $I \cap S$ incurs exactly one RMR during $E(C'_S,\sigma'_S)$.
	By construction, for every set $S \subseteq \mathcal{P}$,
		if $\sigmaLowB[S] \neq \bot$, then
		$\sigmaLowB[S] = \sigmaLowA[S] \circ \sigma'_S$,
		i.e., every process in $I \cap S = S \setminus F(\sigmaLowA[S_I])$ 
		incurs exactly one more RMR during $E(\sigmaLowB[S])$
		than during $E(\sigmaLowA[S])$.
	By \autoref{thm:LowPreInvars} (\ref{LPI:FUnchanged}),
		$F(\sigmaLowA[S_I]) = F(\sigmaLowB[S_I])$. 
	Since we have already proven that Invariant~\ref{RME-invar:FSame}
		holds for $\sigmaLowB[0..2^n-1]$,
		$F(\sigmaLowB[S_I]) = F(\sigmaLowB[S])$
		for every set $S \subseteq \mathcal{P}$
		such that $\sigmaLowB[S] \neq \bot$.
	Thus for every set $S \subseteq \mathcal{P}$
		and every process $p \in S \setminus F(\sigmaLowB[S])$,
		if $\sigmaLowB[S] \neq \bot$, then 
		$p$ incurs at least $i$ RMRs
		during $E(\sigmaLowA[S])$
		(Invariant~\ref{RME-invar:iRMRs}). 
\end{proof} 
	
We now construct another array $\sigmaLowC[0..2^n-1]$ 
	with the goal of satisfying Invariant~\ref{RME-invar:CS} as follows.
If no process is within the critical section at the end of $E(\sigmaLowB[S_I])$,
	we simply construct $\sigmaLowC[0..2^n-1]$ such that
	$\sigmaLowC[0..2^n-1] = \sigmaLowB[0..2^n-1]$.
Furthermore, we define $\SLMax = S_I$.
  
Otherwise, to avoid violating mutual exclusion,
	there must be exactly one process 
	$p \in S_I \setminus F(\sigmaLowB[S_I])$ such that 
	at the end of $E(\sigmaLowB[S_I])$,
	$p$ is within the critical section.
Then note that by \autoref{thm:LowBInvars} and Invariant~\ref{RME-invar:subsetStatesSame2},
	for every set $S \subseteq \mathcal{P}$ such that $\sigmaLowB[S] \neq \bot$,
	if $p \in S$ then $p$ is also within the critical section at the end of $E(\sigmaLowB[S])$;
	otherwise no process is within the critical section at the end of $E(\sigmaLowB[S])$.
Thus we construct $\sigmaLowC[0..2^n-1]$ such that
	for every set $S \subseteq \mathcal{P}$,
	if $p \in S$, then $\sigmaLowC[S] = \bot$;
	otherwise $\sigmaLowC[S] = \sigmaLowB[S]$.
Furthermore, we define $\SLMax = S_I \setminus \{p\}$.
	
\begin{lem}
	\label{thm:LowCInvars}
	This new array $\sigmaLowC[0..2^n-1]$ is $i$-compliant
		with $\SMax(\sigmaLowC[0..2^n-1]) = \SLMax$.
\end{lem} 

\toOmit{ 

	\begin{proof}
		By \autoref{thm:LowBInvars},
			except for Invariant~\ref{RME-invar:CS},
			$\sigmaLowB[0..2^n-1]$ is $i$-compliant
			with $\SMax(\sigmaLowB[0..2^n-1]) = S_I$. 
			
		If no process is within the critical section at the end of $E(\sigmaLowB[S_I])$,
			then $\sigmaLowC[0..2^n-1] = \sigmaLowB[0..2^n-1]$.
		Thus by \autoref{thm:LowBInvars}
			and \autoref{thm:LowPreInvars} (\ref{LPI:notLeftCS}),
			Invariant~\ref{RME-invar:CS} also holds for $\sigmaLowC[0..2^n-1]$,
			and so it follows that $\sigmaLowC[0..2^n-1] = \sigmaLowB[0..2^n-1]$ 
			is $i$-compliant with $\SMax(\sigmaLowC[0..2^n-1]) = \SLMax = S_I$.
			
		Otherwise, there is exactly one process 
			$p \in S_I \setminus F(\sigmaLowB[S_I])$ such that  
			$p$ is within the critical section at the end of $E(\sigmaLowB[S_I])$,
			and for every set $S \subseteq \mathcal{P}$,
			if $p \in S$, then $\sigmaLowC[S] = \bot$;
			otherwise $\sigmaLowC[S] = \sigmaLowB[S]$.
		By \autoref{thm:LowBInvars},
			except for Invariant~\ref{RME-invar:CS},
			$\sigmaLowB[0..2^n-1]$ is $i$-compliant
			with $\SMax(\sigmaLowB[0..2^n-1]) = S_I$. 
		Thus observe that by the construction of $\sigmaLowC[0..2^n-1]$,
			Invariants~\ref{RME-invar:PsubsetS}, \ref{RME-invar:uniqueSMax},
			\ref{RME-invar:subsetStatesSame2}, \ref{RME-invar:FSame},
			\ref{RME-invar:subsetRegsSimilar}, \ref{RME-invar:crashes},
			\ref{RME-invar:DSMSelfAccessOnly}, \ref{RME-invar:ValidCCSame},
			\ref{RME-invar:iRMRs} must all also hold for $\sigmaLowC[0..2^n-1]$
			with $\SMax(\sigmaLowC[0..2^n-1]) = \SLMax
			= S_I \setminus \{p\}$.
		
		By the construction of $\sigmaLowC[0..2^n-1]$,
			for every set $S \subseteq \mathcal{P}$,
			if $p \in S$, then $\sigmaLowC[S] = \bot$.
		Thus for every set $S \subseteq \mathcal{P}$,
			if $\sigmaLowC[S] \neq \bot$,
			then $p \not\in S$.
		Consequently, for every set $S \subseteq \mathcal{P}$ 
			such that $\sigmaLowC[S] \neq \bot$,
			no process is within the critical section at the end of 
			$E(\sigmaLowC[S]) = E(\sigmaLowB[S])$.
		Therefore by \autoref{thm:LowPreInvars} (\ref{LPI:notLeftCS}), 
			Invariant~\ref{RME-invar:CS} holds for $\sigmaLowC[0..2^n-1]$,
			and thus $\sigmaLowC[0..2^n-1]$ is $i$-compliant
			with $\SMax(\sigmaLowC[0..2^n-1]) = \SLMax$.
	\end{proof}

}

Finally, we terminate this $i$-th iteration by setting 
	$\sigmaFinal[i,0..2^n-1] = \sigmaLowC[0..2^n-1]$.


\paragraph{$i$-th Iteration (High Contention Phase if $|L| < |H|$):} 

By \autoref{thm:SetupBInvars}, $\sigmaSetupB[0..2^n-1]$ is $(i-1)$-compliant
	with $\SMax(\sigmaSetupB[0..2^n-1]) = \SSMax$.
Recall that $H = \bigcup_{R\in\mathcal{R} \atop |B_{R}|\geq k} B_R$,	
	where for every register $R \in \mathcal{R}$,
	$B_R$ is the set of processes poised to access $R$
	at the end of $E(\sigmaSetupB[\SSMax])$.
	
We first divide the processes in $H$ into	groups of exactly $k$ processes
	such that within each group, all processes are poised to access the same register
	at the end of $E(\sigmaSetupB[\SSMax])$.
We make as many such groups as possible. 
Then let $H_1$ be the set of processes in the resulting groups, i.e.,
	$H_1$ is a modification of $H$ where all processes that are not in any group
	are removed.
Note that by this construction, 
	since $H = \bigcup_{R\in\mathcal{R} \atop |B_{R}|\geq k} B_R$,	
	$|H_1| > |H| / 2$.

Next, let $H_2$ be a modification of $H_1$ such that
	for each process $p \in H_1$,
	$p$ is in $H_2$ if and only if both of the following are true:
\begin{itemize}
	\item No process in $H_1$ is poised to access
		a register owned by $p$
		at the end of $E(\sigmaSetupB[\SSMax])$.
		
	\item No process in $H_1$ is poised to access
		a register $R$ at the end of $E(\sigmaSetupB[\SSMax])$
		such that $\textit{last}_R(\sigmaSetupB[\SSMax]) = p$.
\end{itemize} 
Note that since $H_1$ is composed of groups of exactly $k$ processes
	that are all poised to access the same register
	at the end of $E(\sigmaSetupB[\SSMax])$,
	there are at most $|H_1| / k$ registers that are poised to be accessed
	by processes in $H_1$,
	and thus there are at most $2 |H_1| / k$ processes 
	removed in the construction of $H_2$ from $H_1$.
	
Then let $H_3$ be a modification of $H_2$ such that
	each remaining group of $H_2$ with at least $k/4$ processes 
	is shrunk to contain only $k/4$ processes,
	and all other groups are removed.
Since at most $2 |H_1| / k$ processes were removed in 
	the construction of $H_2$ from $H_1$,
	and $k > \log n$, at least half of the processes remain,
	and so it is easy to see that at least a quarter of the groups
	in $H_2$ remain with at least $k / 4$ processes.
So $|H_3| \geq |H_1| / 16 > |H| / 32$.

Next, recall that all registers support only read,
	fetch-and-store (FAS), fetch-and-increment (FAI), and compare-and-swap (CAS) operations.
Let	$\POP$ denote one of these 4 operation types such that 
	the plurality of processes in $H$
	are poised to perform an operation of type $\POP$
	at the end of $E(\sigmaSetupB[\SSMax])$. 
Then let $H_4$ be a modification of $H_3$ such that
	for each process $p \in H_3$,
	$p$ is in $H_4$ if and only if
	$p$ is poised to perform an operation of type $\POP$
	at the end of $E(\sigmaSetupB[\SSMax])$. 
Since there are only 4 operation types,
	$|H_4| \geq |H_3| / 4$. 
	
Now let $H_5$ be a modification of $H_4$ such that
	each group of $H_4$ with at least $k / 32$ processes 
	is shrunk to contain only $k / 32$ processes,
	and all other groups are removed.
Since each group of $H_3$ originally had exactly $k / 4$ processes
	and $|H_4| \geq |H_3| / 4$ where $k > \log n$,
	it is easy to see that at least an eighth of the groups
	in $H_4$ remain with at least $k / 32$ of the processes.
So $|H_5| \geq |H_3| / 64 \geq |H_1| / 1024 > |H| / 2048$.

Let $h$ be the number of remaining groups in $H_5$.
Then since $H_5$ only contains groups 
	with exactly $k / 32$ processes,
	$h = 32 |H_5| / k$.
We arbitrarily order these groups, 
	and construct an array $G[0..h-1]$ such that
	for every integer $j \in \{0, 1, \ldots, h-1\}$,
	$G[j]$ is the $j$-th group in the ordering.
Then for every integer $j \in \{0, 1, \ldots, h-1\}$,
	let $R[j]$ be the register that every process in $G[j]$
	is poised to access at the end of $E(\sigmaSetupB[\SSMax])$.

Finally, let $S_H = F(\sigmaSetupB[\SSMax]) \cup H_5$.
We now construct a new array $\sigmaHighA[0..2^n-1]$ such that
	for every set $S \subseteq \mathcal{P}$,
	if $S \not\subseteq S_H$, 
	then $\sigmaHighA[S] = \bot$;
	otherwise $\sigmaHighA[S] = \sigmaSetupB[S]$.
	 
\begin{lem}
	\label{thm:HighAInvars}
	This new array $\sigmaHighA[0..2^n-1]$ is $(i-1)$-compliant
		with $\SMax(\sigmaHighA[0..2^n-1]) = S_H$.
	Furthermore, for every integer $j \in \{0, 1, \ldots, h-1\}$: 
	\begin{itemize}
		\item For every set $S \subseteq \mathcal{P}$ 
			such that $\sigmaHighA[S] \neq \bot$,
			every process in $G[j] \cap S$ is poised to access $R[j]$
			at the end of $E(\sigmaHighA[S])$.
			
		\item The owner of $R[j]$ is not in $S_H \setminus F(\sigmaHighA[S_H])$.
		
		\item For every set $S \subseteq \mathcal{P}$ 
			such that $\sigmaHighA[S] \neq \bot$,
			$\textit{val}_{R[j]}(\sigmaHighA[S]) = 
			\textit{val}_{R[j]}(\sigmaHighA[S_H])$.
	\end{itemize}
\end{lem} 

\toOmit{

	\begin{proof}
		By \autoref{thm:SetupBInvars},
			$\sigmaSetupB[0..2^n-1]$ is $(i-1)$-compliant
			with $\SMax(\sigmaSetupB[0..2^n-1]) = \SSMax$.
		By construction, $\sigmaHighA[0..2^n-1]$ is simply 
			a modification of $\sigmaSetupB[0..2^n-1]$
			where every set $S \subseteq \mathcal{P}$ 
			that contains any process in $\SSMax \setminus S_H$
			has had $\sigmaHighA[S]$ set to $\bot$,
			where $F(\sigmaSetupB[\SSMax]) \subseteq S_H \subseteq \SSMax$.
		It suffices to observe that since every invariant still holds
			with $\SMax(\sigmaHighA[0..2^n-1]) = S_H$,
			$\sigmaHighA[0..2^n-1]$ is $(i-1)$-compliant
			with $\SMax(\sigmaHighA[0..2^n-1]) = S_H$.
			
		Now let $j$ be an integer in $\{0, 1, \ldots, h-1\}$.
		By definition, $R[j]$ is the register that every process in $G[j]$
			is poised to access at the end of $E(\sigmaSetupB[\SSMax])$.
		By \autoref{thm:SetupBInvars},
			$\sigmaSetupB[0..2^n-1]$ is $(i-1)$-compliant
			with $\SMax(\sigmaSetupB[0..2^n-1]) = \SSMax$.
		So by Invariant~\ref{RME-invar:subsetStatesSame2}, 
			for every process $p \in \SSMax$ 
			and every set $S \subseteq \mathcal{P}$ that contains $p$,
			if $\sigmaSetupB[S] \neq \bot$, then 
			$\textit{state}_p(\sigmaSetupB[S]) = 
			\textit{state}_p(\sigmaSetupB[\SSMax])$.
		Thus for every process $p \in G[j]$ 
			and every set $S \subseteq \mathcal{P}$ that contains $p$,
			if $\sigmaSetupB[S] \neq \bot$, then
			$p$ is also poised to access $R[j]$ 
			at the end of $E(\sigmaSetupB[S])$.
		
		By the construction of $\sigmaHighA[0..2^n-1]$,
			for every set $S \subseteq \mathcal{P}$ 
			such that $\sigmaHighA[S] \neq \bot$,
			$\sigmaHighA[S] = \sigmaSetupB[S]$.
		So for every process $p \in G[j]$ 
			and every set $S \subseteq \mathcal{P}$ that contains $p$,
			if $\sigmaHighA[S] \neq \bot$, then
			$p$ is also poised to access $R[j]$ 
			at the end of $E(\sigmaHighA[S])$.
		Thus for every set $S \subseteq \mathcal{P}$ 
			such that $\sigmaHighA[S] \neq \bot$,
			every process in $G[j] \cap S$ is poised to access $R[j]$
			at the end of $E(\sigmaHighA[S])$.
				
		Now recall that $S_H = F(\sigmaSetupB[\SSMax]) \cup H_5$,
			where $H_5 \subseteq H_2$.
		By construction, for each process $p \in H_1$,
			$p$ is in $H_2$ if and only if both of the following are true:
		\begin{itemize}
			\item No process in $H_1$ is poised to access
				a register owned by $p$
				at the end of $E(\sigmaSetupB[\SSMax])$.
				
			\item No process in $H_1$ is poised to access
				a register $R$ at the end of $E(\sigmaSetupB[\SSMax])$
				such that $\textit{last}_R(\sigmaSetupB[\SSMax]) = p$.
		\end{itemize}
		
		We have already shown that $\sigmaHighA[0..2^n-1]$ is $(i-1)$-compliant
			with $\SMax(\sigmaHighA[0..2^n-1]) = S_H$,
			and that for every set $S \subseteq \mathcal{P}$ 
			such that $\sigmaHighA[S] \neq \bot$,
			every process in $G[j] \cap S$ is poised to access $R[j]$
			at the end of $E(\sigmaHighA[S])$.
		So by Invariant~\ref{RME-invar:uniqueSMax},
			every process in $G[j] \cap S_H = G[j]$
			is poised to access $R[j]$ at the end of $E(\sigmaHighA[S_H])$.
		Thus:
		\begin{itemize}
			\item The owner of $R[j]$ is not in $S_H \setminus F(\sigmaHighA[S_H])$.
			\item $\textit{last}_{R[j]}(\sigmaSetupB[\SSMax])$
				is not in $S_H \setminus F(\sigmaHighA[S_H])$.
		\end{itemize}
		
		By \autoref{thm:SetupBInvars},
			$\sigmaSetupB[0..2^n-1]$ is $(i-1)$-compliant
			with $\SMax(\sigmaSetupB[0..2^n-1]) = \SSMax$.
		So by Invariant~\ref{RME-invar:subsetRegsSimilar},
			for every register $R \in \mathcal{R}$,
			there is a value $y_R$ such that
			for every set $S \subseteq \mathcal{P}$,
			if $\sigmaSetupB[S] \neq \bot$, then:
		\begin{displaymath}
			\textit{val}_R\bparen{\sigmaSetupB[S]}=
			\begin{cases}
					\textit{val}_R(\sigmaSetupB[\SSMax]) & 
						\text{if $\textit{last}_R(\sigmaSetupB[\SSMax]) \in S$} \\
					y_R & \text{otherwise}
			\end{cases}
		\end{displaymath}
		
		We have just shown that $\textit{last}_{R[j]}(\sigmaSetupB[\SSMax])$
			is not in $S_H \setminus F(\sigmaHighA[S_H])$.
		So either $\textit{last}_{R[j]}(\sigmaSetupB[\SSMax]) \in F(\sigmaHighA[S_H])$
			or $\textit{last}_R(\sigmaSetupB[\SSMax]) \not\in S_H$.
		 
		By the construction of $\sigmaHighA[0..2^n-1]$, 
			for every set $S \subseteq \mathcal{P}$,
			if $\sigmaHighA[S] \neq \bot$, then
			$\sigmaHighA[S] = \sigmaSetupB[S]$.
		Thus if $\textit{last}_{R[j]}(\sigmaSetupB[\SSMax])$ 
			is in $F(\sigmaHighA[S_H])$, 
			then for every set $S \subseteq \mathcal{P}$ 
			such that $\sigmaHighA[S] \neq \bot$,
			$\textit{val}_{R[j]}(\sigmaHighA[S]) = 
			\textit{val}_{R[j]}(\sigmaSetupB[\SSMax])$.
			
		Otherwise $\textit{last}_R(\sigmaSetupB[\SSMax])$ is not in $S_H$.
		Since we have already proven that
			$\sigmaHighA[0..2^n-1]$ is $(i-1)$-compliant
			with $\SMax(\sigmaHighA[0..2^n-1]) = S_H$,
			by Invariant~\ref{RME-invar:uniqueSMax},
			for every set $S \subseteq \mathcal{P}$ 
			such that $\sigmaHighA[S] \neq \bot$,
			$S \subseteq S_H$, and so 
			$\textit{last}_R(\sigmaSetupB[\SSMax])$ is not in $S$.
		Thus for every set $S \subseteq \mathcal{P}$ 
			such that $\sigmaHighA[S] \neq \bot$,
			$\textit{val}_R(\sigmaHighA[S]) = y_R$. 
			
		So in both cases, for every set $S \subseteq \mathcal{P}$ 
			such that $\sigmaHighA[S] \neq \bot$,
			$\textit{val}_{R[j]}(\sigmaHighA[S]) = 
			\textit{val}_{R[j]}(\sigmaHighA[S_H])$.
	\end{proof}
	
}
 
We now iterate over $j \in \{0, 1, \ldots, h-1\}$
	to construct two arrays $\alpha_1[0..h-1]$ and $\alpha_2[0..h-1]$
	of processes and an array $\alpha[0..h-1]$ of schedules as follows.
If $\POP$ is not CAS, then let $\alpha_1[j]$ and $\alpha_2[j]$ 
	be two arbitrary but distinct processes in $G[j]$,
	and let $\alpha[j] = \alpha_1[j] \circ \alpha_2[j]$.
Otherwise, consider the register $R[j]$ that every process in $G[j]$
	is poised to access at the end of $E(\sigmaHighA[S_H])$.
Let $v_j = \textit{val}_{R[j]}(\sigmaHighA[S_H] \circ \alpha[0] 
	\circ \alpha[1] \circ \ldots \circ \alpha[j-1])$.
Then let $\alpha_1[j]$ be any process in $G[j]$
	such that $\alpha_1[j]$ is about to perform
	a $\textsc{CAS}(v_j,v')$ operation where $v' \neq v_j$;
	if no such process exists, then let $\alpha_1[j]$ be any process in $G[j]$.
Next, let $\alpha_2[j]$ be any process in $G[j] \setminus \{\alpha_1[j]\}$.
Finally, let $\alpha[j] = \alpha_1[j] \circ \alpha_2[j]$.

For every integer $j \in \{0, 1, \ldots, h-1\}$,
	let $\sigma^j_\alpha$ be the concatenation of all schedules in $\alpha[0..j]$,
	i.e., $\sigma^j_\alpha = \alpha[0] \circ \alpha[1] \circ \ldots \circ \alpha[j]$.
Then let $\sigma_\alpha = \sigma^{h-1}_\alpha$.
Furthermore, let $S_\alpha$ be the set of all processes
	with steps in $\sigma_\alpha$.
Note that since $\alpha[j] = \alpha_1[j] \circ \alpha_2[j]$ for $0 \leq j \leq h-1$,
\begin{displaymath}
	S_\alpha=\bigcup^{h-1}_{j=0} \{\alpha_1[j], \alpha_2[j]\}
\end{displaymath}
In addition, note that $S_\alpha \subseteq H_5$
	and $|S_\alpha| = 2h = 64 |H_5| / k$. 
	
Next, let $S_F = S_\alpha \cup F(\sigmaHighA[S_H])$. 
Since $S_\alpha \subseteq H_5$ and $S_H = F(\sigmaHighA[S_H]) \cup H_5$,
	$F(\sigmaHighA[S_H]) \subseteq S_F \subseteq S_H$.
By \autoref{thm:HighAInvars},
	$\sigmaHighA[0..2^n-1]$ is $(i-1)$-compliant, so:
\begin{itemize}
	\item By Invariant~\ref{RME-invar:uniqueSMax},
		$\sigmaHighA[S_F] \neq \bot$.
		
	\item By Invariant~\ref{RME-invar:PsubsetS},
		$P(\sigmaHighA[S_F]) \subseteq S_F$.
		
	\item By Invariant~\ref{RME-invar:FSame}
		$F(\sigmaHighA[S_F]) = F(\sigmaHighA[S_H])$.
\end{itemize}
Thus observe that to avoid violating deadlock freedom,
	there must exist a schedule $\sigma_F$ such that:
\begin{itemize}
	\item $\sigma_F$ begins with exactly one crash step 
		of every process in $S_\alpha$,
		and contains no other crash steps.
		
	\item $\sigma_F$ contains only steps of processes in
		$S_\alpha = S_F \setminus F(\sigmaHighA[S_F])$, i.e.,
		$P(\sigma_F) = S_\alpha$.
		
	\item During $E(\sigmaHighA[S_F] \circ \sigma_\alpha \circ \sigma_F)$,
		every process in $S_F$ begins and then completes a super-passage, i.e.,
		$F(\sigmaHighA[S_F] \circ \sigma_\alpha \circ \sigma_F) = S_F$.
\end{itemize}

Let $C_F$ be the configuration at the end of $E(\sigmaHighA[S_F] \circ \sigma_\alpha)$.
Then let $\mathcal{R}_F$ be the set of every register that is accessed
	during $E(C_F,\sigma_F)$ (after the crash steps of 
	every process in $S_\alpha$ at the beginning of $\sigma_F$).
Next, let $\mathcal{D} \subseteq \mathcal{P}$
	be the set of every process $p \in H_5 \setminus S_\alpha$
	such that there exists a register $R \in \mathcal{R}_F$ such that
	either $p$ owns $R$, or $\textit{last}_R(\sigmaHighA[S_H]) = p$.
	
\begin{lem}
	\label{thm:DSize}
	$|\mathcal{D}| \leq 2 |S_\alpha| \log n = \frac{128}{k} |H_5| \log n$. 
\end{lem} 

\toOmit{

	\begin{proof}
		First, consider each register $R \in \mathcal{R}_F$ such that
			$R$ is owned by a process in $S_\alpha$.
		Since $\mathcal{D} \subseteq H_5 \setminus S_\alpha$, 
			the owner of $R$ is not in $\mathcal{D}$.
		Furthermore, by \autoref{thm:HighAInvars},
			$\sigmaHighA[0..2^n-1]$ is $(i-1)$-compliant
			with $\SMax(\sigmaHighA[0..2^n-1]) = S_H$,
			so by Invariant~\ref{RME-invar:DSMSelfAccessOnly},
			since the owner of $R$ is in 
			$S_\alpha \subseteq H_5 \subseteq S_H \setminus F(\sigmaHighA[S_H])$,
			$R$ cannot be accessed by any process in $H_5 \setminus S_\alpha$
			during $E(\sigmaHighA[S_H])$.
		So $\textit{last}_R(\sigmaHighA[S_H]) \not\in \mathcal{D}$.
		Thus intuitively, each register $R \in \mathcal{R}_F$ that
			is owned by a process in $S_\alpha$ does not 
			contribute any processes to $\mathcal{D}$.
		
		So it suffices to consider the registers in $\mathcal{R}_F$
			that are not owned by any process in $S_\alpha$. 
		By Assumption~\ref{RME-assumption:logNRMRs},
			each process accesses at most $\log n$ registers 
			that it does not own during a passage.
		Thus there are at most $|S_\alpha| \log n$ 
			registers in $\mathcal{R}_F$
			that are not owned by any process in $S_\alpha$.
		Consequently, $|\mathcal{D}| \leq 2 |S_\alpha| \log n
			= \frac{160}{k} |H_5| \log n$. 
	\end{proof}
	
}

Now let $H_6$ be a modification of $H_5$ where
	every process $p \in H_5$ is in $H_6$
	if and only if $p \not\in \mathcal{D}$.
Note that since $\mathcal{D} \subseteq H_5 \setminus S_\alpha$
	and $S_\alpha \subseteq H_5$,
	$S_\alpha \subseteq H_6$.
By \autoref{thm:DSize}, $|\mathcal{D}| \leq \frac{128}{k} |H_5| \log n$,
	so $|H_6| \geq |H_5| - \frac{128}{k} |H_5| \log n$.
For sufficiently large $k$ ($k \geq 256 \log n$),
	$\frac{128}{k} |H_5| \log n \leq 0.5|H_5|$.
Thus, at least half of the processes in $H_5$ remain in $H_6$,
	and so it is easy to see that at least a quarter of the groups
	in $H_6$ remain with at least $k / 128$ processes
	(out of the $k/32$ originally in $H_5$).
Furthermore, since $S_\alpha \subseteq H_6$,
	for $0 \leq j \leq h-1$,
	since $S_\alpha$ contains $\alpha_1[j]$ and $\alpha_2[j]$,
	$\{\alpha_1[j], \alpha_2[j]\} \subseteq G[j] \cap H_6$.
So let $G'[0..h-1]$ be a new array such that 
	for all $j \in \{0, 1, \ldots, h-1\}$,
	$G'[j] = G[j] \cap H_6$, and so 
	$\{\alpha_1[j], \alpha_2[j]\} \subseteq G'[j]$.
	
Now for every integer $j \in \{0, 1, \ldots, h-1\}$, 
	let $\beta_1[j]$ be $\varnothing$ if $|G'[j]| < k / 160$;
	otherwise, let $\beta_1[j]$ be an arbitrary process in
	$G'[j] \setminus \{\alpha_1[j], \alpha_2[j]\}$.
Then let:
\begin{displaymath}
	S_\beta=\bigcup^{h-1}_{j=0} \{\alpha_1[j], \alpha_2[j], \beta_1[j]\}
\end{displaymath}
So: 
\begin{displaymath}
	S_\beta= S_\alpha \cup \bigcup^{h-1}_{j=0} \{\beta_1[j]\}
\end{displaymath} 
Note that by construction, 
	$S_\alpha \subseteq S_\beta \subseteq H_6 \subseteq H_5$.

\begin{lem}
	\label{thm:BetaSize}
	$|S_\beta \setminus S_\alpha| > \frac{|H|}{204.8 k}$. 
\end{lem} 
	
\toOmit{
	\begin{proof}
		Since at least a quarter of the $h$ groups in $H_6$
			have at least $k/128$ processes,
			$\beta_1[j] \neq \varnothing$ for at least
			a quarter of the integers $j \in \{0,1, \ldots, h-1\}$.
		Thus $|S_\beta \setminus S_\alpha| \geq 0.25h = 8 |H_5| / k$.
		
		Now recall that $|H_5| \geq |H_3| / 64 \geq |H_1| / 1024 > |H| / 2048$.
		So $|S_\beta \setminus S_\alpha| \geq 10 |H_5| / k > \frac{|H|}{204.8 k}$.
	\end{proof} 
}	
	
Next, let $S_B = S_\beta \cup F(\sigmaHighA[S_H])$.   
We now construct a new array $\sigmaHighB[0..2^n-1]$ such that
	for every set $S \subseteq \mathcal{P}$,
	if $S \not\subseteq S_B$, 
	then $\sigmaHighB[S] = \bot$;
	otherwise $\sigmaHighB[S] = \sigmaHighA[S]$.
Note that since $S_\beta \subseteq H_6 \subseteq H_5$ and 
	$S_H = F(\sigmaHighA[S_H]) \cup H_5$,
	$F(\sigmaHighA[S_H]) \subseteq S_B \subseteq S_H$.

\begin{lem}
	\label{thm:HighBInvars}
	This new array $\sigmaHighB[0..2^n-1]$ is $(i-1)$-compliant
		with $\SMax(\sigmaHighB[0..2^n-1]) = S_B$.
	Furthermore, for every integer $j \in \{0, 1, \ldots, h-1\}$: 
	\begin{itemize}
		\item For every set $S \subseteq \mathcal{P}$ 
			such that $\sigmaHighB[S] \neq \bot$,
			every process in $G[j] \cap S$ is poised to access $R[j]$
			at the end of $E(\sigmaHighB[S])$.
			
		\item The owner of $R[j]$ is not in $S_B \setminus F(\sigmaHighB[S_B])$.
		
		\item For every set $S \subseteq \mathcal{P}$ 
			such that $\sigmaHighB[S] \neq \bot$,
			$\textit{val}_{R[j]}(\sigmaHighB[S]) = 
			\textit{val}_{R[j]}(\sigmaHighB[S_B])$.
	\end{itemize}
	 
	In addition, for every register $R \in \mathcal{R}_F$,
	\begin{itemize}
		\item The owner of $R$ is not in $S_B \setminus F(\sigmaHighB[S_B])$.
		
		\item For every set $S \subseteq \mathcal{P}$ 
			such that $\sigmaHighB[S] \neq \bot$,
			$\textit{val}_R(\sigmaHighB[S]) = 
			\textit{val}_R(\sigmaHighB[S_B])$.
	\end{itemize} 
\end{lem}

\toOmit{

	\begin{proof}
		By \autoref{thm:HighAInvars},
			$\sigmaHighA[0..2^n-1]$ is $(i-1)$-compliant
			with $\SMax(\sigmaHighA[0..2^n-1]) = S_H$.
		By construction, $\sigmaHighB[0..2^n-1]$ is simply 
			a modification of $\sigmaHighA[0..2^n-1]$
			where every set $S \subseteq \mathcal{P}$ 
			that contains any process in 
			$S_H \setminus S_B$
			has had $\sigmaHighB[S]$ set to $\bot$,
			where $F(\sigmaHighA[S_H]) \subseteq S_B \subseteq S_H$.
		It suffices to observe that since every invariant still holds
			with $\SMax(\sigmaHighB[0..2^n-1]) = S_B$,
			$\sigmaHighB[0..2^n-1]$ is $(i-1)$-compliant
			with $\SMax(\sigmaHighB[0..2^n-1]) = S_B$.
			 
		Furthermore, by \autoref{thm:HighAInvars}, 
			for every integer $j \in \{0, 1, \ldots, h-1\}$: 
		\begin{itemize}
			\item For every set $S \subseteq \mathcal{P}$ 
				such that $\sigmaHighA[S] \neq \bot$,
				every process in $G[j] \cap S$ is poised to access $R[j]$
				at the end of $E(\sigmaHighA[S])$.
				
			\item The owner of $R[j]$ is not in $S_H \setminus F(\sigmaHighA[S_H])$.
			
			\item For every set $S \subseteq \mathcal{P}$ 
				such that $\sigmaHighA[S] \neq \bot$,
				$\textit{val}_{R[j]}(\sigmaHighA[S]) = 
				\textit{val}_{R[j]}(\sigmaHighA[S_H])$.
		\end{itemize} 
		By the construction of $\sigmaHighB[0..2^n-1]$,
			for every set $S \subseteq \mathcal{P}$ 
			such that $\sigmaHighB[S] \neq \bot$,
			$\sigmaHighB[S] = \sigmaHighA[S]$.
		Thus for every integer $j \in \{0, 1, \ldots, h-1\}$ and
			every set $S \subseteq \mathcal{P}$ 
			such that $\sigmaHighB[S] \neq \bot$,
			every process in $G[j] \cap S$ is poised to access $R[j]$
			at the end of $E(\sigmaHighB[S])$.
		
		By \autoref{thm:HighAInvars},
			$\sigmaHighA[0..2^n-1]$ is $(i-1)$-compliant.
		So by Invariant~\ref{RME-invar:FSame},
			$F(\sigmaHighA[S_H]) = F(\sigmaHighA[S_B]) = F(\sigmaHighB[S_B])$.
		Then, since $S_B \subseteq S_H$,
			for every integer $j \in \{0, 1, \ldots, h-1\}$,
			the owner of $R[j]$ is not in $S_B \setminus F(\sigmaHighB[S_B])$.
		
		Next, for every integer $j \in \{0, 1, \ldots, h-1\}$ and
			every set $S \subseteq \mathcal{P}$ 
			such that $\sigmaHighB[S] \neq \bot$, 
		\begin{align*}
			\textit{val}_{R[j]}(\sigmaHighB[S]) 
				& = \textit{val}_{R[j]}(\sigmaHighA[S]) \\
				& = \textit{val}_{R[j]}(\sigmaHighA[S_H]) \\
				& = \textit{val}_{R[j]}(\sigmaHighA[S_B]) \\
				& = \textit{val}_{R[j]}(\sigmaHighB[S_B])
		\end{align*}
		
		Thus we have proven that 
			for every integer $j \in \{0, 1, \ldots, h-1\}$: 
		\begin{itemize}
			\item For every set $S \subseteq \mathcal{P}$ 
				such that $\sigmaHighB[S] \neq \bot$,
				every process in $G[j] \cap S$ is poised to access $R[j]$
				at the end of $E(\sigmaHighB[S])$.
				
			\item The owner of $R[j]$ is not in $S_B \setminus F(\sigmaHighB[S_B])$.
			
			\item For every set $S \subseteq \mathcal{P}$ 
				such that $\sigmaHighB[S] \neq \bot$,
				$\textit{val}_{R[j]}(\sigmaHighB[S]) = 
				\textit{val}_{R[j]}(\sigmaHighB[S_B])$.
		\end{itemize} 
		
		By definition, $\mathcal{D} \subseteq \mathcal{P}$
			is the set of every process $p \in H_5 \setminus S_\alpha$
			such that there exists a register $R \in \mathcal{R}_F$ such that
			either $p$ owns $R$, or $\textit{last}_R(\sigmaHighA[S_H]) = p$.
		By construction, $H_6 \cap \mathcal{D} = \varnothing$,
			$S_\beta \subseteq H_6$, and $S_B = S_\beta \cup F(\sigmaHighA[S_H])$.
		Furthermore, recall that $F(\sigmaHighA[S_H]) = F(\sigmaHighB[S_B])$.
		Thus for every register $R \in \mathcal{R}_F$, 
			the owner of $R$ is not in $S_B \setminus F(\sigmaHighB[S_B])$.
		
		By \autoref{thm:HighAInvars},
			$\sigmaHighA[0..2^n-1]$ is $(i-1)$-compliant
			with $\SMax(\sigmaHighA[0..2^n-1]) = S_H$.
		So by Invariant~\ref{RME-invar:subsetRegsSimilar},
			for every register $R \in \mathcal{R}_F$,
			there is a value $y_R$ such that
			for every set $S \subseteq \mathcal{P}$,
			if $\sigmaHighA[S] \neq \bot$, then:
		\begin{displaymath}
			\textit{val}_R\bparen{\sigmaHighA[S]}=
			\begin{cases}
					\textit{val}_R(\sigmaHighA[S_H]) & 
						\text{if $\textit{last}_R(\sigmaHighA[S_H]) \in S$} \\
					y_R & \text{otherwise}
			\end{cases}
		\end{displaymath}
		Since $\textit{last}_R(\sigmaHighA[S_H]) \in \mathcal{D}$
			and $\mathcal{D} \cap S_\beta = \varnothing$,
			for every set $S \subseteq \mathcal{P}$ such that 
			$F(\sigmaHighA[S_H]) = F(\sigmaHighB[S_B]) \subseteq 
			S \subseteq S_B \subseteq S_H$,
			$\textit{val}_R(\sigmaHighB[S]) = y_R$.
		Furthermore, since we have already proven that
			$\sigmaHighB[0..2^n-1]$ is $(i-1)$-compliant
			with $\SMax(\sigmaHighB[0..2^n-1]) = S_B$,
			by Invariant~\ref{RME-invar:uniqueSMax},
			for every set $S \subseteq \mathcal{P}$ such that $\sigmaHighB[S] \neq \bot$,
			$F(\sigmaHighB[S_B]) \subseteq S \subseteq S_B$.
		Thus for every register $R \in \mathcal{R}_F$,
			and every set $S \subseteq \mathcal{P}$ such that $\sigmaHighB[S] \neq \bot$,
			$\textit{val}_R(\sigmaHighB[S]) = \textit{val}_R(\sigmaHighB[S_B])$.
	\end{proof}

}
 
We now iterate over $j \in \{0, 1, \ldots, h-1\}$
	to construct an array $\beta[0..h-1]$ of schedules as follows.
Recall that by definition,
\begin{itemize}
	\item $\beta_1[j]$ is $\varnothing$ if $|G'[j]| < k / 160$;
		otherwise, $\beta_1[j]$ is an arbitrary process in
		$G'[j] \setminus \{\alpha_1[j], \alpha_2[j]\}$.	
		
	\item $v_j = \textit{val}_{R[j]}(\sigmaHighA[S_H] \circ \sigma^{j-1}_\alpha)$.
\end{itemize}  
If $\POP$ is CAS and $\beta_1[j] \neq \varnothing$,
	then let $v_\beta$ and $v'_\beta$ be such that $\beta_1[j]$ 
	is poised to perform a $\textsc{CAS}(v_\beta,v'_\beta)$ operation on $R[j]$ 
	at the end of $E(\sigmaHighB[S_B])$.
Note that by \autoref{thm:HighBInvars}
	and Invariant~\ref{RME-invar:subsetStatesSame2},
	$\beta_1[j]$ would also be poised to perform 
	a $\textsc{CAS}(v_\beta,v'_\beta)$ operation on $R[j]$ 
	at the end of $E(\sigmaHighB[S])$ for every set $S \subseteq \mathcal{P}$
	such that $\sigmaHighB[S] \neq \bot$ and $\beta_1[j] \in S$.
We then define:
\begin{displaymath}
	\beta[j] =
	\begin{cases}
			\alpha_1[j] \circ \alpha_2[j] = \alpha[j] & \text{if $|G'[j] < k / 160$} \\
			\beta_1[j] \circ \alpha_2[j] & \text{if $\POP$ is FAI} \\
			\beta_1[j] \circ \alpha_1[j] \circ \alpha_2[j] & 
				\text{if $\POP$ is CAS and $v_\beta \neq v_j$} \\
			\alpha_1[j] \circ \beta_1[j] \circ \alpha_2[j] & \text{otherwise} 
	\end{cases}
\end{displaymath} 
 
By \autoref{thm:HighBInvars},
	$\sigmaHighB[0..2^n-1]$ is $(i-1)$-compliant
	with $\SMax(\sigmaHighB[0..2^n-1]) = S_B$.
Furthermore, $S_\alpha \subseteq S_\beta \subseteq S_B$,
	so $F(\sigmaHighB[S_B]) \cup S_\alpha \subseteq S_B$.
Thus by Invariant~\ref{RME-invar:uniqueSMax},
	for every set $S \subseteq \mathcal{P}$ such that 
	$F(\sigmaHighB[S_B]) \cup S_\alpha \subseteq S \subseteq S_B$, 
	$\sigmaHighB[S] \neq \bot$.
	
So for every set $S \subseteq \mathcal{P}$ such that 
	$F(\sigmaHighB[S_B]) \cup S_\alpha \subseteq S \subseteq S_B$, 
	let $C'_S$ be the configuration at the end of $E(\sigmaHighB[S])$,
	and let $\sigma'_S$ be a modification of $\sigma_\alpha$ such that 
	for every integer $j \in \{0, 1, \ldots, h-1\}$,
	if $\beta_1[j] \neq \varnothing$ and $\beta_1[j] \in S$,
	then $\alpha[j]$ is replaced by $\beta[j]$ in $\sigma'_S$.
Note that by this construction, $\sigma'_S$ contains exactly one non-crash step
	of each process in $S_\beta \cap S$ and no other steps.
	
\begin{lem}
	\label{thm:ExactlyGroupRegsAccessed}
	For every set $S \subseteq \mathcal{P}$ such that 
		$F(\sigmaHighB[S_B]) \cup S_\alpha \subseteq S \subseteq S_B$,
		the set of registers accessed during $E(C'_S,\sigma'_S)$
		is exactly the set of registers accessed during $E(C'_S,\sigma_\alpha)$ 
		and exactly the set of registers in $R[0..h-1]$. 
\end{lem}

\toOmit{

	\begin{proof}
		First, recall that for every set $S \subseteq \mathcal{P}$ such that 
			$F(\sigmaHighB[S_B]) \cup S_\alpha \subseteq S \subseteq S_B$, 
			$\sigmaHighB[S] \neq \bot$.
		So let $S \subseteq \mathcal{P}$ be a set of processes such that
			$F(\sigmaHighB[S_B]) \cup S_\alpha \subseteq S \subseteq S_B$. 
		
		By \autoref{thm:HighBInvars}, for every integer $j \in \{0, 1, \ldots, h-1\}$,
			every process in $G[j] \cap S$ is poised to access $R[j]$
			at the end of $E(\sigmaHighB[S])$.
		By construction, both $\sigma'_S$ and $\sigma_\alpha$
			contain at most one non-crash step 
			of each process in $S_\beta \cap S$ and no other steps.
		Since $S_\beta \subseteq H_5$, where $G[0..h-1]$ 
			are the groups of processes that constitute $H_5$,
			every process with a step in $\sigma'_S$ or $\sigma_\alpha$
			is poised to access a register in $R[0..h-1]$
			at the end of $E(\sigmaHighB[S])$.
		Therefore every register accessed during either $E(C'_S,\sigma'_S)$
			or $E(C'_S,\sigma_\alpha)$ is in $R[0..h-1]$.
			
		Next, by construction, for every integer $j \in \{0, 1, \ldots, h-1\}$,
			both $\sigma'_S$ and $\sigma_\alpha$ 
			contain a non-crash step of $\alpha_2[j]$.
		By \autoref{thm:HighBInvars}, 
			for every integer $j \in \{0, 1, \ldots, h-1\}$,
			$\alpha_2[j]$ is poised to access $R[j]$
			at the end of $E(\sigmaHighB[S])$.
		So every register in $R[0..h-1]$ is accessed
			during both $E(C'_S,\sigma'_S)$ and $E(C'_S,\sigma_\alpha)$.
			
		Thus we have shown that 
			the set of registers accessed during $E(C'_S,\sigma'_S)$
			is exactly the set of registers accessed during $E(C'_S,\sigma_\alpha)$ 
			and exactly the set of registers in $R[0..h-1]$. 
	\end{proof}
}

Now for every integer $j \in \{0,1,\ldots,h-1\}$,
	let $\sigma'_S[j]$ and $\sigma_\alpha[j]$ be the suffixes 
	of $\sigma'_S$ and $\sigma_\alpha$ that contain
	only the steps of processes in $G[0..j]$.
(So for every integer $j \in \{0,1,\ldots,h-1\}$, $\sigma_\alpha[j] = \sigma^j_\alpha$.)
We also define $\sigma'_S[-1] = \sigma_\alpha[-1] = \varnothing$.

\begin{lem}
	\label{thm:AlphaBetaIndistinguishable}
	For every register $R \in \mathcal{R}$
		every integer $j \in \{0, 1, \ldots, h-1\}$,
		and every set $S \subseteq \mathcal{P}$ such that 
		$F(\sigmaHighB[S_B]) \cup S_\alpha \subseteq S \subseteq S_B$,
		$\textit{val}_R(C'_S, \sigma'_S[j]) = \textit{val}_R(C'_S,\sigma_\alpha[j])$. 
\end{lem}

\toOmit{

	\begin{proof}
		First, recall that for every set $S \subseteq \mathcal{P}$ such that 
			$F(\sigmaHighB[S_B]) \cup S_\alpha \subseteq S \subseteq S_B$, 
			$\sigmaHighB[S] \neq \bot$.
		So let $S \subseteq \mathcal{P}$ be a set of processes such that
			$F(\sigmaHighB[S_B]) \cup S_\alpha \subseteq S \subseteq S_B$.
		
		By definition, $\sigma'_S[-1] = \sigma_\alpha[-1] = \varnothing$.
		Thus for every register $R \in \mathcal{R}$,
			$\textit{val}_R(C'_S, \sigma'_S[-1]) = \textit{val}_R(C'_S,\sigma_\alpha[-1])$.
		So it suffices to show that for every register $R \in \mathcal{R}$
			and every integer $j \in \{0, 1, \ldots, h-1\}$,
			if $\textit{val}_R(C'_S, \sigma'_S[j-1]) = \textit{val}_R(C'_S,\sigma_\alpha[j-1])$,
			then $\textit{val}_R(C'_S, \sigma'_S[j]) = \textit{val}_R(C'_S,\sigma_\alpha[j])$.
		
		Thus let $j$ be an integer in $\{0, 1, \ldots, h-1\}$,
			and suppose that for every register $R \in \mathcal{R}$,
			$\textit{val}_R(C'_S, \sigma'_S[j-1]) = \textit{val}_R(C'_S,\sigma_\alpha[j-1])$.
		By the construction of $\sigma'_S$,
			$\alpha[j]$ is replaced by $\beta[j]$
			if and only if $\beta_1[j] \neq \varnothing$ and $\beta_1[j] \in S$.
		Thus if either $\beta_1[j] = \varnothing$ or $\beta_1[j] \not\in S$,
			then for every register $R \in \mathcal{R}$,
			$\textit{val}_R(C'_S, \sigma'_S[j]) = 
			\textit{val}_R(C'_S,\sigma_\alpha[j])$ as wanted.
		
		Otherwise, $\beta_1[j] \neq \varnothing$ and $\beta_1[j] \in S$.
		Then by the definition of $\beta[j]$:
		\begin{displaymath}
			\beta[j] =
			\begin{cases} 
					\beta_1[j] \circ \alpha_2[j] & \text{if $\POP$ is FAI} \\
					\beta_1[j] \circ \alpha_1[j] \circ \alpha_2[j] & 
						\text{if $\POP$ is CAS and $v_\beta \neq v_j$} \\
					\alpha_1[j] \circ \beta_1[j] \circ \alpha_2[j] & \text{otherwise} 
			\end{cases}
		\end{displaymath}
		Since $\alpha_1[j]$ and $\alpha_2[j]$ are in $S_\alpha \subseteq S$,
			all of $\alpha_1[j]$, $\alpha_2[j]$, and $\beta_1[j]$ are in $G[j] \cap S$.
		So by \autoref{thm:HighBInvars},
			all of $\alpha_1[j]$, $\alpha_2[j]$, and $\beta_1[j]$ are poised
			to access $R[j]$ at the end of $E(\sigmaHighB[S])$.
		By construction, both $\sigma'_S$ and $\sigma_\alpha$
			contain at most one non-crash step of each process and no other steps.
		Thus for each process $p \in \{\alpha_1[j], \alpha_2[j], \beta_1[j]\}$,
			$\textit{state}_p(C'_S,\sigma'_S[j-1]) = \textit{state}_p(\sigmaHighB[S])
			= \textit{state}_p(C'_S,\sigma_\alpha[j-1])$.
			
		Therefore all of $\alpha_1[j]$, $\alpha_2[j]$, and $\beta_1[j]$ are poised
			to access $R[j]$ at the end of both 
			$E(C'_S,\sigma'_S[j-1])$ and $E(C'_S,\sigma_\alpha[j-1])$.		
		Thus for every register $R \in \mathcal{R} \setminus R[j]$,
			$\textit{val}_R(C'_S, \sigma'_S[j]) = \textit{val}_R(C'_S,\sigma_\alpha[j])$.  
		It now suffices to show that 
			$\textit{val}_{R[j]}(C'_S, \sigma'_S[j]) 
			= \textit{val}_{R[j]}(C'_S,\sigma_\alpha[j])$.
				
		If $\POP$ is read, then since reads do not change the value of a register,
			$\textit{val}_{R[j]}(C'_S, \sigma'_S[j]) 
			= \textit{val}_{R[j]}(C'_S,\sigma_\alpha[j])$ as wanted. 
			
		If $\POP$ is FAS, 
			then $\alpha[j] = \alpha_1[j] \circ \alpha_2[j]$
			and $\beta[j] = \alpha_1[j] \circ \beta_1[j] \circ \alpha_2[j]$.
		Let $v_2$ be the value such that $\alpha_2[j]$
			is poised to perform $\textsc{FAS}(v_2)$ at the end of both 
			$E(C'_S,\sigma'_S[j-1])$ and $E(C'_S,\sigma_\alpha[j-1])$.	
		Then $\textit{val}_{R[j]}(C'_S, \sigma_\alpha[j]) = v_2$ and 
			$\textit{val}_{R[j]}(C'_S, \sigma'_S[j]) = v_2$.
		Thus $\textit{val}_{R[j]}(C'_S, \sigma'_S[j]) 
			= \textit{val}_{R[j]}(C'_S,\sigma_\alpha[j])$ as wanted.
			
		If $\POP$ is FAI,
			then $\alpha[j] = \alpha_1[j] \circ \alpha_2[j]$
			and $\beta[j] = \beta_1[j] \circ \alpha_2[j]$.
		So $\textit{val}_{R[j]}(C'_S, \sigma_\alpha[j]) = 
			\textit{val}_{R[j]}(C'_S,\sigma'_\alpha[j-1]) + 2$ and 
			$\textit{val}_{R[j]}(C'_S, \sigma'_S[j]) = 
			\textit{val}_{R[j]}(C'_S,\sigma'_S[j-1]) + 2$.
		Thus $\textit{val}_{R[j]}(C'_S, \sigma'_S[j]) 
			= \textit{val}_{R[j]}(C'_S,\sigma_\alpha[j])$ as wanted.
		
		Finally, consider the case where $\POP$ is CAS.
		Recall that $v_j = \textit{val}_{R[j]}(\sigmaHighA[S_H] \circ \sigma_\alpha[j-1])$.
		By \autoref{thm:HighAInvars},	$\textit{val}_{R[j]}(\sigmaHighA[S_H]) 
			= \textit{val}_{R[j]}(\sigmaHighA[S])$.
		Then, since $\sigma_\alpha[j-1]$ contains at most one step by each process,
			observe that $\textit{val}_{R[j]}(\sigmaHighA[S_H] \circ \sigma_\alpha[j-1]) 
			= \textit{val}_{R[j]}(\sigmaHighA[S] \circ \sigma_\alpha[j-1])$.
		By definition, since $\sigmaHighB[S] \neq \bot$,
			$\sigmaHighB[S] = \sigmaHighA[S]$.
		Thus $v_j = \textit{val}_{R[j]}(\sigmaHighB[S] \circ \sigma_\alpha[j-1])$.
		Then, since for every register $R \in \mathcal{R}$,
			$\textit{val}_R(C'_S, \sigma'_S[j-1]) = \textit{val}_R(C'_S,\sigma_\alpha[j-1])$,
			$v_j = \textit{val}_{R[j]}(\sigmaHighB[S] \circ \sigma'_S[j-1])$ too.

		By definition, either $\alpha_1[j]$ is poised to perform
			a $\textsc{CAS}(v_j,v')$ operation where $v' \neq v_j$
			at the end of $E(\sigmaHighA[S_H])$,
			or no process in $G[j]$ is poised to perform
			a $\textsc{CAS}(v_j,v')$ operation where $v' \neq v_j$
			at the end of $E(\sigmaHighA[S_H])$.
		By \autoref{thm:HighAInvars},
			$\sigmaHighA[0..2^n-1]$ is $(i-1)$-compliant,
			with $\SMax(\sigmaHighA[0..2^n-1]) = S_H$.
		So by Invariant~\ref{RME-invar:subsetStatesSame2},
			since $\sigmaHighB[S] = \sigmaHighA[S] \neq \bot$
			and $\{\alpha_1[j], \alpha_2[j], \beta_1[j]\} \subseteq S$,
			for each process $p \in \{\alpha_1[j], \alpha_2[j], \beta_1[j]\}$,
			$\textit{state}_p(\sigmaHighA[S_H]) = \textit{state}_p(\sigmaHighA[S]) 
			= \textit{state}_p(\sigmaHighB[S])$.
			
		Now recall that for each process 
			$p \in \{\alpha_1[j], \alpha_2[j], \beta_1[j]\}$,
			$\textit{state}_p(C'_S,\sigma'_S[j-1]) = \textit{state}_p(\sigmaHighB[S])
			= \textit{state}_p(C'_S,\sigma_\alpha[j-1])$.
		So either $\alpha_1[j]$ is poised to perform
			a $\textsc{CAS}(v_j,v')$ operation where $v' \neq v_j$
			at the end of both 
			$E(C'_S,\sigma'_S[j-1])$ and $E(C'_S,\sigma_\alpha[j-1])$,
			or no process in $\{\alpha_1[j], \alpha_2[j], \beta_1[j]\}$ 
			is poised to perform
			a $\textsc{CAS}(v_j,v')$ operation where $v' \neq v_j$
			at the end of both 
			$E(C'_S,\sigma'_S[j-1])$ and $E(C'_S,\sigma_\alpha[j-1])$.
		 
		In the latter case, $\textit{val}_{R[j]}(C'_S, \sigma_\alpha[j]) = 
			\textit{val}_{R[j]}(C'_S,\sigma'_\alpha[j-1]) = v_j$ and 
			$\textit{val}_{R[j]}(C'_S, \sigma'_S[j]) = 
			\textit{val}_{R[j]}(C'_S,\sigma'_S[j-1]) = v_j$.
		Thus $\textit{val}_{R[j]}(C'_S, \sigma'_S[j]) 
			= \textit{val}_{R[j]}(C'_S,\sigma_\alpha[j])$, as wanted.
		
		In the former case, since $\beta_1[j] \in S$,
			$\beta_1[j]$ is poised to perform 
			a $\textsc{CAS}(v_\beta,v'_\beta)$ operation on $R[j]$ 
			at the end of $E(\sigmaHighB[S])$.
		Recall that for each process 
			$p \in \{\alpha_1[j], \alpha_2[j], \beta_1[j]\}$,
			$\textit{state}_p(C'_S,\sigma'_S[j-1]) = \textit{state}_p(\sigmaHighB[S])
			= \textit{state}_p(C'_S,\sigma_\alpha[j-1])$.
		So $\beta_1[j]$ is also poised to perform 
			a $\textsc{CAS}(v_\beta,v'_\beta)$ operation on $R[j]$ 
			at the end of both 
			$E(C'_S,\sigma'_S[j-1])$ and $E(C'_S,\sigma_\alpha[j-1])$.
			
		By definition, since $\POP$ is CAS, if $v_j \neq v_\beta$, then
			$\beta[j] = \beta_1[j] \circ \alpha_1[j] \circ \alpha_2[j]$
			otherwise $\beta[j] = \alpha_1[j] \circ \beta_1[j] \circ \alpha_2[j]$.
		So if $v_j \neq v_\beta$, then in $E(C'_S,\sigma'_S[j])$,
			$\beta_1[j]$ performs an unsuccessful $\textsc{CAS}(v_\beta,v'_\beta)$ 
			on $R[j]$ when $R[j]$ contains $v_j \neq v_\beta$.
		Otherwise $v_j = v_\beta$, so in $E(C'_S,\sigma'_S[j])$,
			$\beta_1[j]$ performs an unsuccessful $\textsc{CAS}(v_j,v'_\beta)$ 
			on $R[j]$ immediately after $\alpha_1[j]$
			successfully changes the value of $R[j]$
			from $v_j$ to $v' \neq v_j$.
		Thus regardless of whether $v_j = v_\beta$, 
			the $\textsc{CAS}(v_\beta, v'_\beta)$ operation by $\beta_1[j]$
			does not change the value of $R[j]$ during $E(C'_S, \sigma'_S[j])$.
		Consequently, $\textit{val}_{R[j]}(C'_S, \sigma'_S[j]) 
			= \textit{val}_{R[j]}(C'_S,\sigma_\alpha[j])$, as wanted. 
	\end{proof} 
}
   
\begin{lem}
	\label{thm:HighPreInvars} 
	For every set $S \subseteq \mathcal{P}$ such that 
		$F(\sigmaHighB[S_B]) \cup S_\alpha \subseteq S \subseteq S_B$:
	\begin{enumerate}[label=(H\arabic*)] 
		\item For every process $p \in \mathcal{P}$,
			if $p$ accesses a register $R$ during $E(C'_S,\sigma'_S)$,
			then the owner of $R$ is not in $S_B \setminus F(\sigmaHighB[S_B])$.
			\label{HPI:DSMOwnerIsGone}
		
		\item For every integer $j \in \{0, 1, \ldots, h-1\}$,
			$\textit{val}_{R[j]}(\sigmaHighB[S]) = \textit{val}_{R[j]}(\sigmaHighB[S_B])$. 
			\label{HPI:BRegsWereSimilar} 
			
		\item In the CC model, 
			for every process $p \in S \setminus F(\sigmaHighB[S_B])$,
			the set of registers that $p$ has valid cache copies of
			at the end of $E(\sigmaHighB[S] \circ \sigma'_S)$
			is exactly the same as at the end of 
			$E(\sigmaHighB[S_B] \circ \sigma'_{S_B})$.
			\label{HPI:ValidCCSame} 
			
		\item For every integer $j \in \{0, 1, \ldots, h-1\}$,
			$\textit{val}_{R[j]}(C'_S,\sigma'_S) = 
			\textit{val}_{R[j]}(C'_{S_B},\sigma'_{S_B})$.
			\label{HPI:BRegsSimilar}
			
		\item For each process $p \in S \setminus S_\alpha$,
			$\textit{state}_p(C'_{S_B},\sigma'_{S_B}) =
			\textit{state}_p(C'_S,\sigma'_S)$.
			\label{HPI:statesSame}  
		
		\item Each process in $S_\beta \cap S$
			incurs exactly one RMR
			during $E(C'_S,\sigma'_S)$.
			\label{HPI:OneRMREach} 
		
		\item For each process $p \in S \setminus F(\sigmaHighB[S_B])$,
			$p$ has not left the critical section during $E(\sigmaHighB[S] \circ \sigma'_S)$.
			\label{HPI:notLeftCS} 
			
		\item $F(\sigmaHighB[S]) = F(\sigmaHighB[S] \circ \sigma'_S)$.
			\label{HPI:FUnchanged}
	\end{enumerate}
\end{lem}

\begin{proof}
	First, recall that for every set $S \subseteq \mathcal{P}$ such that 
		$F(\sigmaHighB[S_B]) \cup S_\alpha \subseteq S \subseteq S_B$, 
		$\sigmaHighB[S] \neq \bot$.
	So let $S \subseteq \mathcal{P}$ be a set of processes such that
		$F(\sigmaHighB[S_B]) \cup S_\alpha \subseteq S \subseteq S_B$. 
	 
	By \autoref{thm:ExactlyGroupRegsAccessed},
		the set of registers accessed during $E(C'_S,\sigma'_S)$
		is exactly the set of registers in $R[0..h-1]$.
	By \autoref{thm:HighBInvars},
		for every integer $j \in \{0, 1, \ldots, h-1\}$:
	\begin{itemize}
		\item $\textit{val}_{R[j]}(\sigmaHighB[S]) =
			\textit{val}_{R[j]}(\sigmaHighB[S_B])$
			(\ref{HPI:BRegsWereSimilar}).
			
		\item The owner of $R[j]$ is not in $S_B \setminus F(\sigmaHighB[S_B])$.
	\end{itemize}
	Thus for every process $p \in \mathcal{P}$,
		if $p$ accesses a register $R$ during $E(C'_S,\sigma'_S)$,
		then the owner of $R$ is not in $S_B \setminus F(\sigmaHighB[S_B])$
		(\ref{HPI:DSMOwnerIsGone}). 
	
	By \autoref{thm:HighBInvars},
		$\sigmaHighB[0..2^n-1]$ is $(i-1)$-compliant
		with $\SMax(\sigmaHighB[0..2^n-1]) = S_B$.
	So by Invariant~\ref{RME-invar:ValidCCSame},
		for every process $p \in S \setminus F(\sigmaHighB[S_B])$,
		the set of registers that $p$ has valid cache copies of
		at the end of $E(\sigmaHighB[S])$ is exactly the same as
		at the end of $E(\sigmaHighB[S_B])$. 
		
	Suppose $\POP$ is read.
	Then by \autoref{thm:HighBInvars} and Invariant~\ref{RME-invar:subsetStatesSame2},
		for every process $p \in S \setminus F(\sigmaHighB[S_B])$ 
		$\textit{state}_p(\sigmaHighB[S]) = \textit{state}_p(\sigmaHighB[S_B])$.
	So at the end of both $E(\sigmaHighB[S])$ and $E(\sigmaHighB[S_B])$,
		$p \in S \setminus F(\sigmaHighB[S_B])$ is poised to perform 
		a read operation on the same register $R$.
	Then, since $\sigma'_S$ contains exactly one non-crash step
		of each process in $S \cap S_\beta = S \setminus F(\sigmaHighB[S_B])$,
		the set of registers that $p$ has valid cache copies of
		at the end of $E(\sigmaHighB[S] \circ \sigma'_S)$
		is the union of $\{R\}$ and the set of registers that 
		$p$ has valid cache copies of
		at the end of $E(\sigmaHighB[S])$.
	Furthermore, since $\sigma'_{S_B}$ contains exactly one non-crash step
		of each process in $S_B \cap S_\beta = S_B \setminus F(\sigmaHighB[S_B])$,
		the set of registers that $p$ has valid cache copies of
		at the end of $E(\sigmaHighB[S_B] \circ \sigma'_{S_B})$
		is the union of $\{R\}$ and the set of registers that 
		$p$ has valid cache copies of
		at the end of $E(\sigmaHighB[S_B])$.
	Thus the set of registers that $p$ has valid cache copies of
		at the end of $E(\sigmaHighB[S] \circ \sigma'_S)$ is exactly the same as
		at the end of $E(\sigmaHighB[S_B] \circ \sigma'_{S_B})$. 
		
	Now suppose instead that $\POP$ is not read.
	Then by \autoref{thm:ExactlyGroupRegsAccessed},
		the set of registers accessed during $E(C'_S,\sigma'_S)$
		is exactly the set of registers in $R[0..h-1]$
		and exactly the set of registers accessed during $E(C'_{S_B},\sigma'_{S_B})$.
	Thus, since every non-read operation invalidates all cache copies on a register,
		all cache copies of registers in $R[0..h-1]$
		are invalidated during both $E(C'_S,\sigma'_S)$ and $E(C'_{S_B},\sigma'_{S_B})$.
	So for every process $p \in S \setminus F(\sigmaHighB[S_B])$,
		the set of registers that $p$ has valid cache copies of
		at the end of $E(\sigmaHighB[S] \circ \sigma'_S)$ is exactly the same as
		at the end of $E(\sigmaHighB[S_B] \circ \sigma'_{S_B})$.
	
	Consequently, regardless of $\POP$, in the CC model, 
		for every process $p \in S \setminus F(\sigmaHighB[S_B])$,
		the set of registers that $p$ has valid cache copies of
		at the end of $E(\sigmaHighB[S] \circ \sigma'_S)$
		is exactly the same as at the end of 
		$E(\sigmaHighB[S_B] \circ \sigma'_{S_B})$
		(\ref{HPI:ValidCCSame}).
		 
	Now consider $E(C'_S,\sigma_\alpha)$ and 
		$E(C'_{S_B},\sigma_\alpha)$.
	By construction, $\sigma_\alpha$ contains exactly one non-crash step
		of every process in $S_\alpha$ and no other steps.
	By \autoref{thm:HighBInvars},
		$\sigmaHighB[0..2^n-1]$ is $(i-1)$-compliant
		with $\SMax(\sigmaHighB[0..2^n-1]) = S_B$.
	So by Invariants~\ref{RME-invar:uniqueSMax} and \ref{RME-invar:subsetStatesSame2},
		since $F(\sigmaHighB[S_B]) \cup S_\alpha \subseteq S \subseteq S_B$,
		for every process $p \in S \supset S_\alpha$,
		$\textit{state}_p(\sigmaHighB[S]) = \textit{state}_p(\sigmaHighB[S_B])$.
	By \autoref{thm:ExactlyGroupRegsAccessed},
		the set of registers accessed during $E(C'_{S_B},\sigma_\alpha)$ 
		is exactly the set of registers in $R[0..h-1]$
		and exactly the set of registers accessed during $E(C'_S,\sigma_\alpha)$.
	Since we have already proven \ref{HPI:BRegsWereSimilar},
		for every integer $j \in \{0, 1, \ldots, h-1\}$,
		$\textit{val}_{R[j]}(\sigmaHighB[S]) = \textit{val}_{R[j]}(\sigmaHighB[S_B])$.
	Therefore during both $E(C'_S,\sigma_\alpha)$ and 
		$E(C'_{S_B},\sigma_\alpha)$,
		the same set $S_\alpha$ of processes begin in the same states,
		and perform the same operations in the same order on the same registers
		which also begin with the same values.
	So for every integer $j \in \{0, 1, \ldots, h-1\}$,
		$\textit{val}_{R[j]}(C'_S, \sigma_\alpha) = 
		\textit{val}_{R[j]}(C'_{S_B}, \sigma_\alpha)$.
	Thus by \autoref{thm:AlphaBetaIndistinguishable},
		for every integer $j \in \{0, 1, \ldots, h-1\}$:
	\begin{align*}
		\textit{val}_{R[j]}(C'_S, \sigma'_S)  
			& = \textit{val}_{R[j]}(C'_S, \sigma_\alpha) \\ 
			& = \textit{val}_{R[j]}(C'_{S_B}, \sigma_\alpha) \\
			& = \textit{val}_{R[j]}(C'_{S_B}, \sigma'_{S_B}) 
	\end{align*}
	So for every integer $j \in \{0, 1, \ldots, h-1\}$,
		$\textit{val}_{R[j]}(C'_S,\sigma'_S) = 
		\textit{val}_{R[j]}(C'_{S_B},\sigma'_{S_B})$
		(\ref{HPI:BRegsSimilar}).
	
	Now let $j'$ be an integer in $\{0, 1, \ldots, h-1\}$.
	Then by \autoref{thm:AlphaBetaIndistinguishable}:
	\begin{align*}
		\textit{val}_{R[j']}(C'_S, \sigma'_S[j'-1])  
			& = \textit{val}_{R[j']}(C'_S, \sigma_\alpha[j'-1]) \\ 
			& = \textit{val}_{R[j']}(C'_{S_B}, \sigma_\alpha[j'-1]) \\
			& = \textit{val}_{R[j']}(C'_{S_B}, \sigma'_{S_B}[j'-1]) 
	\end{align*}
	Let $C^{j'-1}_S$ and $C^{j'-1}_{S_B}$ be the configurations 
		at the end of $E(C'_S,\sigma'_S[j'-1])$ 
		and $E(C'_{S_B},\sigma'_{S_B}[j'-1])$ respectively.
	So $R[j']$ has the same value in 
		both $C^{j'-1}_S$ and $C^{j'-1}_{S_B}$.
	In addition, recall that by the definition 
		of $\sigma'_S[j'-1]$ and $\sigma'_{S_B}[j'-1]$,
		processes in $G[j'] \cap S$
		do not take any steps during $E(C'_S,\sigma'_S[j'-1])$ 
		and $E(C'_{S_B},\sigma'_{S_B}[j'-1])$.
	Thus, for every process $p \in G[j'] \cap S \subseteq S$,
		since $\textit{state}_p(\sigmaHighB[S]) = \textit{state}_p(\sigmaHighB[S_B])$,
		the state of $p$ is still the same in
		both $C^{j'-1}_S$ and $C^{j'-1}_{S_B}$, i.e.,
		$p$ is poised to perform the same operation on $R[j']$
		in both $C^{j'-1}_S$ and $C^{j'-1}_{S_B}$.
	
	Now suppose that $\beta_1[j']$ is in $S \subseteq S_B$.
	Then $\sigma'_S[j'] = \sigma'_S[j'-1] \circ \beta[j']$
		and $\sigma'_{S_B}[j'] = \sigma'_{S_B}[j'-1] \circ \beta[j']$.
	Furthermore, in addition to $\alpha_1[j']$ and $\alpha_2[j']$,
		$\beta_1[j']$ is also in $G[j'] \cap S$.
	So $\beta[j']$ contains only processes in $G[j'] \cap S$. 
	Thus during both $E(C^{j'-1}_S,\beta[j'])$ 
		and $E(C^{j'-1}_{S_B},\beta[j'])$,
		the same set of processes (namely those in $\beta[j']$)
		begin in the same set of states and
		perform the same set of operations in the same order
		on $R[j']$ (which begins with the same value),
		and thus must reach the same resulting states.
	Finally, since each process in $G[j'] \cap S$
		takes no more steps in the remainder of $\sigma'_S$ and $\sigma'_{S_B}$,
		$\textit{state}_{\beta_1[j']}(C'_S,\sigma'_S) = 
		\textit{state}_{\beta_1[j']}(C'_{S_B},\sigma'_{S_B})$.
		
	Thus we have shown that for every integer $j \in \{0, 1, \ldots, h-1\}$,
		$\textit{state}_{\beta_1[j]}(C'_S,\sigma'_S) = 
		\textit{state}_{\beta_1[j]}(C'_{S_B},\sigma'_{S_B})$.
	Now note that for each process $p \in S \setminus S_\alpha$,
		since $S \subseteq S_B$,
		there exists an integer $j \in \{0, 1, \ldots, h-1\}$ 
		such that $p = \beta_1[j]$.
	So for each process $p \in S \setminus S_\alpha$,
		$\textit{state}_p(C'_{S_B},\sigma'_{S_B}) =
		\textit{state}_p(C'_S,\sigma'_S)$
		(\ref{HPI:statesSame}). 
	
	Since $\sigmaHighB[S] \neq \bot$, by construction,
		$\sigmaOld[S] \neq \bot$.
	By \autoref{thm:SetupPreInvars} (\ref{SPI:aboutToRMR}),
		for each process $p \in S \setminus F(\sigmaOld[\SOMax])$,
		$p$ incurs an RMR at the end of $E(\sigmaSetupA[S] \circ p)$.
	By construction and \autoref{thm:SetupAInvars} 
		(Invariant~\ref{RME-invar:FSame}),
		$\sigmaHighB[S] = \sigmaSetupB[S] = \sigmaSetupA[S]$,
		and $F(\sigmaHighB[S_B]) = F(\sigmaSetupA[S_I]) = F(\sigmaSetupA[\SOMax])
		= F(\sigmaOld[\SOMax])$.
	Thus for each process $p \in S \setminus F(\sigmaHighB[S_B])$,
		$p$ incurs an RMR at the end of $E(\sigmaHighB[S] \circ p)$.
	By definition, $S_B = S_\beta \cup F(\sigmaHighB[S_B])$,
		so since $S \subseteq S_B$, $S \setminus F(\sigmaHighB[S_B]) = S_\beta \cap S$.
	By construction, $\sigma'_S$ contains exactly one 
		non-crash step of each process in $S_\beta \cap S$ and no other steps.
	So in the DSM model, every process in $S_\beta \cap S$
		is poised to access a register it does not own in $C'_S$,
		and so every process in $S_\beta \cap S$ incurs exactly one RMR
		during $E(C'_S,\sigma'_S)$.
	In the CC model, every process in $S_\beta \cap S$
		is poised to perform a non-read operation
		or read a register that it does not have 
		a valid cache copy of in $C'_S$,
		and so every process in $S_\beta \cap S$ incurs exactly one RMR
		during $E(C'_S,\sigma'_S)$.
	Thus in both the DSM and CC models,
		each process in $S_\beta \cap S$
		incurs exactly one RMR
		during $E(C'_S,\sigma'_S)$ (\ref{HPI:OneRMREach}).
	 
	By \autoref{thm:HighBInvars},
		$\sigmaHighB[0..2^n-1]$ is $(i-1)$-compliant,
		so by Invariant~\ref{RME-invar:CS},
		each process that is not in $F(\sigmaHighB[S])$
		does not enter the critical section during $E(\sigmaHighB[S])$.
	Recall that $\sigma'_S$ consists of exactly one non-crash step
		of each process in $S_\beta \cap S = 
		S \setminus F(\sigmaHighB[S_B])$ and no other steps.
	Thus, with only one step, although a process could enter the critical section
		during $E(C'_S,\sigma'_S)$,
		it cannot have taken any steps within the critical section.
	Thus by Assumption~\ref{RME-assumption:RMR-in-CS}, 
		no process can leave the critical section during $E(C'_S,\sigma'_S)$.
	So for each process $p \in S \setminus F(\sigmaHighB[S_B])$,
		$p$ has not left the critical section during $E(\sigmaHighB[S] \circ \sigma'_S)$
		(\ref{HPI:notLeftCS}).
	Therefore, since no process in $S \setminus F(\sigmaHighB[S_B])$
		has left the critical section during $E(\sigmaHighB[S] \circ \sigma'_S)$,
		no process in $S \setminus F(\sigmaHighB[S_B])$
		has completed its super-passage during $E(\sigmaHighB[S] \circ \sigma'_S)$. 
	Thus $F(\sigmaHighB[S]) = F(\sigmaHighB[S] \circ \sigma'_S)$ (\ref{HPI:FUnchanged}). 
\end{proof}
 
Now recall that $S_F = S_\alpha \cup F(\sigmaHighA[S_H])$
	and that by \autoref{thm:HighAInvars} and Invariant~\ref{RME-invar:FSame},
	$F(\sigmaHighA[S_H]) = F(\sigmaHighA[S_B]) = F(\sigmaHighB[S_B])$.
Thus $S_F = S_\alpha \cup F(\sigmaHighB[S_B])$, and so
	$F(\sigmaHighB[S_B]) \cup S_\alpha \subseteq S_F \subseteq S_B$.
Furthermore, recall that by definition,
	since $S_F = \sigma_\alpha \cup F(\sigmaHighB[S_B])$,
	$\sigma_\alpha = \sigma'_{S_F}$.
Thus $E(C'_{S_F}, \sigma'_{S_F}) = E(C'_{S_F},\sigma_\alpha)$.
	
\begin{lem}
	\label{thm:RFSame}
	For every register $R \in \mathcal{R}_F$,
		and every set $S \subseteq \mathcal{P}$ such that 
		$F(\sigmaHighB[S_B]) \cup S_\alpha \subseteq S \subseteq S_B$,
		$\textit{val}_R(C'_S, \sigma'_S) =
		\textit{val}_R(C'_{S_F}, \sigma_\alpha)$.
\end{lem}

\begin{proof}
	Let $S \subseteq \mathcal{P}$ be a set of processes such that 
		$F(\sigmaHighB[S_B]) \cup S_\alpha \subseteq S \subseteq S_B$.
	By \autoref{thm:HighBInvars} and Invariant~\ref{RME-invar:uniqueSMax},
		since $F(\sigmaHighB[S_B]) \cup S_\alpha \subseteq S_F \subseteq S_B$
		and $F(\sigmaHighB[S_B]) \cup S_\alpha \subseteq S \subseteq S_B$,
		both $\sigmaHighB[S_F] \neq \bot$ and $\sigmaHighB[S] \neq \bot$.
	So by \autoref{thm:HighBInvars},
		for every register $R \in \mathcal{R}_F$,
		$\textit{val}_R(\sigmaHighB[S_F]) 
		= \textit{val}_R(\sigmaHighB[S_B])
		= \textit{val}_R(\sigmaHighB[S])$.
	
	Now consider $E(C'_S,\sigma'_S)$
		and $E(C'_{S_F},\sigma'_{S_F})$.
	By \autoref{thm:ExactlyGroupRegsAccessed},
		only registers in $R[0..h-1]$
		are accessed during both executions.
	So for every register $R \in \mathcal{R}_F$,
		if $R$ is not one of $R[0..h-1]$,
		then $\textit{val}_R(C'_S, \sigma'_S) =
		\textit{val}_R(C'_{S_F}, \sigma'_{S_F})$.	
	
	Next, by \autoref{thm:HighPreInvars} (\ref{HPI:BRegsSimilar}),
		for every integer $j \in \{0, 1, \ldots, h-1\}$,
		$\textit{val}_{R[j]}(C'_S,\sigma'_S) 
		= \textit{val}_{R[j]}(C'_{S_B},\sigma'_{S_B})
		= \textit{val}_{R[j]}(C'_{S_F},\sigma'_{S_F})$.
	Therefore, for every register $R \in \mathcal{R}_F$,
		regardless of whether $R$ is one of $R[0..h-1]$,
		$\textit{val}_R(C'_S, \sigma'_S) =
		\textit{val}_R(C'_{S_F}, \sigma'_{S_F}) =
		\textit{val}_R(C'_{S_F}, \sigma_\alpha)$.	 
\end{proof}
 
Next, recall that there exists a schedule $\sigma_F$ such that:
\begin{itemize}
	\item $\sigma_F$ begins with exactly one crash step of every process in $S_\alpha$,
		and contains no other crash steps.
		
	\item $\sigma_F$ contains only steps of processes in $S_\alpha = 
		S_F \setminus F(\sigmaHighA[S_F])$.
		
	\item During $E(\sigmaHighA[S_F] \circ \sigma_\alpha \circ \sigma_F)$,
		every process in $S_F$ begins and then completes
		a super-passage, i.e., $F(\sigmaHighA[S_F] \circ \sigma_\alpha \circ \sigma_F) = S_F$.
\end{itemize}
By \autoref{thm:HighAInvars} and Invariant~\ref{RME-invar:FSame},
	$F(\sigmaHighA[S_F]) = F(\sigmaHighA[S_B]) = F(\sigmaHighB[S_B])$.
By construction, $\sigmaHighB[S_F] = \sigmaHighA[S_F]$.
Thus:
\begin{itemize}
	\item $\sigma_F$ begins with exactly one crash step of every process in $S_\alpha$,
		and contains no other crash steps.
		
	\item $\sigma_F$ contains only steps of processes in $S_\alpha = 
		S_F \setminus F(\sigmaHighB[S_B])$.
		
	\item During $E(\sigmaHighB[S_F] \circ \sigma_\alpha \circ \sigma_F)$,
		every process in $S_F$ begins and then completes
		a super-passage, i.e., $F(\sigmaHighB[S_F] \circ \sigma_\alpha \circ \sigma_F) = S_F$.
\end{itemize}
Further recall that by definition,
	$C_F$ is the configuration at the end of 
	$E(\sigmaHighA[S_F] \circ \sigma_\alpha)$, and
	$\mathcal{R}_F$ is the set of every register that is
	accessed during $E(C_F,\sigma_F)$ 
	(after the crash steps of every process in $S_\alpha$ 
	at the beginning of $\sigma_F$).
	
Now for every set $S \subseteq \mathcal{P}$ such that 
	$F(\sigmaHighB[S_B]) \cup S_\alpha \subseteq S \subseteq S_B$,
	let $C''_S$ be the configuration at the end of $E(\sigmaHighB[S] \circ \sigma'_S)$.
Then $E(\sigmaHighA[S_F] \circ \sigma_\alpha) 
	= E(\sigmaHighB[S_F] \circ \sigma'_{S_F})$,
	so $C_F = C''_{S_F}$.

\begin{lem}
	\label{thm:sigmaFApplicable}
	For every set $S \subseteq \mathcal{P}$ such that 
		$F(\sigmaHighB[S_B]) \cup S_\alpha \subseteq S \subseteq S_B$,
		during both $E(C''_S,\sigma_F)$ and $E(C''_{S_F},\sigma_F) = E(C_F,\sigma_F)$,
		the same set of processes (namely $S_\alpha$) crash,
		then perform the same operations in the same order
		on the same set of registers (namely $\mathcal{R}_F$)
		and so must reach the same resulting states. 
\end{lem} 

\toOmit{
	\begin{proof}
		By \autoref{thm:RFSame},
			for every register $R \in \mathcal{R}_F$
			and every set $S \subseteq \mathcal{P}$ such that 
			$F(\sigmaHighB[S_B]) \cup S_\alpha \subseteq S \subseteq S_B$,
			$R$ has the same value in $C_F = C''_{S_F}$ as in $C''_S$. 
		Furthermore, by the definition of $\sigma_F$,
			$\sigma_F$ begins with a crash step of 
			every process in $P(\sigma_F) = S_\alpha$.
		The lemma immediately follows.
	\end{proof} 
}

We now construct a new array $\sigmaHighC[0..2^n-1]$ such that
	for every set $S \subseteq \mathcal{P}$,
	if $F(\sigmaHighB[S_B]) \cup S_\alpha \subseteq S \subseteq S_B$,
	then $\sigmaHighC[S] = \sigmaHighB[S] \circ \sigma'_S \circ \sigma_F$;
	otherwise $\sigmaHighC[S] = \bot$.
	
\begin{lem}	
	\label{thm:HighCInvars}
	This new array $\sigmaHighC[0..2^n-1]$ is $i$-compliant
		with $\SMax(\sigmaHighC[0..2^n-1]) = S_B$.
	Furthermore, $F(\sigmaHighC[S_B]) = F(\sigmaHighB[S_B]) \cup S_\alpha$.
\end{lem}

\begin{proof} 
	For every set $S \subseteq \mathcal{P}$,
		if $\sigmaHighC[S] \neq \bot$,
		then by construction, $\sigmaHighC[S] = 
		\sigmaHighB[S] \circ \sigma'_S \circ \sigma_F$.
	By \autoref{thm:HighBInvars},
		$\sigmaHighB[0..2^n-1]$ is $(i-1)$-compliant, 
		so by Invariant~\ref{RME-invar:PsubsetS},
		$P(\sigmaHighB[S]) \subseteq S$.
	By the definition of $\sigma'_S$, $\sigma'_S$ 
		contains only steps of processes in $S \cap S_\beta$.
	By the definition of $\sigma_F$, $\sigma_F$
		contains only steps of processes in $S_\alpha \subseteq S$.
	Thus $P(\sigmaHighC[S]) \subseteq S$ (Invariant~\ref{RME-invar:PsubsetS}).
	
	Now for every set $S \subseteq \mathcal{P}$
		such that $F(\sigmaHighB[S_B]) \cup S_\alpha \subseteq S \subseteq S_B$,
		consider $F(\sigmaHighC[S]) = 
		F(\sigmaHighB[S] \circ \sigma'_{S} \circ \sigma_F)$.
	By \autoref{thm:HighPreInvars} (\ref{HPI:FUnchanged}),
		$F(\sigmaHighB[S]) = F(\sigmaHighB[S] \circ \sigma'_{S})
		= F(C''_{S}, \sigma'_{S})$.
	By definition, every process in $S_\alpha$ completes its super-passage
		during $E(C_F,\sigma_F)$.
	So by \autoref{thm:sigmaFApplicable},
		every process in $S_\alpha$ also completes its super-passage
		during $E(C''_{S},\sigma_F)$.
	Thus $F(\sigmaHighC[S]) = F(\sigmaHighB[S]) \cup S_\alpha$.
	Therefore, $F(\sigmaHighC[S_B]) = F(\sigmaHighB[S_B]) \cup S_\alpha$.
	
	Furthermore, by \autoref{thm:HighBInvars} and 
		Invariants~\ref{RME-invar:uniqueSMax} and~\ref{RME-invar:FSame},
		for every set $S \subseteq \mathcal{P}$
		such that $F(\sigmaHighB[S_B]) \cup S_\alpha \subseteq S \subseteq S_B$,
		$F(\sigmaHighB[S]) = F(\sigmaHighB[S_B])$,
		and so $F(\sigmaHighC[S]) = F(\sigmaHighB[S]) \cup S_\alpha
		= F(\sigmaHighB[S_B]) \cup S_\alpha = F(\sigmaHighC[S_B])$.
	
	By construction, for every set $S \subseteq \mathcal{P}$,
		if $F(\sigmaHighB[S_B]) \cup S_\alpha \subseteq S \subseteq S_B$,
		then $\sigmaHighC[S] = \sigmaHighB[S] \circ \sigma'_S \circ \sigma_F$;
		otherwise $\sigmaHighC[S] = \bot$.
	By \autoref{thm:HighBInvars},
		$\sigmaHighB[0..2^n-1]$ is $(i-1)$-compliant
		with $\SMax(\sigmaHighB[0..2^n-1]) = S_B$.
	So by Invariant~\ref{RME-invar:uniqueSMax},
		for every set $S \subseteq \mathcal{P}$,
		$\sigmaHighB[S] \neq \bot$ if and only if
		$F(\sigmaHighB[S_B]) \subseteq S \subseteq S_B$.
	Thus for every set $S \subseteq \mathcal{P}$,
		if $F(\sigmaHighB[S_B]) \cup S_\alpha \subseteq S \subseteq S_B$,
		then $\sigmaHighB[S] \neq \bot$ and
		$\sigmaHighC[S] = \sigmaHighB[S] \circ \sigma'_S \circ \sigma_F \neq \bot$.
	Then, since we have already proven that 
		$F(\sigmaHighC[S_B]) = F(\sigmaHighB[S_B]) \cup S_\alpha$,
		for every set $S \subseteq \mathcal{P}$,
		$\sigmaHighC[S] \neq \bot$ if and only if
		$F(\sigmaHighC[S_B]) \subseteq S \subseteq S_B$
		(Invariant~\ref{RME-invar:uniqueSMax}).
	
	Furthermore, we have already shown that 
		for every set $S \subseteq \mathcal{P}$
		such that $F(\sigmaHighB[S_B]) \cup S_\alpha \subseteq S \subseteq S_B$,
		$F(\sigmaHighC[S]) = F(\sigmaHighC[S_B])$.
	Thus, since we just proved that Invariant~\ref{RME-invar:uniqueSMax}
		holds for $\sigmaHighC[0..2^n-1]$ with
		$\SMax(\sigmaHighC[0..2^n-1]) = S_B$
		and $F(\sigmaHighC[S_B]) = F(\sigmaHighB[S_B]) \cup S_\alpha$,
		it follows that for every set $S \subseteq \mathcal{P}$
		such that $\sigmaHighC[S] \neq \bot$,
		$F(\sigmaHighC[S]) = F(\sigmaHighC[S_B])$
		(Invariant~\ref{RME-invar:FSame}).
		 
	By \autoref{thm:HighPreInvars} (\ref{HPI:statesSame}),
		for every set $S \subseteq \mathcal{P}$ such that 
		$F(\sigmaHighB[S_B]) \cup S_\alpha =
		F(\sigmaHighC[S_B]) \subseteq S \subseteq S_B$,
		for every process $p \in S \setminus S_\alpha$,
		$\textit{state}_p(C'_{S_B},\sigma'_{S_B}) 
		= \textit{state}_p(C'_S,\sigma'_S)$.
	Since $\sigma_F$ only contains steps of processes in $S_\alpha$,
		for every set $S \subseteq \mathcal{P}$ such that 
		$F(\sigmaHighB[S_B]) \cup S_\alpha =
		F(\sigmaHighC[S_B]) \subseteq S \subseteq S_B$,
		for every process $p \in S \setminus S_\alpha$,
		$\textit{state}_p(C'_{S_B},\sigma'_{S_B} \circ \sigma_F) 
		= \textit{state}_p(C'_S,\sigma'_S \circ \sigma_F)$.
	Thus for every set $S \subseteq \mathcal{P}$ such that 
		$F(\sigmaHighC[S_B]) \subseteq S \subseteq S_B$,
		for every process $p \in S \setminus S_\alpha$,
		$\textit{state}_p(\sigmaHighC[S_B]) 
		= \textit{state}_p(\sigmaHighC[S])$.
	Furthermore, we have already proven that 
		$F(\sigmaHighC[S_B]) = F(\sigmaHighB[S_B]) \cup S_\alpha$,
		and that Invariants~\ref{RME-invar:uniqueSMax} and~\ref{RME-invar:FSame} 
		hold for $\sigmaHighC[0..2^n-1]$ with $\SMax(\sigmaHighC[0..2^n-1]) = S_B$,
		for every set $S \subseteq \mathcal{P}$ such that 
		$F(\sigmaHighC[S_B]) \subseteq S \subseteq S_B$,
		$S_\alpha$ is in both $F(\sigmaHighC[S])$ and $F(\sigmaHighC[S_B])$, 
		so for every process $p \in S_\alpha$,
		$\textit{state}_p(\sigmaHighC[S_B]) 
		= \textit{state}_p(\sigmaHighC[S])$.
	Consequently, for every set $S \subseteq \mathcal{P}$ such that 
		$\sigmaHighC[S] \neq \bot$,
		for every process $p \in S$,
		regardless of whether $p$ is in $S_\alpha$,
		$\textit{state}_p(\sigmaHighC[S_B]) 
		= \textit{state}_p(\sigmaHighC[S])$
		(Invariant~\ref{RME-invar:subsetStatesSame2}).
		  
	Next, by \autoref{thm:sigmaFApplicable},
		for every set $S \subseteq \mathcal{P}$ such that
		$F(\sigmaHighC[S_B]) \subseteq S \subseteq S_B$,
		and every register $R \in \mathcal{R}_F$,
		$\textit{val}_R(\sigmaHighC[S]) =
		\textit{val}_R(\sigmaHighC[S_F])$.
	Furthermore, only registers in $\mathcal{R}_F$
		are accessed during $E(C''_S,\sigma_F)$.
	So for every integer $j \in \{0,1, \ldots,h-1\}$,
		if $R[j] \not\in \mathcal{R}_F$,
		then by \autoref{thm:HighPreInvars} (\ref{HPI:BRegsSimilar}),
		$\textit{val}_{R[j]}(C'_S,\sigma'_S \circ \sigma_F) =
		\textit{val}_{R[j]}(C'_{S_B},\sigma'_{S_B} \circ \sigma_F)$,
		i.e., $\textit{val}_{R[j]}(\sigmaHighC[S]) =
		\textit{val}_{R[j]}(\sigmaHighC[S_B])$.
	Thus we have shown that for every register $R \in \mathcal{R}$
		such that either $R \in \mathcal{R}_F$ or
		$R$ is one of $R[0..h-1]$,
		and every set $S \subseteq \mathcal{P}$ such that
		$F(\sigmaHighC[S_B]) \subseteq S \subseteq S_B$,
		$\textit{val}_R(\sigmaHighC[S]) =
		\textit{val}_R(\sigmaHighC[S_B])$.
	Then since we have already proven that Invariant~\ref{RME-invar:uniqueSMax}
		holds for $\sigmaHighC[0..2^n-1]$ with $\SMax(\sigmaHighC[0..2^n-1]) = S_B$,
		for every register $R \in \mathcal{R}$
		such that either $R \in \mathcal{R}_F$ or
		$R$ is one of $R[0..h-1]$,
		and every set $S \subseteq \mathcal{P}$ such that
		$\sigmaHighC[S] \neq \bot$,
		$\textit{val}_R(\sigmaHighC[S]) =
		\textit{val}_R(\sigmaHighC[S_B])$.
	So for every register $R \in \mathcal{R}$
		such that either $R \in \mathcal{R}_F$ or
		$R$ is one of $R[0..h-1]$,
		regardless of $\textit{last}_R(\sigmaHighC[S_B])$,
		if $y_R = \textit{val}_R(\sigmaHighC[S_B])$, 
		then we have that for every set $S \subseteq \mathcal{P}$,
		if $\sigmaHighC[S] \neq \bot$, then:
	\begin{displaymath}
		\textit{val}_R\bparen{\sigmaHighC[S]}=
		\begin{cases}
				\textit{val}_R(\sigmaHighC[S_B]) & 
					\text{if $\textit{last}_R(\sigmaHighC[S_B]) \in S$} \\
				y_R & \text{otherwise}
		\end{cases}
	\end{displaymath}
	
	Now consider the registers that are not in $\mathcal{R}_F$
		and not one of $R[0..h-1]$, i.e.,
		the registers that are not accessed during
		$E(C'_S,\sigma'_S \circ \sigma_F)$
		for every set $S \subseteq \mathcal{P}$ such that
		$F(\sigmaHighC[S_B]) \subseteq S \subseteq S_B$.
	By \autoref{thm:HighBInvars},
		$\sigmaHighB[0..2^n-1]$ is $(i-1)$-compliant
		with $\SMax(\sigmaHighB[0..2^n-1]) = S_B$.
	So by Invariant~\ref{RME-invar:subsetRegsSimilar},
		for each such register $R$,
		there exists a value $y_R$ such that 
		for every set $S \subseteq \mathcal{P}$,
		if $\sigmaHighB[S] \neq \bot$, then:
	\begin{displaymath}
		\textit{val}_R\bparen{\sigmaHighB[S]}=
		\begin{cases}
				\textit{val}_R(\sigmaHighB[S_B]) & 
					\text{if $\textit{last}_R(\sigmaHighB[S_B]) \in S$} \\
				y_R & \text{otherwise}
		\end{cases}
	\end{displaymath}
	Then, since each such register $R$ is not accessed during 
		$E(C'_S,\sigma'_S \circ \sigma_F)$
		for every set $S \subseteq \mathcal{P}$ such that
		$F(\sigmaHighC[S_B]) \subseteq S \subseteq S_B$,
		and we have already proven that Invariant~\ref{RME-invar:uniqueSMax}
		holds for $\sigmaHighC[0..2^n-1]$ with $\SMax(\sigmaHighC[0..2^n-1]) = S_B$:
	\begin{displaymath}
		\textit{val}_R\bparen{\sigmaHighC[S]}=
		\begin{cases}
				\textit{val}_R(\sigmaHighC[S_B]) & 
					\text{if $\textit{last}_R(\sigmaHighC[S_B]) \in S$} \\
				y_R & \text{otherwise}
		\end{cases}
	\end{displaymath}
	
	Consequently, for every register $R \in \mathcal{R}$,
		regardless of whether $R \in \mathcal{R}_F$
		and whether $R$ is one of $R[0..h-1]$,
		there is a value $y_R$ such that
		for every set $S \subseteq \mathcal{P}$,
		if $\sigmaHighC[S] \neq \bot$, then:
	\begin{displaymath}
		\textit{val}_R\bparen{\sigmaHighC[S]}=
		\begin{cases}
				\textit{val}_R(\sigmaHighC[S_B]) & 
					\text{if $\textit{last}_R(\sigmaHighC[S_B]) \in S$} \\
				y_R & \text{otherwise}
		\end{cases}
	\end{displaymath}
	So Invariant~\ref{RME-invar:subsetRegsSimilar} holds
		for $\sigmaHighC[0..2^n-1]$.
	
	By \autoref{thm:HighBInvars},
		$\sigmaHighB[0..2^n-1]$ is $(i-1)$-compliant.
	So by Invariant~\ref{RME-invar:crashes},
		for every set $S \subseteq \mathcal{P}$ with $\sigmaHighB[S] \neq \bot$,
		during $E(\sigmaHighB[S])$, each process crashes at most once,
		and each process that is not in $F(\sigmaHighB[S])$ never crashes.
		
	By construction, for every set $S \subseteq \mathcal{P}$ 
		with $\sigmaHighC[S] \neq \bot$,
		$\sigma'_S$ does not contain any crash steps.
	Furthermore, $\sigma_F$ contains exactly one crash step 
		for each process in $S_\alpha$
		and no other crash steps.
	
	We have already proven that $F(\sigmaHighC[S_B]) = F(\sigmaHighB[S_B]) \cup S_\alpha$
		and that Invariant~\ref{RME-invar:FSame} holds for $\sigmaHighC[0..2^n-1]$.
	By \autoref{thm:HighBInvars},
		Invariant~\ref{RME-invar:FSame} also holds for $\sigmaHighB[0..2^n-1]$.
	So for every set $S \subseteq \mathcal{P}$ 
		with $\sigmaHighC[S] \neq \bot$,
		$F(\sigmaHighC[S]) = F(\sigmaHighB[S]) \cup S_\alpha$.
	
	So for every process $p \not\in F(\sigmaHighC[S_B])$,
		$p$ never crashes during $E(\sigmaHighC[S])$.
	Furthermore, since $S_\alpha \not\in F(\sigmaHighB[S])$,
		processes in $S_\alpha$ never crash during $E(\sigmaHighB[S])$,
		and thus crash at most once during $E(\sigmaHighC[S]) = 
		E(\sigmaHighB[S] \circ \sigma'_S \circ \sigma_F)$.
	Finally, processes in $F(\sigmaHighB[S])$
		do not have any (crash) steps in either $\sigma'_S$ or $\sigma_F$,
		so they also still crash at most once during $E(\sigmaHighC[S]) = 
		E(\sigmaHighB[S] \circ \sigma'_S \circ \sigma_F)$.
	Consequently, for every set $S \subseteq \mathcal{P}$ with $\sigmaHighC[S] \neq \bot$,
		during $E(\sigmaHighC[S])$, each process crashes at most once,
		and each process that is not in $F(\sigmaHighC[S])$ never crashes 
		(Invariant~\ref{RME-invar:crashes}).
	
	Now suppose, for contradiction, that 
		for some set $S \subseteq \mathcal{P}$ 
		with $\sigmaHighC[S] \neq \bot$,
		some process $p$ that is not in $F(\sigmaHighC[S])$
		enters the critical section during $E(\sigmaHighC[S]) = 
		E(\sigmaHighB[S] \circ \sigma'_S \circ \sigma_F)$.
	Since we have already proven that Invariant~\ref{RME-invar:PsubsetS}
		holds for $\sigmaHighC[0..2^n-1]$,
		$p \in S$.
	Furthermore, we have also already shown that 
		$F(\sigmaHighC[S]) = F(\sigmaHighB[S]) \cup S_\alpha$,
		so $p \not\in F(\sigmaHighB[S]) \cup S_\alpha$.
	 
	By \autoref{thm:HighBInvars}, $\sigmaHighB[0..2^n-1]$ is $(i-1)$-compliant.
	So by Invariant~\ref{RME-invar:CS},
		since $p \not\in F(\sigmaHighB[S])$,
		$p$ does not enter the critical section during $E(\sigmaHighB[S])$.
	Furthermore, by the definition of $\sigma_F$, 
		since $p \not\in S_\alpha$,
		$\sigma_F$ contains no steps of $p$.
	Thus $p$ must be one of the processes that take a step during $E(C'_S,\sigma'_S)$,
		and must enter the critical section with this one step.
	By Assumption~\ref{RME-assumption:RMR-in-CS},
		a process that enters the critical section cannot leave the critical section
		before incurring an RMR within the critical section.
	Thus, since $\sigma_F$ contains no steps of $p \not\in S_\alpha$,
		$p$ remains in the critical section throughout $E(C''_S,\sigma_F)$.
	
	Now consider the processes in $S_\alpha$.
	Since $\sigmaHighB[0..2^n-1]$ is $(i-1)$-compliant,
		and $S_\alpha \cap F(\sigmaHighB[S]) = \varnothing$,
		the processes in $S_\alpha$ do not enter the critical section during $E(\sigmaHighB[S])$.
	Since $p$ is in the critical section in $C''_S$,
		and $\sigma'_S$ contains at most one step of each process,
		to avoid violating mutual exclusion,
		each process in $S_\alpha$ cannot enter the critical section 
		with its at most one step taken during $E(C'_S,\sigma'_S)$
		by Assumption~\ref{RME-assumption:RMR-in-CS}.
	Thus during $E(\sigmaHighB[S] \circ \sigma'_S)$,
		the processes in $S_\alpha$ do not enter the critical section.
	Furthermore, since $p$ remains in the critical section throughout $E(C''_S,\sigma_F)$,
		to avoid violating mutual exclusion,
		the processes in $S_\alpha$ also do not enter the critical section
		during $E(\sigmaHighB[S] \circ \sigma'_S \circ \sigma_F)$.
	
	However, we have already shown that 
		$F(\sigmaHighC[S]) = F(\sigmaHighB[S]) \cup S_\alpha$.
	Thus we have that during $E(\sigmaHighC[S]) = 
		E(\sigmaHighB[S] \circ \sigma'_S \circ \sigma_F)$,
		every process in $S_\alpha$ completes its super-passage
		without entering the critical section --- a contradiction.
	Consequently, for every set $S \subseteq \mathcal{P}$ 
		with $\sigmaHighC[S] \neq \bot$,
		each process that is not in $F(\sigmaHighC[S])$ 
		does not enter the critical section during $E(\sigmaHighC[S])$
		(Invariant~\ref{RME-invar:CS}).
		
	Next, by \autoref{thm:HighBInvars},
		for every register $R \in \mathcal{R}$
		such that either $R$ is one of $R[0..h-1]$
		or $R \in \mathcal{R}_F$,
		the owner of $R$ is not in 
		$S_B \setminus F(\sigmaHighB[S_B])$.
	By \autoref{thm:ExactlyGroupRegsAccessed},
		the definition of $\mathcal{R}_F$,
		and \autoref{thm:sigmaFApplicable},
		for every set $S \subseteq \mathcal{P}$ such that
		$\sigmaHighC[S] \neq \bot$,
		each register $R \in \mathcal{R}$ is only accessed during 
		$E(\sigmaHighC[S]) = E(\sigmaHighB[S] \circ \sigma'_S \circ \sigma_F)$
		if either $R$ is one of $R[0..h-1]$ or $R \in \mathcal{R}_F$.
	Therefore, for every register $R \in \mathcal{R}$,
		if $R$ is accessed during $E(C'_S,\sigma'_S \circ \sigma_F)$,
		then the owner of $R$ is not in $S_B \setminus F(\sigmaHighB[S_B])$.
	Consequently, for every set $S \subseteq \mathcal{P}$ such that
		$\sigmaHighC[S] \neq \bot$,
		during $E(C'_S,\sigma'_S \circ \sigma_F)$,
		each register $R \in \mathcal{R}$ cannot be accessed
		if the owner of $R$ is in $S_B \setminus F(\sigmaHighB[S_B])$.
		 
	By \autoref{thm:HighBInvars},
		Invariant~\ref{RME-invar:DSMSelfAccessOnly} holds for $\sigmaHighB[0..2^n-1]$
		and $\SMax(\sigmaHighB[0..2^n-1]) = S_B$.
	So in the DSM model, for every set $S \subseteq \mathcal{P}$ such that
		$\sigmaHighB[S] \neq \bot$,
		during $E(\sigmaHighB[S])$,
		each register $R \in \mathcal{R}$ can only be accessed by its owner
		if the owner of $R$ is in $S_B \setminus F(\sigmaHighB[S_B])$.
	Thus for every set $S \subseteq \mathcal{P}$ such that
		$\sigmaHighC[S] \neq \bot$,
		during $E(\sigmaHighC[S]) = E(\sigmaHighB[S] \circ \sigma'_S \circ \sigma_F)$,
		each register $R \in \mathcal{R}$ can only be accessed by its owner
		if the owner of $R$ is in $S_B \setminus F(\sigmaHighB[S_B])$.
	Then, since $F(\sigmaHighC[S_B]) \supseteq F(\sigmaHighB[S_B])$,
		$S_B \setminus F(\sigmaHighC[S_B]) \subseteq 
		S_B \setminus F(\sigmaHighB[S_B])$.
	Consequently, in the DSM model,
		for every set $S \subseteq \mathcal{P}$ such that $\sigmaHighC[S] \neq \bot$,
		during $E(\sigmaHighC[S]) = E(\sigmaHighB[S] \circ \sigma'_S \circ \sigma_F)$,
		each register $R \in \mathcal{R}$ can only be accessed by its owner
		if the owner of $R$ is in $S_B \setminus F(\sigmaHighC[S_B])$
		(Invariant~\ref{RME-invar:DSMSelfAccessOnly}).
	
	By \autoref{thm:HighPreInvars} (\ref{HPI:ValidCCSame}),
		in the CC model,for every set $S \subseteq \mathcal{P}$
		such that $F(\sigmaHighB[S_B]) \cup S_\alpha \subseteq S \subseteq S_B$,
		and every process $p \in S \setminus F(\sigmaHighB[S_B])$,
		the set of registers that $p$ has valid cache copies of
		at the end of $E(\sigmaHighB[S] \circ \sigma'_S)$
		is exactly the same as at the end of 
		$E(\sigmaHighB[S_B] \circ \sigma'_{S_B})$.
	By \autoref{thm:sigmaFApplicable},
		for every register $R \in \mathcal{R}$,
		if a non-read operation is performed on $R$
		during $E(C''_S,\sigma_F)$,
		it is also performed on $R$
		during $E(C''_{S_B},\sigma_F)$.
	By definition, $\sigma_F$ contains only steps of processes in $S_\alpha$.
	So for every set $S \subseteq \mathcal{P}$
		such that $F(\sigmaHighB[S_B]) \cup S_\alpha \subseteq S \subseteq S_B$,
		and every process $p \in S \setminus (F(\sigmaHighB[S_B]) \cup S_\alpha)$,
		the set of registers that $p$ has valid cache copies of
		at the end of $E(\sigmaHighC[S]) = E(\sigmaHighB[S] \circ \sigma'_S \circ \sigma_F)$
		is exactly the same as at the end of 
		$E(\sigmaHighC[S_B]) = E(\sigmaHighB[S_B] \circ \sigma'_{S_B} \circ \sigma_F)$.
	
	Now recall that we have already proven that
		$F(\sigmaHighC[S_B]) = F(\sigmaHighB[S_B]) \cup S_\alpha$,
		and that Invariant~\ref{RME-invar:uniqueSMax}
		holds for $\sigmaHighC[0..2^n-1]$ with 
		$\SMax(\sigmaHighC[0..2^n-1]) = S_B$.
	Thus for every set $S \subseteq \mathcal{P}$
		such that $\sigmaHighC[S] \neq \bot$,
		and every process $p \in S \cap (S_B \setminus F(\sigmaHighC[S_B]))$,
		the set of registers that $p$ has valid cache copies of
		at the end of $E(\sigmaHighC[S])$
		is exactly the same as at the end of $E(\sigmaHighC[S_B])$
		(Invariant~\ref{RME-invar:ValidCCSame}).
	
	Finally, by \autoref{thm:HighBInvars}, 
		$\sigmaHighB[0..2^n-1]$ is $(i-1)$-compliant,
		so by Invariant~\ref{RME-invar:iRMRs},
		for every set $S \subseteq \mathcal{P}$
		and every process $p \in S \setminus F(\sigmaHighB[S])$,
		if $\sigmaHighB[S] \neq \bot$, then 
		$p$ incurs at least $i-1$ RMRs
		during $E(\sigmaHighB[S])$.
	By \autoref{thm:HighPreInvars} (\ref{HPI:OneRMREach}),
		for every set $S \subseteq \mathcal{P}$ such that 
		$F(\sigmaHighB[S_B]) \cup S_\alpha \subseteq S \subseteq S_B$, 
		each process in $S_\beta \cap S$ incurs 
		exactly one RMR during $E(C'_S,\sigma'_S)$.
	By construction, for every set $S \subseteq \mathcal{P}$,
		if $\sigmaHighC[S] \neq \bot$, then
		$\sigmaHighC[S] = \sigmaHighB[S] \circ \sigma'_S \circ \sigma_F$,
		i.e., every process in $S_\beta \cap S = S \setminus F(\sigmaHighB[S_B])$ 
		incurs at least one more RMR during $E(\sigmaHighC[S])$
		than during $E(\sigmaHighB[S])$.
	Thus for every set $S \subseteq \mathcal{P}$
		and every process $p \in S \setminus F(\sigmaHighB[S_B])$
		if $\sigmaHighC[S] \neq \bot$,
		then $p$ incurs at least $i$ RMRs during $E(\sigmaHighC[S])$.
		 
	Now recall that we have already proven that
		$F(\sigmaHighC[S_B]) = F(\sigmaHighB[S_B]) \cup S_\alpha$,
		and that Invariant~\ref{RME-invar:FSame} holds
		for $\sigmaHighC[0..2^n-1]$.
	So for every set $S \subseteq \mathcal{P}$,
		if $\sigmaHighC[S] \neq \bot$, then 
		$F(\sigmaHighC[S]) = F(\sigmaHighB[S_B]) \cup S_\alpha$.
	Thus $(S \setminus F(\sigmaHighB[S_B]) ) \supseteq (S \setminus F(\sigmaHighC[S]))$.
	Therefore for every set $S \subseteq \mathcal{P}$
		and every process $p \in S \setminus F(\sigmaHighC[S])$
		if $\sigmaHighC[S] \neq \bot$,
		then $p$ incurs at least $i$ RMRs during $E(\sigmaHighC[S])$
		(\ref{RME-invar:iRMRs}).
\end{proof}

Finally, we terminate this $i$-th iteration by setting 
	$\sigmaFinal[i,0..2^n-1] = \sigmaHighC[0..2^n-1]$.

\newcommand{\SiMax}[1]{S^{#1}_\textit{max}}
\newcommand{\ActiveI}[1]{n_{#1}}

\paragraph{Analysis:}

For every non-negative integer $i$,
	if $\sigmaFinal[i,0..2^n-1]$ is $i$-compliant,
	then let $\SiMax{i} = \SMax(\sigmaFinal[i,0..2^n-1])$,
	and let $\ActiveI{i} = |\SiMax{i} \setminus F(\sigmaFinal[i,\SiMax{i}])|$.

\begin{lem}
	\label{thm:FinalCompliance}
	For every non-negative integer $i$,
		if $\sigmaFinal[i,0..2^n-1]$ has non-$\bot$ entries,
		then $\sigmaFinal[i,0..2^n-1]$ is $i$-compliant.
\end{lem} 

\toOmit{ 
	\begin{proof}
		If $i = 0$, then
			 every entry of $\sigmaFinal[0,0..2^n-1]$ is the empty schedule.
		Clearly, the array $\sigmaFinal[0,0..2^n-1]$ is $0$-compliant.

		So suppose $i > 0$.
		Thus if $\sigmaFinal[i,0..2^n-1]$ has non-$\bot$ entries,
			then either $\sigmaFinal[i,0..2^n-1] = \sigmaLowC[0..2^n-1]$
			or $\sigmaFinal[i,0..2^n-1] = \sigmaHighC[0..2^n-1]$.
		The lemma immediately follows from \autoref{thm:LowCInvars}
			and \autoref{thm:HighCInvars}.
	\end{proof}
}

\begin{lem}
	\label{thm:ShrinkSizeLowerBound}
	For every positive integer $i$,
		if $\sigmaFinal[i,0..2^n-1]$ is $i$-compliant,
		then $\ActiveI{i} \geq \ActiveI{i-1} / (640 \log^{d+1} n) - 2$
\end{lem} 

\toOmit{
	\begin{proof}
		Since $\sigmaFinal[i,0..2^n-1]$ is $i$-compliant,
			either $\sigmaFinal[i,0..2^n-1] = \sigmaLowC[0..2^n-1]$
			or $\sigmaFinal[i,0..2^n-1] = \sigmaHighC[0..2^n-1]$.

		\begin{description}
			\item[\textbf{Case 1.}] $\sigmaFinal[i,0..2^n-1] = \sigmaLowC[0..2^n-1]$.

				Then by \autoref{thm:LowCInvars},
					$\SMax(\sigmaLowC[0..2^n-1]) = \SLMax$.
				Thus $\ActiveI{i} = |\SLMax \setminus F(\sigmaLowC[\SLMax])|$.

				Now recall that in the construction of $\sigmaLowC[0..2^n-1]$,
					we checked whether there exists a process
					$p \in S_I \setminus F(\sigmaLowB[S_I])$ such that
					$p$ is within the critical section at the end of $E(\sigmaLowB[S_I])$.
				If such a process $p$ exists,
					then we set $\sigmaLowC[0..2^n-1]$ to be a simple modification
					of $\sigmaLowB[0..2^n-1]$ where every set $S \subseteq \mathcal{P}$
					that contains $p$ has had $\sigmaLowC[S]$ set to $\bot$;
					otherwise we set $\sigmaLowC[0..2^n-1] = \sigmaLowB[0..2^n-1]$.
				By \autoref{thm:LowBInvars},
					$\SMax(\sigmaLowB[0..2^n-1]) = S_I$.
				So by the construction of $\sigmaLowC[0..2^n-1]$,
					$|\SLMax \setminus F(\sigmaLowC[\SLMax])| \geq
					|S_I \setminus F(\sigmaLowB[S_I])| - 1$.
				Thus $\ActiveI{i} \geq |S_I \setminus F(\sigmaLowB[S_I])| - 1$.

				Then recall that by the construction of $\sigmaLowB[0..2^n-1]$,
					for every set $S \subseteq \mathcal{P}$,
					$\sigmaLowB[S] = \bot$ if and only if $\sigmaLowA[S] = \bot$.
				By \autoref{thm:LowBInvars} and \autoref{thm:LowAInvars},
					Invariant~\ref{RME-invar:uniqueSMax} holds for both
					$\sigmaLowB[0..2^n-1]$ and $\sigmaLowA[0..2^n-1]$,
					with $\SMax(\sigmaLowB[0..2^n-1]) = \SMax(\sigmaLowA[0..2^n-1]) = S_I$.
				So $|S_I \setminus F(\sigmaLowB[S_I])| = |S_I \setminus F(\sigmaLowA[S_I])|$.
				Thus $\ActiveI{i} \geq |S_I \setminus F(\sigmaLowA[S_I])| - 1$.

				Next, recall that by the construction of $\sigmaLowA[0..2^n-1]$,
					for every set $S \subseteq \mathcal{P}$,
					if $S \not\subseteq S_I$, then $\sigmaLowA[S] = \bot$;
					otherwise $\sigmaLowA[S] = \sigmaSetupB[S]$.
				By definition, $S_I = F(\sigmaSetupB[\SSMax]) \cup I$.

				Furthermore, by \autoref{thm:SetupBInvars},
					$\sigmaSetupB[0..2^n-1]$ is $(i-1)$-compliant with
					$\SMax(\sigmaSetupB[0..2^n-]) = \SSMax$.
				So by Invariant~\ref{RME-invar:FSame},
					$F(\sigmaSetupB[\SSMax]) = F(\sigmaLowA[S_I])$.
				By construction, $I \cap F(\sigmaSetupB[\SSMax]) = \varnothing$.
				Thus $\ActiveI{i} \geq |S_I \setminus F(\sigmaLowA[S_I])| - 1 = |I| - 1$.

				By \autoref{thm:Isize}, $|I| \geq |L| / (7 k \log n)$.
				By construction, $|L| \geq 0.5|\SSMax \setminus F(\sigmaSetupB[\SSMax])|$.
				Thus $\ActiveI{i} \geq |L| / (7 k \log n) - 1
					\geq |\SSMax \setminus F(\sigmaSetupB[\SSMax])| / (14 k \log n) - 1$.

			\item[\textbf{Case 2.}] $\sigmaFinal[i,0..2^n-1] = \sigmaHighC[0..2^n-1]$.

				Then by \autoref{thm:HighCInvars},
					$\SMax(\sigmaHighC[0..2^n-1]) = S_B$.
				Thus $\ActiveI{i} = |S_B \setminus F(\sigmaHighC[S_B])|$.

				Recall that by definition, $S_B = S_\beta \cup F(\sigmaHighA[S_H])$.
				Furthermore, by \autoref{thm:HighCInvars},
					$F(\sigmaHighC[S_B]) = F(\sigmaHighB[S_B]) \cup S_\alpha$.

				Also recall that by the construction of $\sigmaHighB[0..2^n-1]$,
					for every set $S \subseteq \mathcal{P}$,
					if $S \not\subseteq S_B$,
					then $\sigmaHighB[S] = \bot$;
					otherwise $\sigmaHighB[S] = \sigmaHighA[S]$.
				Thus by \autoref{thm:HighAInvars} and Invariant~\ref{RME-invar:FSame},
					$F(\sigmaHighA[S_H]) = F(\sigmaHighA[S_B]) = F(\sigmaHighB[S_B])$.

				Finally, recall that by construction,
					$S_\alpha \cap F(\sigmaHighA[S_H]) = \varnothing$ and
					$S_\beta \cap F(\sigmaHighA[S_H]) = \varnothing$.
				Therefore:
				\begin{align*}
					S_B \setminus F(\sigmaHighC[S_B])
						& = S_B \setminus (F(\sigmaHighB[S_B]) \cup S_\alpha) \\
						& = S_B \setminus (F(\sigmaHighA[S_H]) \cup S_\alpha) \\
						& = (S_\beta \cup F(\sigmaHighA[S_H])) \setminus (F(\sigmaHighA[S_H]) \cup S_\alpha) \\
						& = S_\beta \setminus S_\alpha
				\end{align*}

				By \autoref{thm:BetaSize},
					$|S_\beta \setminus S_\alpha| > \frac{|H|}{204.8 k}$.
				Thus $\ActiveI{i} = |S_B \setminus F(\sigmaHighC[S_B])| > \frac{|H|}{204.8 k}$.

				By construction, $|H| \geq 0.5|\SSMax \setminus F(\sigmaSetupB[\SSMax])|$.
				Thus $\ActiveI{i} > \frac{|H|}{204.8 k}
					\geq |\SSMax \setminus F(\sigmaSetupB[\SSMax])| / (409.6 k)$.
		\end{description}
		So in both cases,
			$\ActiveI{i} \geq |\SSMax \setminus F(\sigmaSetupB[\SSMax])| / (409.6 k \log n) - 1$.

		Now recall that in the construction of $\sigmaSetupB[0..2^n-1]$,
			we checked whether there exists a process
			$p \in \SOMax \setminus F(\sigmaSetupA[\SOMax])$ such that
			$p$ is within the critical section at the end of $E(\sigmaSetupA[\SOMax])$.
		If such a process $p$ exists,
			then we set $\sigmaSetupB[0..2^n-1]$ to be a simple modification
			of $\sigmaSetupA[0..2^n-1]$ where every set $S \subseteq \mathcal{P}$
			that contains $p$ has had $\sigmaSetupB[S]$ set to $\bot$;
			otherwise we set $\sigmaSetupB[0..2^n-1] = \sigmaSetupA[0..2^n-1]$.
		By \autoref{thm:SetupAInvars},
			$\SMax(\sigmaSetupA[0..2^n-1]) = \SOMax$.
		So by the construction of $\sigmaSetupB[0..2^n-1]$,
			$|\SSMax \setminus F(\sigmaSetupB[\SSMax])| \geq
			|\SOMax \setminus F(\sigmaSetupA[\SOMax])| - 1$.
		Thus $\ActiveI{i} \geq |\SSMax \setminus F(\sigmaSetupB[\SSMax])| / (409.6 k \log n) - 1
			\geq |\SOMax \setminus F(\sigmaSetupA[\SOMax])| / (409.6 k \log n) - 2$

		Recall that by the construction of $\sigmaSetupA[0..2^n-1]$,
			for every set $S \subseteq \mathcal{P}$,
			$\sigmaSetupA[S] = \bot$ if and only if $\sigmaOld[S] = \bot$.
		By \autoref{thm:SetupAInvars}, Invariant~\ref{RME-invar:uniqueSMax} holds
			for $\sigmaSetupA[0..2^n-1]$ and $\SMax(\sigmaSetupA[0..2^n-1]) = \SOMax$.
		Also recall that by definition,
			$\sigmaOld[0..2^n-1]$ is $(i-1)$-compliant
			and $\SMax(\sigmaOld[0..2^n-1]) = \SOMax$.
		So $|\SOMax \setminus F(\sigmaSetupA[\SOMax])| = |\SOMax \setminus F(\sigmaOld[\SOMax])|$.
		Thus $\ActiveI{i} \geq |\SOMax \setminus F(\sigmaSetupA[\SOMax])| / (409.6 k \log n) - 2
			\geq |\SOMax \setminus F(\sigmaOld[\SOMax])| / (409.6 k \log n) - 2$

		Finally, recall that $\sigmaOld[0..2^n-1]$ is simply $\sigmaFinal[i-1, 0..2^n-1]$.
		So $|\SOMax \setminus F(\sigmaOld[\SOMax])| = \ActiveI{i-1}$.
		Consequently, $\ActiveI{i} \geq \ActiveI{i-1} / (409.6 k \log n) - 2$.
		Then, since $k = \log^d n$, $\ActiveI{i} \geq \ActiveI{i-1} / (409.6 \log^{d+1} n) - 2$.
	\end{proof}
}

Since $\SMax(\sigmaFinal[0,0..2^n-1]) = \mathcal{P}$
	and $\sigmaFinal[0,\mathcal{P}]$ is the empty schedule,
	$F(\sigmaFinal[0,\mathcal{P}]) = \varnothing$ and
	$\ActiveI{0} = n$.
By \autoref{thm:ShrinkSizeLowerBound},
	for every positive integer $i$,
	if $\sigmaFinal[i,0..2^n-1]$ is $i$-compliant,
	then $\ActiveI{i} \geq \ActiveI{i-1} / O(\log^{d+1} n)$.

Consequently, if $\mathcal{I}$ is the largest positive integer such that
	$\sigmaFinal[\mathcal{I},0..2^n-1]$ is $\mathcal{I}$-compliant,
	then $\mathcal{I}$ is $\Omega(\log n / \log \log n)$.
So $\sigmaFinal[\mathcal{I},\SiMax{\mathcal{I}}]$ contains a schedule such that:
\begin{itemize}
	\item Since we reach the $\mathcal{I}$-th iteration,
		$\sigmaFinal[\mathcal{I}-1,0..2^n-1]$ has at least $2^{(k^3)}$ non-$\bot$ entries, i.e.,
		$\ActiveI{\mathcal{I} - 1} \geq k^3 = \log^{3d} n$.
	So $\ActiveI{\mathcal{I}} \geq \log^{3d} n / (640 \log^{d+1} n) - 2$,
		which for a sufficiently large constant $d$,
		$\ActiveI{\mathcal{I}} \geq \log^d n$.
	Thus $|\SiMax{\mathcal{I}} \setminus
		F(\sigmaFinal[\mathcal{I},\SiMax{\mathcal{I}}])| \geq \log^d n$

	\item For every process $p \in \SiMax{\mathcal{I}}
		\setminus F(\sigmaFinal[\mathcal{I},\SiMax{\mathcal{I}}])$,
		$p$ incurs at least $\mathcal{I}$ RMRs during
		$E(\sigmaFinal[\mathcal{I},\SiMax{\mathcal{I}}])$
		(Invariant~\ref{RME-invar:iRMRs}).

	\item For every process $p \in \SiMax{\mathcal{I}}
		\setminus F(\sigmaFinal[\mathcal{I},\SiMax{\mathcal{I}}])$,
		$p$ never crashes during $E(\sigmaFinal[\mathcal{I},\SiMax{\mathcal{I}}])$
		(Invariant~\ref{RME-invar:crashes}).

	\item For every process $p \in \SiMax{\mathcal{I}}
		\setminus F(\sigmaFinal[\mathcal{I},\SiMax{\mathcal{I}}])$,
		$p$ never enters the critical section during $E(\sigmaFinal[\mathcal{I},\SiMax{\mathcal{I}}])$
		(Invariant~\ref{RME-invar:CS}).
\end{itemize}

Thus we have proven \autoref{thm:MainTheorem}.

  \section{Conclusion}
  We proved a tight RMR lower bound for RME, which applies to almost all standard shared memory primitives that have been used to solve the problem.
  The lower bound separates the RMR complexity of mutual exclusion in the traditional, non-recoverable model from the recoverable model, for systems that provide fetch-and-store and fetch-and-increment objects in addition to registers and compare-and-swap objects.
  It applies to objects of arbitrary (even unbounded) size.

  RME can be solved in constant RMRs with fetch-and-add primitives of size $n^{\Omega(1)}$ bits \cite{DM2020a}, so obviously our lower bound cannot be extended to cover such primitives.
  But it remains an open problem, whether fetch-and-add operations can help, under the standard assumption that objects can store only $O(\log n)$-bits.
  We believe that this is not the case.
  In fact, we conjecture that in general objects that can only store $O(\log n)$ bits of information are not sufficient to break through the $\Omega(\log n/\log\log n)$ RMR complexity barrier. 
	
\begin{acks}
Support is gratefully acknowledged from the Natural Science and Engineering Research Council of
Canada (NSERC) under Discovery Grant RGPIN/2019-04852, and the Canada Research Chairs program.
\end{acks}

\bibliographystyle{ACM-Reference-Format}
 \bibliography{pwliterature}

%
 %
\end{document}